\documentclass[10pt,a4paper]{amsart}
\usepackage[utf8]{inputenc}
\usepackage{amssymb, amsmath, amscd, amsthm}
\usepackage{dsfont,longtable,graphicx}

\usepackage{color}
\usepackage[font=small]{caption}
\usepackage{subcaption}
\usepackage[]{hyperref}

\theoremstyle{definition}
\newtheorem{definition}{Definition}
\newtheorem{proposition}[definition]{Proposition}
\newtheorem{theorem}[definition]{Theorem}
\newtheorem{lemma}[definition]{Lemma}

\numberwithin{equation}{section}
\numberwithin{definition}{section}
\numberwithin{figure}{section}

\newcommand{\GL}{ {\widetilde{GL}^+}}
\newcommand{\id}{\text{id}}
\newcommand{\one}{\mathbf{1}}
\newcommand{\Rep}{\text{Rep}\,}

\newcommand{\Cb}{\mathbb{C}}

\allowdisplaybreaks

\begin{document}

\title[Spin from defects in 2d QFT]{Spin from defects in two-dimensional\\quantum field theory}

\begin{abstract}
We build two-dimensional quantum field theories on spin surfaces starting from theories on oriented surfaces with networks of topological defect lines and junctions. The construction uses a combinatorial description of the spin structure in terms of a triangulation equipped with extra data. The amplitude for the spin surfaces is defined to be the amplitude for the underlying oriented surface together with a defect network dual to the triangulation. Independence of the triangulation and of the other choices follows if the line defect and junctions are obtained from a $\Delta$-separable Frobenius algebra with involutive Nakayama automorphism in the monoidal category of topological defects. For rational conformal field theory we can give a more explicit description of the defect category, and we work out two examples related to free fermions in detail: the Ising model and the $so(n)$ WZW model at level $1$.
\end{abstract}

\date{}

\author{Sebastian Novak}
\address{Sebastian Novak\\Fachbereich Mathematik\\Universit\"at Hamburg\\Bundesstra\ss e 55\\20146 Hamburg\\Germany}
\email{sebastian.novak@mailbox.org}

\author{Ingo Runkel}
\address{Ingo Runkel\\Fachbereich Mathematik\\Universit\"at Hamburg\\Bundesstra\ss e 55\\20146 Hamburg\\Germany}
\email{ingo.runkel@uni-hamburg.de}

\maketitle

\tableofcontents

\newpage

\section{Introduction}

Consider a two-dimensional quantum field theory $Q$ which is defined on oriented surfaces with metric, and which allows for the presence of line defects. We will be interested in topological line defects, that is, line defects that can be deformed on the surface without changing the amplitude which $Q$ assigns to the surface. Line defects can meet in junction points, and also here we will be interested only in topological junctions. Altogether, the QFT $Q$ assigns amplitudes to surfaces with defect networks, and these amplitudes are invariant under deformations which move the defect lines and junctions without generating new intersections. We will briefly review this framework in Section \ref{sec:def-QFT}.

In this paper we show how given a topological line defect of the QFT $Q$, together with 1- and 3-valent topological junctions for this defect -- all subject to conditions we list below -- one can construct a new QFT $Q_\text{spin}$ which is defined on surfaces with metric and with spin structure. 

\medskip

Quite generally, to construct new $n$-dimensional QFTs from a given $n$-dimensional QFT, one can try to implement the following idea:
\begin{quote}
Carry out a state sum construction of an $n$-dimensional topological field theory inside an $n$-dimensional QFT with defects.
\end{quote}
That is, try to model the cell decomposition of the $n$-dimensional manifolds used in the state sum construction with defects of dimensions $0,1,\dots,n{-}1$ in the QFT with defects, and translate the invariance conditions of the state sum construction into requirements on these defects.
This idea generalises the state sum construction of topological field theories in the sense that the latter can be thought of as being internal to the trivial QFT.

So far the above idea has been applied only in two dimensions, and we will describe this case in more detail now.

\medskip

The topological line defects and topological junctions of an oriented two-di\-men\-sional QFT $Q$ form a $\mathbb{C}$-linear monoidal category $\mathcal{D}_Q$ \cite{davydov2011field}. The objects in $\mathcal{D}_Q$ are the possible line defects, the tensor product amounts to fusion of line defects, and the morphisms are given by junctions. Since junctions can be thought of as fields inserted at the point were the defect lines meet, they form a $\mathbb{C}$-vector space. In addition, $\mathcal{D}_Q$ has two-sided duals (given by orientation reversal of the defect line) and is pivotal. 

The category $\mathcal{D}_Q$ contains much interesting information about the QFT $Q$. For example, given a $\Delta$-separable symmetric Frobenius algebra $A$ in $\mathcal{D}_Q$ (we will explain these notions in the table below), 
one can construct a new QFT $Q_\text{orb}$, 
the {\em orbifold of $Q$ by $A$}, which is again defined on oriented
surfaces with metric \cite{frohlich2009defect,carqueville2012orbifold}. 
This amounts to applying the above idea to the state sum construction of oriented two-dimensional topological field theories \cite{bachas1993topological,fukuma1994lattice}.
The name ``orbifold'' is justified since the construction of $Q_\text{orb}$
reduces to the standard orbifold construction if $A$ comes from a group symmetry of $Q$. But there are $A$'s which do not come from group symmetries (e.g.\ the exceptional cases in \cite{tft1,Carqueville:2013mpa}), so in this sense
$Q_\text{orb}$ could be called a generalised orbifold.

Another example is given by the main result of this paper: a $\Delta$-separable Frobenius algebra $A \in \mathcal{D}_Q$ whose Nakayama automorphism (see table below) squares to the identity allows one to define a QFT $Q_\text{spin}$ on spin surfaces starting from the QFT $Q$, which was defined on oriented surfaces with defects. 
This is again an instance of the above idea, applied now to the state sum construction of 2d spin TFTs given in \cite{Novak:2014oca}. Let us describe the resulting construction in more detail.

To evaluate $Q_\text{spin}$ on a spin surface $\Sigma$, one first encodes the spin structure in terms of 
\begin{itemize}
\item[-]
a triangulation of the underlying oriented surface $\underline\Sigma$,
\item[-]
a choice of an orientation for each edge of the triangulation and of a preferred edge for each triangle (a {\em marking}),
\item[-]
a sign $\{ \pm1\}$ for each edge (the {\em edge signs}).
\end{itemize}
This combinatorial model for spin structures was introduced in \cite{Novak:2014oca} and will be reviewed in Section \ref{sec:spin-comb}. 
Other closely related combinatorial models can be found in \cite{kuperberg1998exploration,cimasoni2007dimers,budney2013combinatorial}. An extension of the work in \cite{Novak:2014oca} to $r$-spin surfaces can be found in \cite{Novak-PhD}.

Next, one places a network of line defects $A$ and junctions labelled by the structure maps of the Frobenius algebra on the graph dual to the triangulation. The precise type of the junctions depends on the marking and the edge signs. The amplitude $Q_\text{spin}(\Sigma)$ is defined as the amplitude that $Q$ assigns to $\underline\Sigma$ with the above defect network. This is described in detail in Section \ref{sec:amp-spin-def}.

Frobenius algebras as above also appear in another state sum construction of two-dimensional spin topological field theory in \cite{barrett2013spin} and in the description of  ``generalised twisted sectors'' in orbifolds in \cite{Brunner:2013ota} (there without restriction on the order of the Nakayama automorphism).

\medskip

The following table gives the definitions of the various algebraic notions above, as well as their meaning in the context of topological defects. More details can be found in Section \ref{sec:amp-spin-def}. The identities of topological defect networks in the table are to be understood as follows. Given two surfaces with metric, $\Sigma$ and $\Sigma'$, equipped with defect networks such that the defect networks only differ in a small patch as shown in the table, the amplitudes of $\Sigma$ and $\Sigma'$ agree, i.e.\ $Q(\Sigma) = Q(\Sigma')$.

{\small
\renewcommand{\arraystretch}{1.5}
\begin{longtable}{p{5.5em}|p{12em}|p{17em}}
concept & algebraic description & description in QFT with defects\\
\hline
$A$ & object in the pivotal monoidal category $\mathcal{D}_Q$ & topological defect of $Q$ \\
algebra 
&
morphisms $\eta : \mathbf{1} \to A$ and $\mu : A \otimes A \to A$ such that $\eta$ is a unit for $\mu$ and $\mu$ is associative.
&
a 1-valent and a 3-valent topological defect junction,
\begin{center}
$\raisebox{-0.5\height}{\includegraphics[scale=0.4]{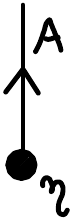}}$
\hspace{1em},\hspace{2em}
$\raisebox{-0.5\height}{\includegraphics[scale=0.4]{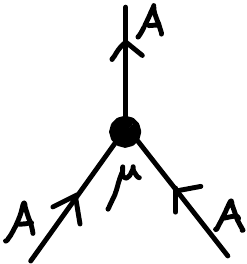}}$
\hspace{1em},
\end{center}
such that we have the identities
\begin{center}
$\raisebox{-0.5\height}{\includegraphics[scale=0.4]{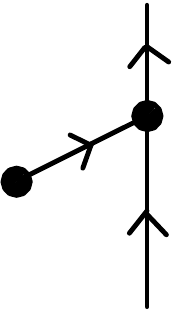}}
~=~
\raisebox{-0.5\height}{\includegraphics[scale=0.4]{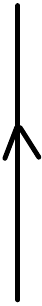}}
~=~
\raisebox{-0.5\height}{\includegraphics[scale=0.4]{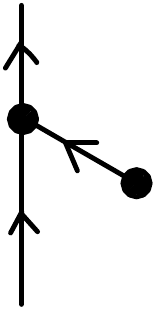}}$
\hspace{1em},

\medskip

$\raisebox{-0.5\height}{\includegraphics[scale=0.4]{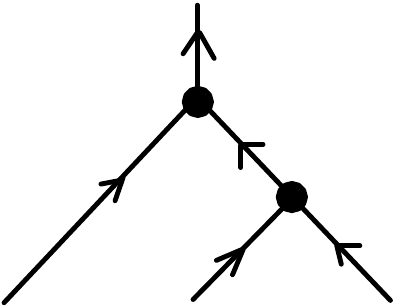}}
=\raisebox{-0.5\height}{\includegraphics[scale=0.4]{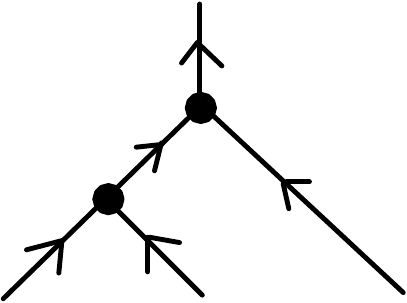}}$
\hspace{1em}
\end{center}
\\
Frobenius \newline algebra 
&
An algebra and a coalgebra (with counit $\varepsilon : A \to \mathbf{1}$ and coproduct $\Delta : A \to A \otimes A$) such that the Frobenius condition holds: 

$(\id \otimes \mu) \circ (\Delta \otimes \id)= \Delta \circ \mu$ 

$= (\mu \otimes \id) \circ (\id \otimes \Delta)$.
&
Additional 1- and 3-valent defect junctions,
\begin{center}
$\raisebox{-0.5\height}{\includegraphics[scale=0.4]{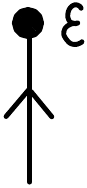}}$
\hspace{1em},\hspace{2em}
$\raisebox{-0.5\height}{\includegraphics[scale=0.4]{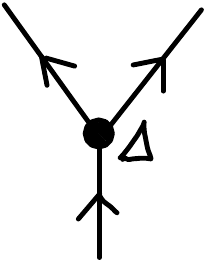}}$
\hspace{1em},
\end{center}
such that $\varepsilon$ and $\Delta$ satisfy the same identities as $\eta$ and $\mu$ but with reversed orientations on the defect lines, as well as
\begin{center}
$\raisebox{-0.5\height}{\includegraphics[scale=0.4]{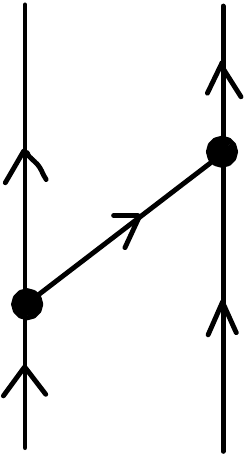}}
\quad=~
\raisebox{-0.5\height}{\includegraphics[scale=0.4]{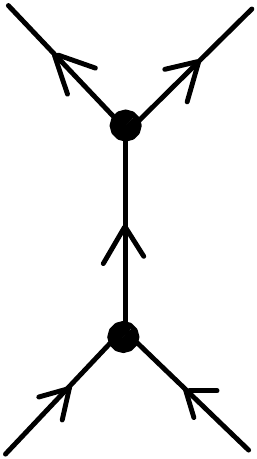}}
~=\quad
\raisebox{-0.5\height}{\includegraphics[scale=0.4]{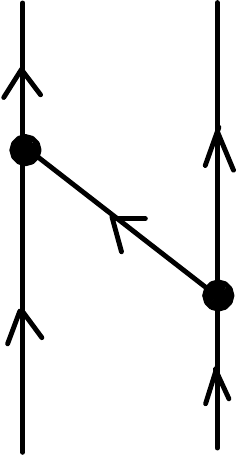}}$.
\end{center}
\\
$\Delta$-separable 
&
The product and coproduct of the Frobenius algebra satisfy $\mu \circ \Delta = \id$.
&
The defect junctions $\Delta$ and $\mu$ satisfy
\begin{center}
$\raisebox{-0.5\height}{\includegraphics[scale=0.4]{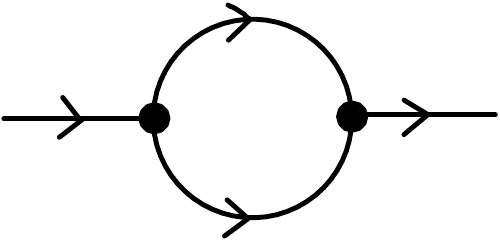}}
\hspace{1em}=\hspace{1em}
\raisebox{-0.5\height}{\includegraphics[scale=0.4]{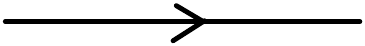}}$.
\end{center}
\\
symmetric 
&
The pairing $b := \varepsilon \circ \mu$ satisfies
$(\id \otimes b) \circ (\widetilde{\mathrm{coev}}_A \otimes \id)
=
(b \otimes \id) \circ (\id \otimes \mathrm{coev}_A)$.
&
The defect junctions $\mu$ and $\varepsilon$ satisfy
\begin{center}
$\raisebox{0pt}{\includegraphics[scale=0.4]{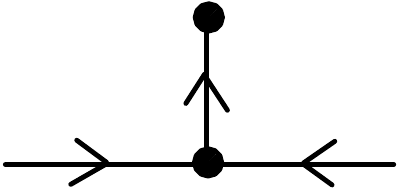}}
\hspace{1em}=\hspace{1em}
\raisebox{-.8\height}{\includegraphics[scale=0.4]{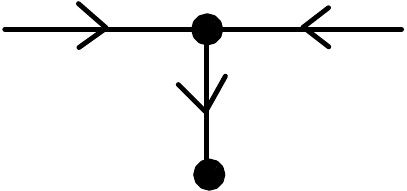}}$.
\end{center}
\\
Nakayama \newline autom.\ $N$
&
Let $b$ be as above and set $c = \Delta \circ \eta$. Then
$N : A \to A$ is given by
$N = (\id \otimes \widetilde{\mathrm{ev}}_A) \circ
(c \otimes \id) \circ (b \otimes \id) \circ
(\id \otimes \mathrm{coev}_A)$. $N$ is an automorphism
of the Frobenius algebra $A$.
&
$N$ abbreviates the following combination of defect lines and junctions
\begin{center}
$\raisebox{0pt}{\includegraphics[scale=0.4]{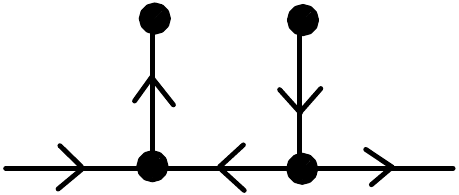}}$.
\end{center}
\\
$N$ squares to
\newline
the identity
&
$N^2 = \id$.
&
\begin{center}
$\raisebox{-.1\height}{\includegraphics[scale=0.4]{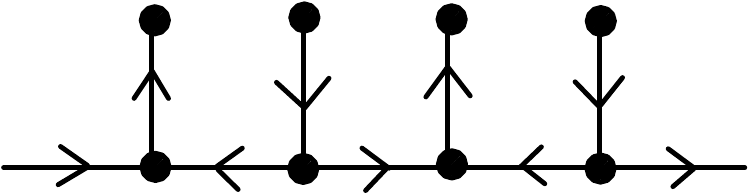}}
~=~
\raisebox{-.2\height}{\includegraphics[scale=0.4]{cpic31e2.pdf}}$.
\end{center}
\end{longtable}
}

We remark that a Frobenius algebra is symmetric if and only if its Nakayama automorphism is the identity.

\medskip

The description of QFTs $Q_\text{spin}$ on spin surfaces in terms of QFTs $Q$ on oriented surfaces with defect networks is useful if one has good control over the defect category $\mathcal{D}_Q$. 
One class of theories where this is the case are rational two-dimensional conformal field theories. In fact, the study of rational CFTs on spin surfaces is one of the main motivations for us to develop the present formalism. Let us have a closer look at this case, which we also discuss in detail in Section \ref{sec:RCFT}.

Fix a vertex operator algebra $\mathcal{V}$ with the property that its category of representations, $\mathrm{Rep}\mathcal{V}$, is a modular tensor category. We will call such vertex operator algebras {\em rational}. Rational CFTs are then those with left/right symmetry given by a rational $\mathcal{V}$. The simplest rational CFT with a given symmetry $\mathcal{V}$ is called the charge-conjugate theory, or the ``Cardy case''. In the Cardy case CFT for $\mathcal{V}$, the category  of topological defects, which are in addition transparent to the left/right symmetry $\mathcal{V}$, is given by $\mathrm{Rep}\mathcal{V}$ itself. 

In the free fermion examples which we investigate in Sections \ref{sec:1ff-ex} and \ref{sec:so(n)}, it turns out that instead of the CFT on its own, one should consider the product of the CFT with the trivial 2d\,TFT that takes values in super vector spaces (see Section \ref{sec:QxSV}). The relevant category of topological defects is then $\mathrm{Rep}\mathcal{V} \boxtimes \mathbf{SVect}$, and it is this category in which we need to find $\Delta$-separable Frobenius algebras whose Nakayama automorphism squares to the identity.

\medskip

There are many questions still unanswered in the present paper. For example, in the context of rational CFTs: Can every 2d spin CFT be obtained from a suitable oriented CFT in this way? Can one explicitly classify all $\Delta$-separable Frobenius algebras $A$ with $N^2=\id$ (or suitable Morita classes thereof) in $\mathrm{Rep}\mathcal{V} \boxtimes \mathbf{SVect}$ in interesting examples? For supersymmetric CFTs, what characterises those $A$ which give world sheet and/or spacetime supersymmetric theories? 

Finally, it should be straightforward to extend the constructions in this paper to spin surfaces with boundaries and defects, as well as to $r$-spin surfaces (following \cite{Novak-PhD}). We hope to return to these points in the future. 

\bigskip
\noindent
{\bf Acknowledgments:} The authors would like to thank Nils Carqueville for helpful discussions and comments on a draft version of this paper. IR thanks the Erwin-Schr\"odinger Institute in Vienna for hospitality during the programme ``Modern trends in topological quantum field theory'' (February and March 2014) where part of this work was completed.
SN was supported by the DFG funded Research Training Group 1670 ``Mathematics inspired by string theory and quantum field theory''. IR is supported in part by the DFG funded Collaborative Research Center 676 ``Particles, Strings, and the Early Universe''.

\section{A combinatorial model for spin surfaces}\label{sec:spin-comb}

In this section we briefly review the combinatorial model for spin surfaces introduced in \cite{Novak:2014oca}. 

\medskip

By a {\em surface} $\underline\Sigma$ we mean an oriented two-dimensional smooth compact real manifold with parametrised boundary.  
The boundary parametrisation consists of, firstly, an ordering $\{1,\dots,B\}$ of the connected components of $\partial\underline\Sigma$, and, secondly, for each $i \in \{1,\dots,B\}$ an orientation preserving smooth embedding $\varphi_i : \underline{U_i} \to \underline\Sigma$. 
Here, $\underline{U_i}$ is a half-open annulus $\{ z \in \Cb | 1 \le |z| < r\}$ for some $r>1$ and $\varphi_i$ maps the unit circle to the $i$'th boundary component.

Denote by $GL_2^+$ the $2{\times}2$-matrices with positive determinant. 
Let $\GL_2$ be the connected double cover of $GL_2^+$.
A {\em spin surface} $\Sigma$ is a surface $\underline\Sigma$ together with a $\GL_2$-principal bundle, which is a double cover of the oriented frame bundle such that the $\GL_2$ action is compatible with the  $GL_2^+$ action on the frame bundle. Furthermore, $\Sigma$ is equipped with a lift $\tilde\varphi_i$ of the boundary parametrisation maps to the spin bundle. 
In more detail, $\underline{U_i}$ allows for two non-isomorphic spin structures: the 
Neveu-Schwarz-type ($NS$-type) spin structure extends to the disc, the Ramond-type ($R$-type) spin structure does not. 
Write $U_i$ for $\underline{U_i}$, equipped with one of these two spin structures.
Then $\tilde\varphi_i : U_i \to \Sigma$ is a map of spin surfaces.
We say the $i$'th boundary component of $\Sigma$ is of $NS$-type ($R$-type) if the spin structure of $U_i$ is of $NS$-type ($R$-type).
More details can be found in \cite[Sect.\,2]{Novak:2014oca}, see in particular \cite[Def.\,2.14]{Novak:2014oca}.

\medskip

Let $\underline\Sigma$ be a surface. 
Fix the standard triangulation of the unit circle to be given by the three arcs between the points $1$, $e^{2 \pi i /3}$ and $e^{-2 \pi i/3}$. 
Via the boundary parametrisations $\varphi_i$, this gives a triangulation of the boundary of  $\underline\Sigma$. Choose an extension of this triangulation from the boundary to the interior. 
 A {\em marking} on this triangulation is an assignment of a preferred edge to each triangle and of an orientation to each edge. 
The orientation of the boundary edges is defined to be that induced by the counter-clockwise orientation of the unit circle.

By a choice of {\em edge signs} we mean a map $s$ from the set of edges of the triangulation to $\{\pm 1\}$. For a given choice of edge signs on a marked triangulation of $\underline\Sigma$ one obtains a spin structure on $\underline\Sigma$ minus the vertices of the triangulation \cite[Sect.\,3.7]{Novak:2014oca}. The spin structure extends to the vertices if and only if the following condition holds at each vertex $v$ 
\cite[Cor.\,3.14\,\&\,Lem.\,3.15]{Novak:2014oca}:
\begin{itemize}
\item {\em $v$ is an inner vertex:} Let $D$ be the number of triangles $\sigma$ containing $v$ such that the preferred edge of $\sigma$ is the first one, counting the edges of $\sigma$ counter-clockwise starting from the vertex $v$. Let $K$ be the number of edges containing $v$ and pointing away from $v$. The edge signs must satisfy
\begin{equation}\label{eq:admiss-inner}
	\prod_{e: v\in e} s(e) = (-1)^{D+K+1} \ ,
\end{equation}
where the product is over all edges containing the vertex $v$.
\item {\em $v$ is a boundary vertex:} Let $D$ and $K$ be as above. For $K$ include the boundary edge pointing away from $v$.
If the boundary component containing $v$ is of $NS$-type,  and if the vertex $v$ is the image of the point 1 under the corresponding parametrisation map  $\varphi_i : U_i \to \underline\Sigma$, set $D'=D+1$. Otherwise set $D'=D$. The edge signs must satisfy
\begin{equation}\label{eq:admiss-bnd}
	\prod_{e: v\in e} s(e) = (-1)^{D'+K+1} \ ,
\end{equation}
where again the product is over all edges containing the vertex $v$, including the two boundary edges adjacent to $v$.
\end{itemize}
If this condition holds at each vertex, we call the edge signs {\em admissible}. By definition, admissible edge signs turn the surface $\underline\Sigma$ into a spin surface $\Sigma$.
This combinatorial model for spin structures is the first main ingredient in this paper.

Given a spin surface, one can ask whether or not a given closed path in the frame bundle lifts to the spin bundle. 
Since we are working in two dimensions, we can turn a smooth closed path $\gamma : S^1 \to \underline\Sigma$ with nowhere vanishing derivative into an -- up to homotopy unique -- path $\hat\gamma$ in the frame bundle.
Namely, at each point of the path pick a second vector, which together with the velocity vector of the path is an oriented frame, and which depends continuously on the parametrisation of the path.
Thus, given $\gamma$ we can ask whether or not $\hat\gamma$ lifts to a closed path in $\Sigma$. We now explain how the lifting property can be read off from the edge signs.

\begin{figure}[tb]
\begin{center}
\raisebox{4em}{a)}
\hspace{.5em}
$\raisebox{-0.5\height}{\includegraphics[scale=0.5]{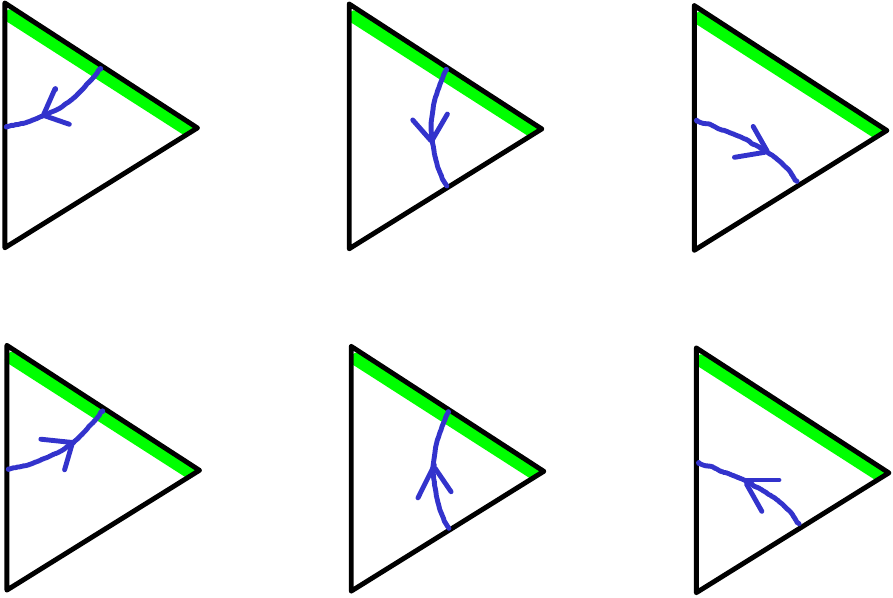}}$
\hspace{.3em}
\raisebox{-.2em}{
\begin{minipage}{4em}
$\Bigg\} +1$
\\[1.8em]
$\Bigg\} -1$
\end{minipage}}
\hspace{.5em}
\raisebox{4em}{b)}
\hspace{.5em}
$\raisebox{-0.5\height}{\includegraphics[scale=0.5]{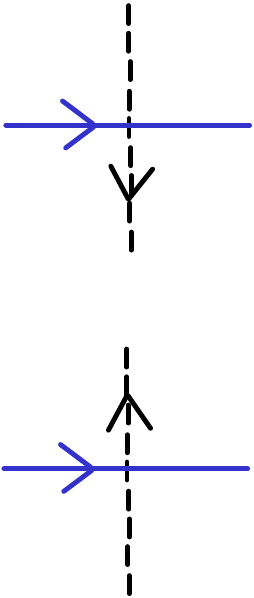}}$
\hspace{.1em}
\raisebox{-.2em}{
\begin{minipage}{4em}
$\Bigg\} +1$
\\[1.8em]
$\Bigg\} -1$
\end{minipage}}
\end{center}
\caption{a) The six configurations a path can transverse a triangle relative to the marked edge of the triangle, together with the sign used to determine the lifting property of the path. b) The two configurations the oriented path (horizontal, solid blue) can cross an oriented edge (vertical, dashed black), again with corresponding signs.}
\label{fig:path-lift-signs}
\end{figure}

Assume that in each triangle of the triangulation, the path $\gamma$ looks as in one of the six configurations shown in Figure \ref{fig:path-lift-signs}\,a. Now multiply together the signs given in Figure  \ref{fig:path-lift-signs}\,a for all triangles transversed by $\gamma$, and the signs given in Figure  \ref{fig:path-lift-signs}\,b for all edges crossed by $\gamma$. Let $S \in \{\pm1 \}$ be the result. Then the corresponding path $\hat\gamma$ in the frame bundle has a  closed lift if and only if
\begin{equation}\label{eq:path-lift-crit}
	\prod_{e} s(e) = S \ ,
\end{equation}
where the product is over all edges crossed by $\gamma$ 
\cite[Lem.\,3.13]{Novak:2014oca}. 

Suppose the path $\gamma$ is the boundary of a small disc containing a vertex $v$ of the triangulation. Then the spin structure extends to the vertex if and only if $\hat\gamma$ does not have a closed lift. One quickly checks that in this way \eqref{eq:path-lift-crit} reproduces the admissibility condition for edge signs at an inner vertex as given in \eqref{eq:admiss-inner}.

\section{Two-dimensional QFT with defects}\label{sec:defect-QFT}\label{sec:def-QFT}

The second ingredient in our construction are surfaces with defects and field theories on such surfaces. We collect some material from \cite{davydov2011field} to fix our conventions.

\subsection{Surfaces with defects, state spaces, amplitudes}\label{sec:def-surf+ampl}

A {\em surface with defects} $\Sigma^d$ is a surface $\underline\Sigma$ together with a compact oriented one-dimensional submanifold $\Delta$, such that $\partial\Delta \subset \partial\underline\Sigma$ and $\Delta$ meets $\partial\underline\Sigma$ transversally.
We also fix a set $D$ of defect conditions. Each connected component of $\Delta$ is labelled by an element of $D$.

So far the underlying surface $\underline\Sigma$ was an oriented smooth manifold. To discuss non-topological quantum field theories, we equip $\underline\Sigma$ with a metric or a conformal structure. For the purpose of this paper, by a {\em quantum field theory on surfaces with defects} we shall mean the following:
\begin{itemize}
\item[{\bf Q1}] For each surface with defects  $\Sigma^d$ and for all boundary components $i = 1,\dots,B$ of $\Sigma^d$, an assignment of a $\Cb$-super vector space $\mathcal{H}_i$ (the {\em state space}). 
\item[{\bf Q2}] 
For each surface with defects $\Sigma^d$ an even linear map $Q(\Sigma^d) : \mathcal{H}_1 \otimes \cdots \otimes \mathcal{H}_B \to \mathbb{C}$ (the {\em amplitude}).
\item[{\bf Q3}]
For a permutation $\sigma$ of $\{1,\dots,B\}$,
let $\Sigma^d_\sigma$ be identical to $\Sigma^d$ except that the ordering of the boundary components of $\Sigma^d$ is changed: The $i$'th boundary of $\Sigma^d_\sigma$ is the $\sigma(i)$'th boundary component of $\Sigma^d$. Then 
\begin{equation}\label{eq:Q-parity}
Q(\Sigma^d)(\psi_1 \otimes \cdots \otimes \psi_n)
= 
(-1)^{P} Q(\Sigma_\sigma^d)(\psi_{\sigma(1)} \otimes \cdots \otimes \psi_{\sigma(n)}) 
\ ,
\end{equation}
where $(-1)^P$ is the parity sign arising from applying the permutation $\sigma$ to the vectors $\psi_1,\dots,\psi_n$. 
(For example, a transposition of two adjacent arguments gives a sign if and only if both are odd.)
\item[{\bf Q4}] 
Compatibility with glueing. Since we will not use the glueing property explicitly in this paper, we will omit its description and refer to \cite{davydov2011field}.\footnote{
In \cite{davydov2011field} the usual functorial formulation of the glueing conditions is employed. When translating this into the present situation with only ``in-going'' boundary components one should add cylinders with two ``out-going'' boundaries and the corresponding copairing on state spaces to the setup.}
\end{itemize}

Consider a surface $C$ with defects which is diffeomorphic to $S^1 \times [0,1]$ such that each connected component of the defect submanifold in $C$ gets mapped to an interval $[0,1]$. That is, up to diffeomorphism $C$ is a cylinder with parallel defect lines connecting the two boundary circles.
We will say $Q$ is {\em non-degenerate} if for all such cylinders $C$ the pairing $Q(C) : \mathcal{H}_1 \otimes \mathcal{H}_2 \to \mathbb{C}$ is non-degenerate in the sense that for all $\psi_i \in \mathcal{H}_1$ there exists a $\psi_2 \in \mathcal{H}_2$ such that $Q(C)(\psi_1 \otimes \psi_2) \neq 0$, and vice versa.\footnote{
	We remark already here that the term ``non-degenerate'' is used in two different ways in this paper. Here it refers to the pairing $Q(C)$ having zero orthogonal subspaces. In Section \ref{sec:amp-spin-def}, when talking about pivotal monoidal categories, it will mean that there exists a copairing for a given pairing. Since $\mathcal{H}_1$ and $\mathcal{H}_2$ will typically be infinite-dimensional, even if $Q(C)$ is non-degenerate there need not exist a corresponding copairing. Hence the second notion of ``non-degenerate'' is stronger.}

We assume that all defect conditions in $B$ describe {\em topological defects}. This is a requirement on the state spaces $\mathcal{H}$ and amplitudes $Q$:
\begin{itemize}
\item $\mathcal{H}$: 
Consider the $i$'th boundary component of a surface with defects with parametrisation map $\varphi_i$. For each defect ending on that boundary component write $(d,\varepsilon)$, where $d \in D$ is the defect condition and $\varepsilon=+$ if the defect is oriented away from the boundary, and $\varepsilon=-$ otherwise.
Starting from $\varphi_i(-1)$ let $\underline d=((d_1,\pm),\dots,(d_n,\pm))$ be the sequence of defects encountered when following $\varphi(S^1)$ clockwise. Then $\mathcal{H}_i$ depends on the defects only through the ordered list $\underline d$, but not on where and how precisely the defect lines end on the boundary component.
\item $Q$: 
Let $\Sigma^d_1$ and $\Sigma^d_2$ be two surfaces with defects with the same underlying surface $\underline\Sigma$, such that the defect submanifold of $\Sigma_1^d$ (with its orientation and defect condition for each component) can be isotoped to that of $\Sigma^d_2$. The isotopy is allowed to move the endpoints of defect lines as long as no such endpoint crosses the points $\varphi_i(-1)$, $i=1,\dots,B$. Then $Q(\Sigma^d_1) = Q(\Sigma^d_2)$.
\end{itemize}
Each state space $\mathcal{H}_i$ contains a subspace $\mathcal{H}^\mathrm{top}_i$ of {\em topological junction fields} (which may be 0). When inserting a topological junction field $\psi$ in the $i$'th argument (corresponding to the $i$'th boundary component) the map $Q(\Sigma^d)$ does not change when changing the position or size of the $i$'th boundary component (by glueing on appropriate cylinders, see \cite{davydov2011field} for more details). 
We also demand that topological junction fields are parity-even.\footnote{The condition that topological junction fields are parity-even does not appear in \cite{davydov2011field} since the quantum field theories with defects discussed there take values in $\mathbf{Vect}$, not $\mathbf{SVect}$.}
This means that if we apply condition {\bf Q3} to a transposition of two adjacent arguments, one of which is a topological junction field, then no parity sign is produced. 

Below we will consider surfaces with defect networks. In such a network, the junction points are labelled by a topological junction field and it is understood that such a junction point represents a small circular hole, parametrised in a way compatible with the linear order of the defect lines attached to the junction point.
Let $\Sigma^d$ be such a surface, and let $p$ be a junction point labelled by a topological junction field $\psi$. 
When evaluating $Q$ on $\Sigma^d$ it is understood that the argument corresponding to the boundary component replacing the point $p$ is taken to be $\psi$. Since $\psi$ is a topological junction field, the precise shape of the small circular hole and the choice of parametrisation map are irrelevant.
Since a topological junction field is parity-even,  we do not need to remember the ordering of topological junction fields on surfaces with defect networks.

The topological junction fields can be used to turn the set $D$ of defect fields into a $\mathbb{C}$-linear pivotal monoidal category $\mathcal{D}$. The objects are sequences of defect conditions with orientations, $\underline d=((d_1,\varepsilon_1),\dots(d_n,\varepsilon_n))$ and morphisms $\underline d \to \underline d'$ are topological junction fields on a boundary circle to which the defects from $\underline d$ and $\underline d'$ are attached. The tensor product is concatenation of sequences; duality is orientation reversal $(d,\varepsilon) \mapsto (d,-\varepsilon)$ and reversion of the order of the sequence. Pivotality follows from rotation invariance of topological junction fields. For details we refer to \cite[Sect.\,2.4]{davydov2011field} and to \cite[Thm.\,3.3]{carqueville2012orbifold}.

We will assume the category $\mathcal{D}$ to be additive. This amounts to completing the set of (sequences of) defect conditions with respect to taking finite sums. The $\mathbb{C}$-linear additive pivotal monoidal category $\mathcal{D}$ of defect conditions and topological junction fields is the second central ingredient in this paper.

Below, we will often draw parts of defect surfaces to illustrate a configuration of defect lines and junction labels. We will use a notation for junction fields that is close to string diagram notation for morphisms in monoidal categories. Namely, we draw a topological junction field $\psi : \underline d \to \underline d'$ inserted on the surface as
\begin{equation}
\raisebox{-0.5\height}{\includegraphics[scale=0.4]{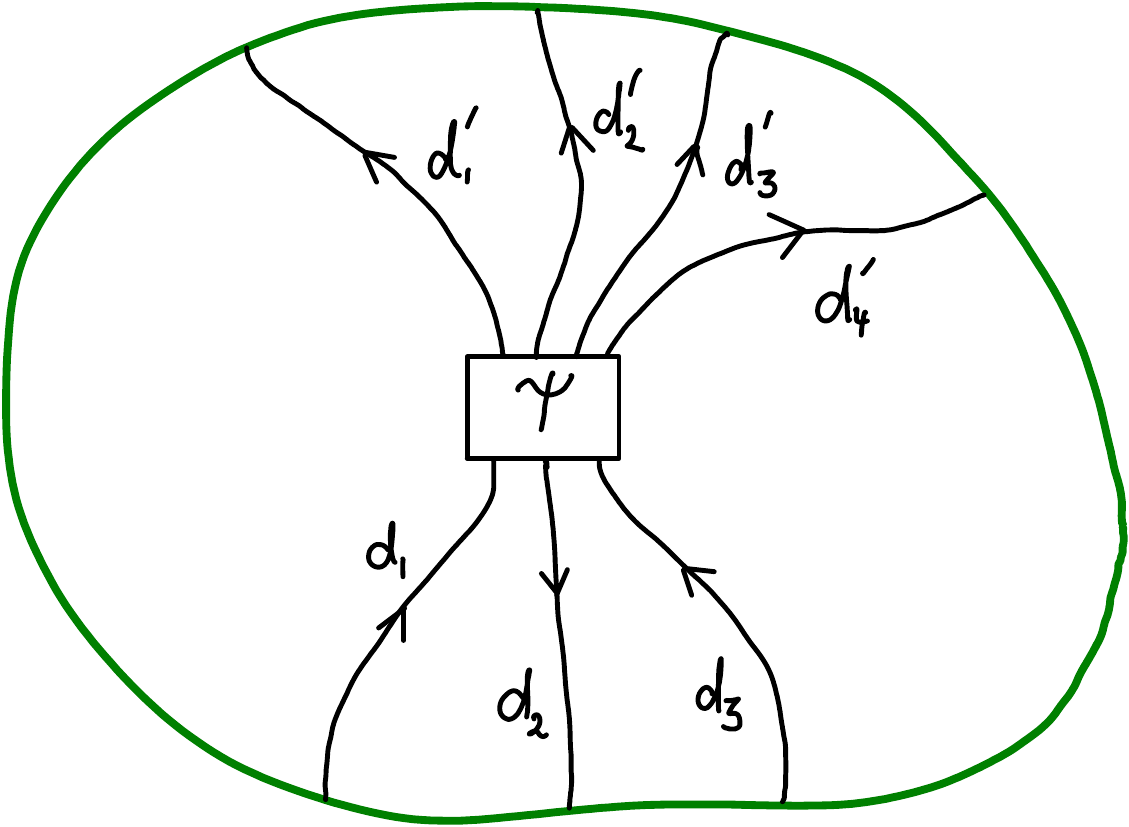}}
\qquad ,
\end{equation}
here in the example $\underline d = ((d_1,+),(d_2,-),(d_3,+))$, etc.
To turn this into a surface with defects, the box has to be replaced by a circular boundary, parametrised such that the image of $-1$ lies on the left vertical side of the box. Again, the precise way in which this replacement is done does not matter as $\psi$ is topological.

\subsection{The trivial defect TFT and products of defect theories}\label{sec:QxSV}

Let $\mathbf{SVect}$ be the category of $\Cb$-super vector spaces. This is a symmetric monoidal category with symmetric braiding given by exchange of factors with a parity sign:
For $U,V \in \mathbf{SVect}$ and $u \in U$, $v \in V$ homogeneous elements of degrees $|u|,|v| \in \mathbb{Z}_2$:
\begin{equation}
	\sigma_{U,V}(u \otimes v) = (-1)^{|u||v|} \, v \otimes u \ .
\end{equation}
We will denote by $Q^{SV}$ the trivial two-dimensional topological field theory with defects which takes values in $\mathbf{SVect}$. By this we mean the lattice TFT with defects constructed from the trivial algebra $\mathbb{C}$ via the procedure in \cite[Sect.\,3]{davydov2011field}.\footnote{The formalism in \cite[Sect.\,3]{davydov2011field} is developed for vector spaces as target category, but the same procedure works for any symmetric monoidal category.}
 Its defect category is $\mathcal{D}^{SV} = \mathbf{SVect}^{fd}$, the category of finite-dimensional $\Cb$-super vector spaces. Let $\underline V = ((V_1,\varepsilon_1),\cdots,(V_n,\varepsilon_n))$, with $V_i \in \mathbf{SVect}^{fd}$, be a sequence of defect conditions on a boundary component. The state space only depends on $\underline V$ and is given by $\mathcal{H}^{SV}(\underline V) = V_1^{\varepsilon_1} \otimes \cdots\otimes V_n^{\varepsilon_n}$, where $V^+ = V$, $V^- = V^*$. The topological junction fields therein are given by the even subspace, that is, $\mathcal{H}^{SV,top}(\underline V) = \mathrm{Hom}_0(\mathbb{C}^{1|0},V_1^{\varepsilon_1} \otimes \cdots\otimes V_n^{\varepsilon_n})$, where $\mathbb{C}^{1|0}$ is $\mathbb{C}$ concentrated in even degree and $\mathrm{Hom}_0$ denotes the ungraded $\Cb$-vector space of even linear maps.

Let us stress this point again: since the morphisms in $\mathcal{D}^{SV}$ are topological junctions, they are given by {\em even} linear maps (i.e.\ by morphisms in $\mathbf{SVect}^{fd}$). The state space assigned to a boundary circle, on the other hand, may contain even and odd elements.

\medskip

We will make use of $Q^{SV}$ to introduce a $\mathbb{Z}_2$-grading on the state spaces of ungraded defect theories $Q$ in Sections \ref{sec:spin-def-graded}--\ref{sec:so(n)}. This will be done by defining a product theory $Q \boxtimes Q^{SV}$. The necessity of such products will be illustrated in the free fermion example in Section \ref{sec:1ff-ex}.

Let $Q^1$, $Q^2$ be two QFTs on surfaces with defects. The product theory  $Q^1 \boxtimes Q^2 =: Q^{12}$ is defined as follows: 
\begin{itemize}
\item
The defect category is $\mathcal{D}^{12} = \mathcal{D}^1 \boxtimes \mathcal{D}^2$. Here ``$\boxtimes$'' is the product of $\mathbb{C}$-linear additive categories: the objects are direct sums of pairs $(\underline{d}_1,\underline{d}_2)$ of objects $\underline{d}_1 \in \mathcal{D}^1$ and $\underline{d}_2 \in \mathcal{D}^2$. The morphism space between two such pairs $(\underline{d}_1,\underline{d}_2) \to (\underline{d}_1',\underline{d}_2')$ is $\mathcal{D}^1(\underline{d}_1,\underline{d}_1') \otimes_\mathbb{C} \mathcal{D}^2(\underline{d}_2,\underline{d}_2')$. For direct sums of such pairs, the morphism spaces are the corresponding direct sums.
\item
The state space $\mathcal{H}^{12}_i$ of the product theory for the $i$'th boundary component of a defect surface $\Sigma^d$  is given by $\mathcal{H}^1_i \otimes \mathcal{H}^2_i$ (tensor product of $\Cb$-super vector spaces).
\item
The amplitudes $Q^{12}$ of the product theory are $Q^{12}(\Sigma^d) = (Q^1(\Sigma^d) \otimes Q^2(\Sigma^d)) \circ P$, where $P$ is the permutation in $\mathbf{SVect}$ (i.e.\ the permutation with parity signs) which reorders the argument from $\mathcal{H}^1_1 \otimes \mathcal{H}^2_1 \otimes \cdots \otimes  \mathcal{H}^1_B \otimes \mathcal{H}^2_B$ to $\mathcal{H}^1_1 \otimes \cdots \otimes \mathcal{H}^1_B \otimes \mathcal{H}^2_1 \otimes \cdots \otimes \mathcal{H}^2_B$.
\end{itemize}

We will only apply the product construction in the case that all state spaces of $Q^1 = Q$ are purely even and $Q^2 = Q^{SV}$ is the trivial theory with values in $\mathbf{SVect}$ defined above. In this situation there are no parity signs in the permutation map $P$ above.
 The topological theory $Q^{SV}$ has two fundamental defects, labelled by the one-dimensional even and odd super vector spaces $\mathbb{C}^{1|0}$ and $\mathbb{C}^{0|1}$. The state spaces of $Q \boxtimes Q^{SV}$ for a boundary circle with a defect line $X \boxtimes \mathbb{C}^{1|0}$ or $X \boxtimes \mathbb{C}^{0|1}$ starting at it are
\begin{align} \label{eq:X-vs-PiX}
	\mathcal{H}^{Q \boxtimes SV}((X \boxtimes \mathbb{C}^{1|0} ,+)) &= 
	\mathcal{H}^Q((X,+)) \otimes \mathcal{H}^{SV}((\Cb^{1|0},+)) 
	\cong \mathcal{H}^Q((X,+))  \ ,
\\
	\mathcal{H}^{Q \boxtimes SV}((X \boxtimes \mathbb{C}^{0|1} ,+)) &= 
	\mathcal{H}^Q((X,+)) \otimes \mathcal{H}^{SV}((\Cb^{0|1},+)) 
	\cong \Pi\mathcal{H}^Q((X,+)) \ ,
\nonumber
\end{align}
where for a super vector space $V$, $\Pi V$ denotes the parity shifted super vector space, i.e.\ $(\Pi V)_0 = V_1$ and $(\Pi V)_1 = V_0$. The last isomorphism in the two sets of identities above holds because by definition $\mathcal{H}^{SV}((\Cb^{1|0},+)) = \Cb^{1|0}$ and $\mathcal{H}^{SV}((\Cb^{0|1},+)) = \Cb^{0|1}$.

\section{Amplitudes on spin surfaces from amplitudes on defect surfaces}\label{sec:amp-spin-def}

Fix a QFT $Q$ on surfaces with defects, possibly equipped with a metric or a conformal structure. Let $\mathcal{D}$ be the corresponding monoidal category of topological defects and topological junction fields. Fix furthermore a defect condition $A \in \mathcal{D}$ and topological junction fields $t, c_{\pm1}$ which are morphisms
\begin{equation}
t : A \otimes A \otimes A \to \mathbf{1}
\quad , \qquad
c_{+1} , c_{-1} : \mathbf{1} \to A \otimes A 
\end{equation}
in $\mathcal{D}$.
The aim of this section is to explain how to obtain from the data $Q,A,t,c_{\pm1}$, subject to certain conditions, a QFT $Q_\mathrm{spin}$ which assigns amplitudes to spin surfaces.

\subsection{From spin structures to defect networks}\label{sec:spin-surf-to-def-surf}

\begin{figure}[tb]

\begin{center}
\raisebox{4em}{a)}
\hspace{1em}
$\raisebox{-0.3\height}{\includegraphics[scale=0.4]{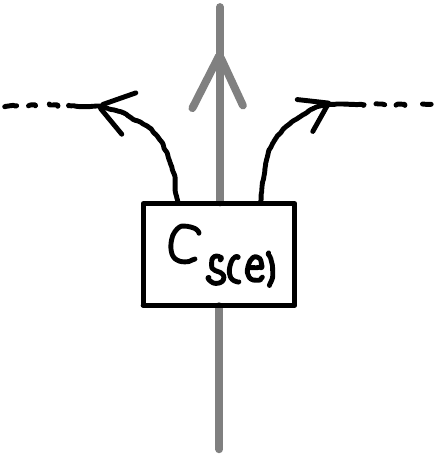}}$
\hspace{2em}
\raisebox{4em}{b)}
\hspace{1em}
$\raisebox{-0.5\height}{\includegraphics[scale=0.4]{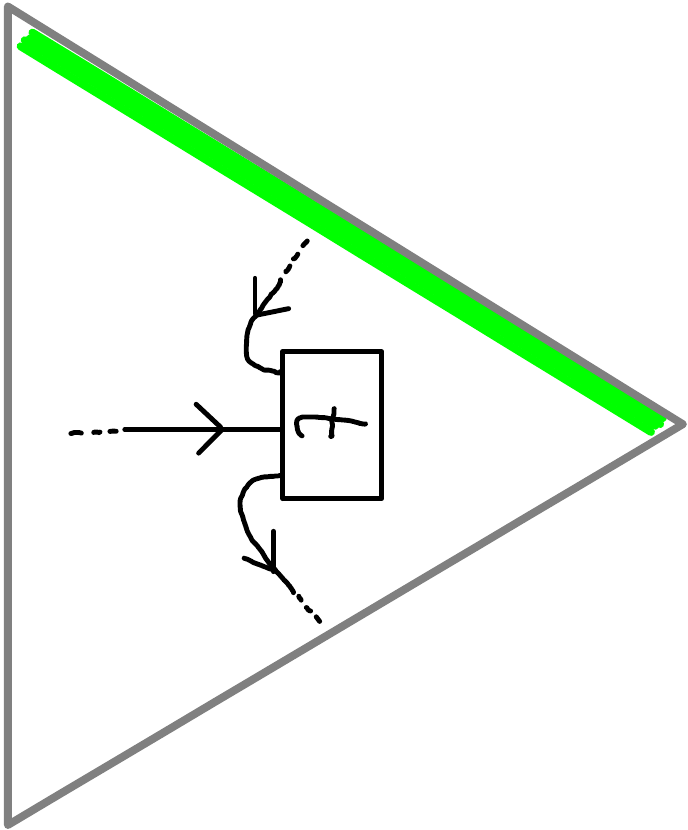}}$
\\[2em]
\raisebox{6em}{c)}
$\raisebox{-0.5\height}{\includegraphics[scale=0.4]{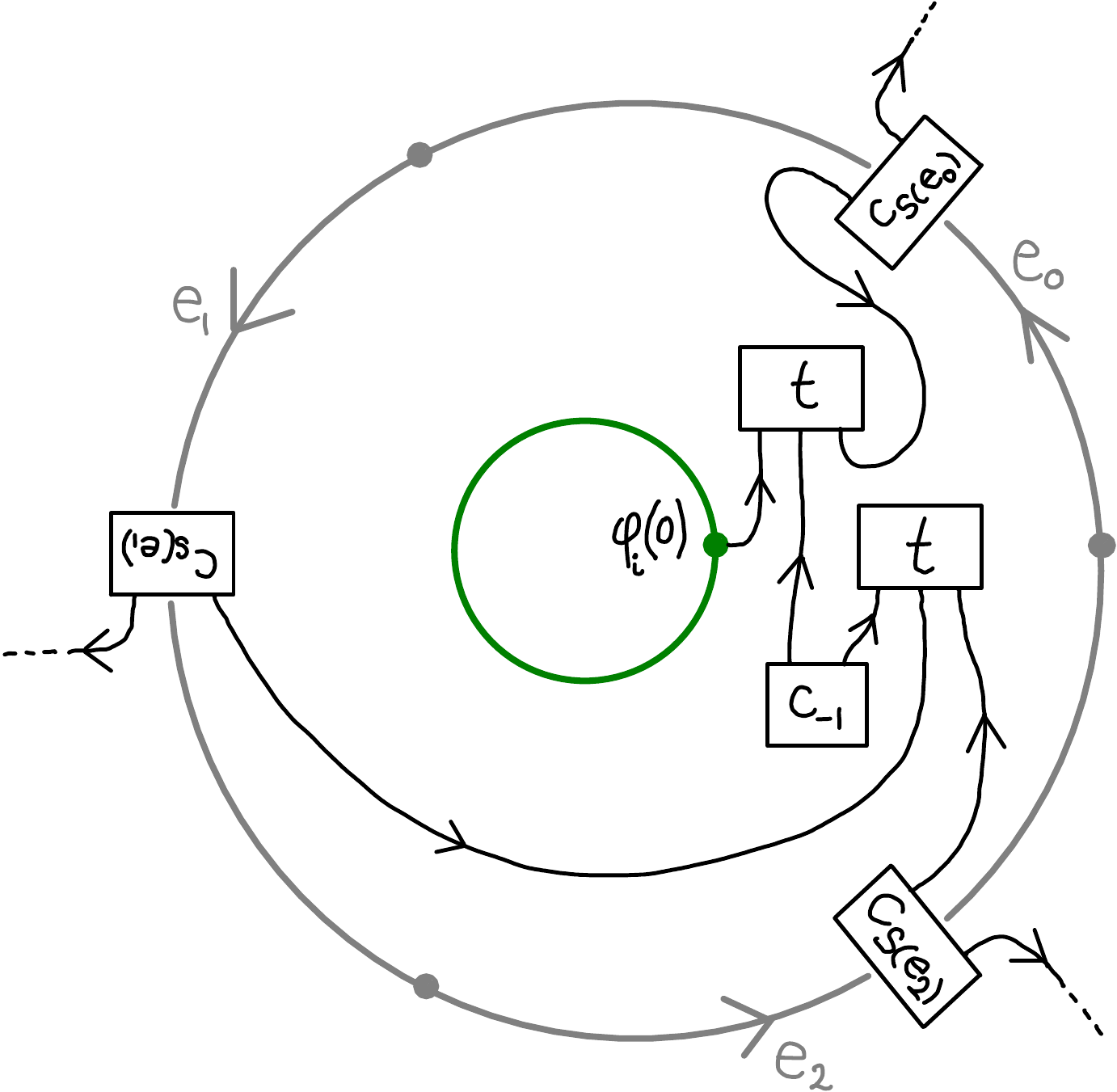}}$
\end{center}
\caption{a) Insertion of the topological junction field $c_{s(e)}$ on an oriented edge $e$ with edge sign $s(e)$.
b) On a triangle $t$ is inserted such that the left-most $A$-line entering $t$ comes from the marked edge of the triangle. 
c) The defect network inserted at each boundary component. Here the boundary triangulation with edges $e_0$, $e_1$, $e_2$ has been moved into the surface to make the defect network easier to draw.}
\label{fig:triang-to-def-rules}
\end{figure}

Let $\Sigma$ be a spin surface. Our first task is to produce from $\Sigma$ a surface with defects $\Sigma^d$. The construction is as follows:
\begin{itemize}
\item[(a)]
Pick a triangulation of $\underline\Sigma$. Equip the triangulation with a marking and with edge signs encoding the spin structure of $\Sigma$.
\item[(b)]
On each inner edge insert the topological junction field $c_{s(e)}$, aligned with the orientation of the edge as shown in Figure \ref{fig:triang-to-def-rules}\,a. In each triangle place the field $t$, aligned with the marked edge of that triangle as shown in Figure \ref{fig:triang-to-def-rules}\,b.
\item[(c)]
At each boundary component of $\underline\Sigma$, the slightly more complicated looking defect network shown in Figure \ref{fig:triang-to-def-rules}\,c is inserted.
\end{itemize}
We now need to make sure that the amplitude $Q(\Sigma^d)$ does not  depend on the choices made in (a). This is guaranteed by a number of identities on $t,c_{\pm1}$. We will write these as identities of morphisms in the defect category $\mathcal{D}$, using string diagram notation. By definition of $\mathcal{D}$, we may then as well think of these string diagrams as identities satisfied by amplitudes of surfaces with defects which only differ locally as indicated in the string diagram. 

The derivation of the identities below is essentially the same as in \cite{Novak:2014oca} with one important exception: In \cite{Novak:2014oca} we worked with a symmetric monoidal category, and so one is allowed to use the symmetric braiding but no duals; here we work in a pivotal non-braided category, so one is allowed to use duals but no braiding. We omit the calculations and just state the results, using the same ordering of the conditions as in \cite[Sect.\,4.2]{Novak:2014oca}.

\begin{enumerate}  \setlength{\leftskip}{-1.5em}
\item {\em Edge orientation change.} For $\varepsilon \in \{ \pm 1 \}$,
\begin{equation}\label{eq:inv-cond-1}
\raisebox{-0.5\height}{\includegraphics[scale=0.4]{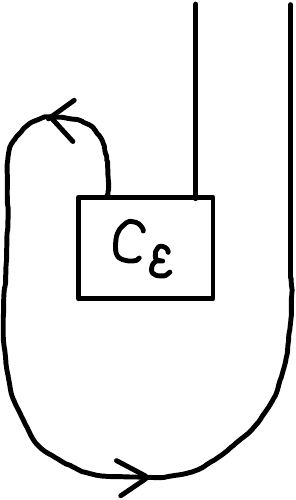}}
\quad=\quad
\raisebox{-0.2\height}{\includegraphics[scale=0.4]{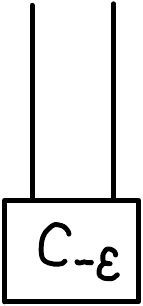}}
\quad=\quad
\raisebox{-0.5\height}{\includegraphics[scale=0.4]{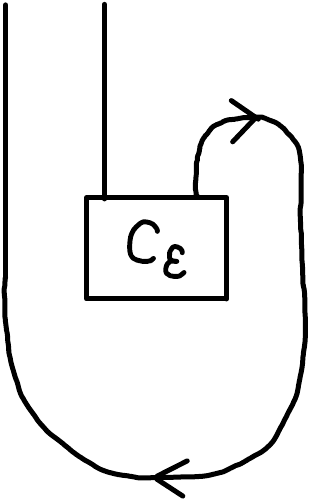}}
\end{equation}
(The first and second equality are equivalent by pivotality.)
\item {\em Leaf exchange automorphism on a single triangle.}
For $\alpha,\beta,\gamma \in \{ \pm 1 \}$,
\begin{equation}\label{eq:inv-cond-2}
\raisebox{-0.5\height}{\includegraphics[scale=0.4]{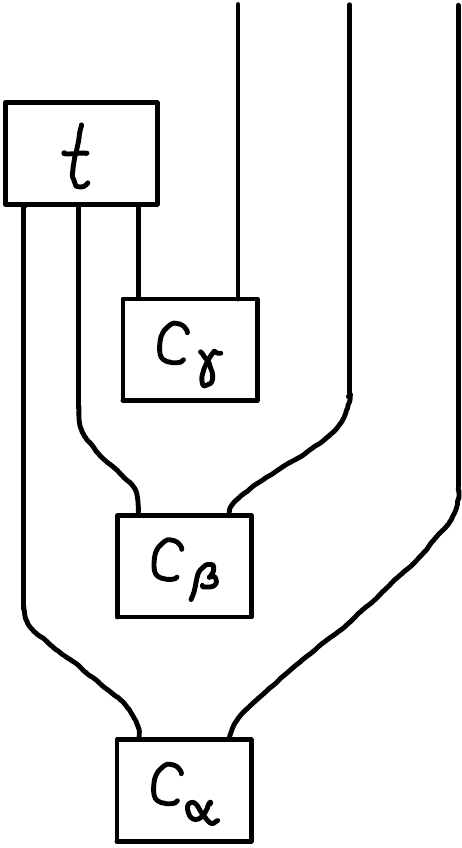}}
\quad=\quad
\raisebox{-0.5\height}{\includegraphics[scale=0.4]{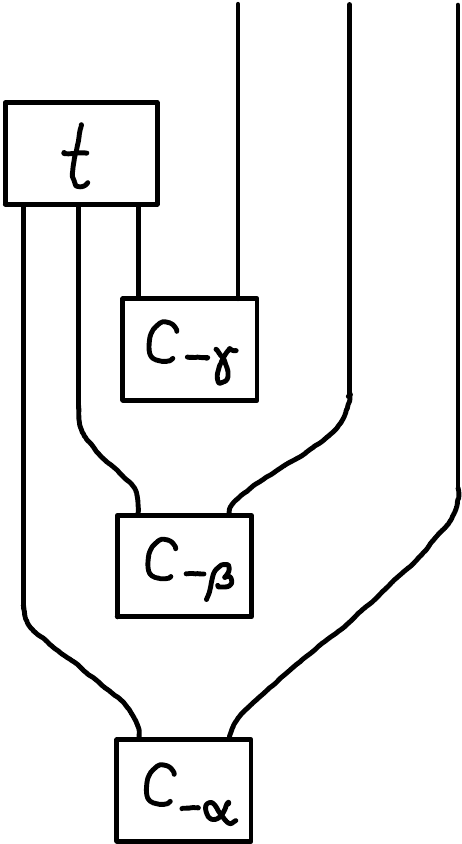}}
\end{equation}
\item {\em Cyclic permutation of boundary edges for a single triangle.}
For $\alpha,\beta,\gamma \in \{ \pm 1 \}$,
\begin{equation}\label{eq:inv-cond-3}
\raisebox{-0.5\height}{\includegraphics[scale=0.4]{cpic04a.pdf}}
\quad=\quad
\raisebox{-0.5\height}{\includegraphics[scale=0.4]{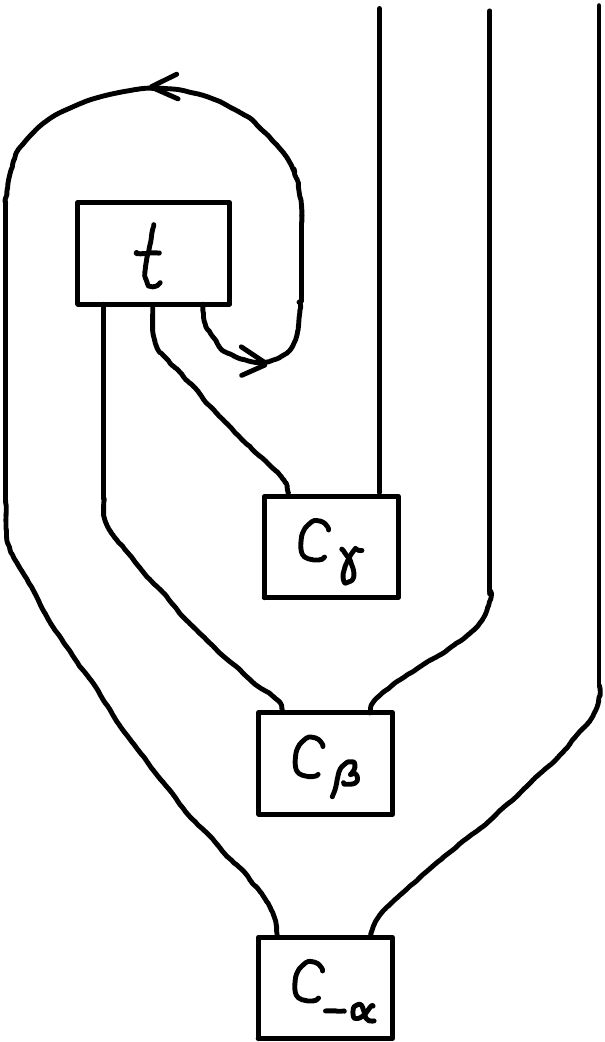}}
\end{equation}
\item {\em Pachner 2-2 move.} 
For $\alpha,\beta,\gamma,\delta,s \in \{ \pm 1 \}$,
\begin{equation}\label{eq:inv-cond-4}
\raisebox{-0.5\height}{\includegraphics[scale=0.4]{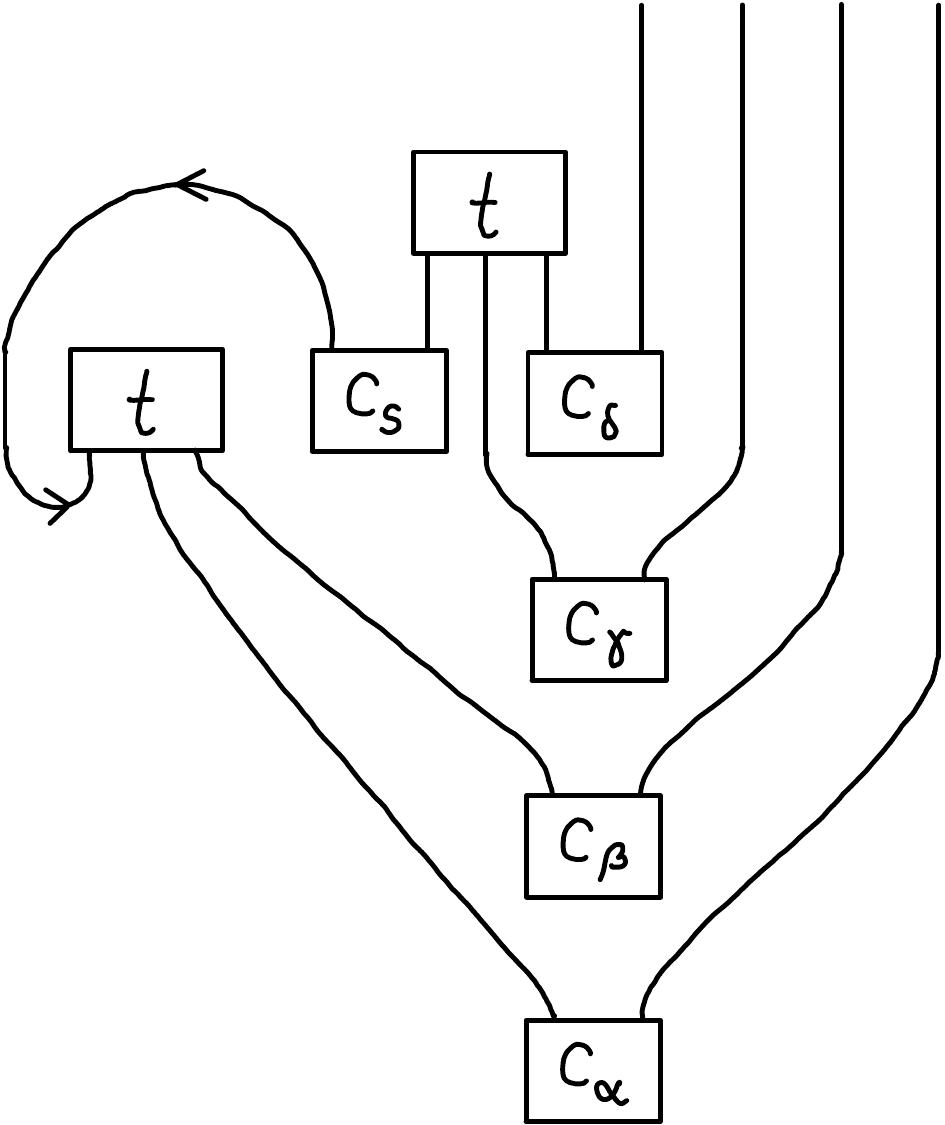}}
\quad=\quad
\raisebox{-0.5\height}{\includegraphics[scale=0.4]{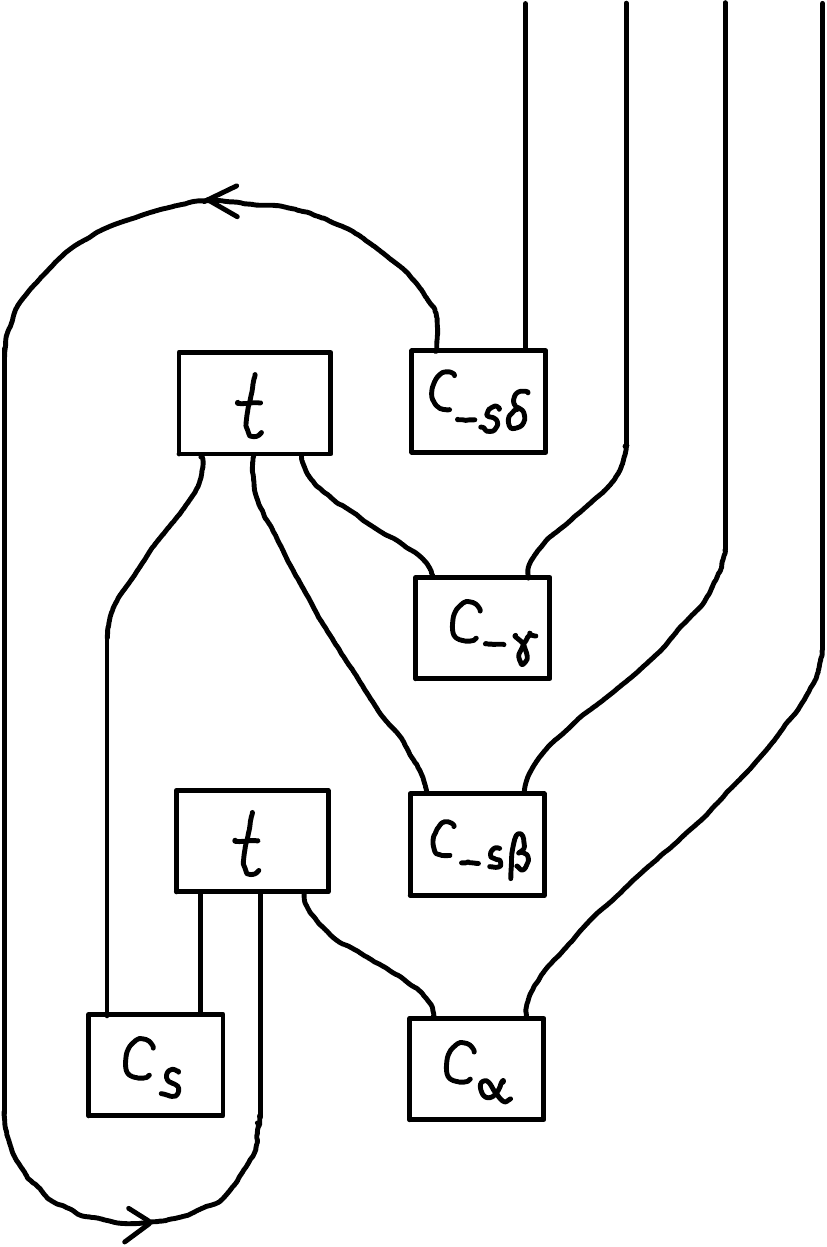}}
\end{equation}
\item {\em Pachner 3-1 move and its inverse.}
For $\alpha,\beta,\gamma,s_{12},s_{23},s_{31} \in \{ \pm 1 \}$ subject to $s_{12}s_{23}s_{31}=-1$,
\begin{equation}\label{eq:inv-cond-5}
\raisebox{-0.5\height}{\includegraphics[scale=0.4]{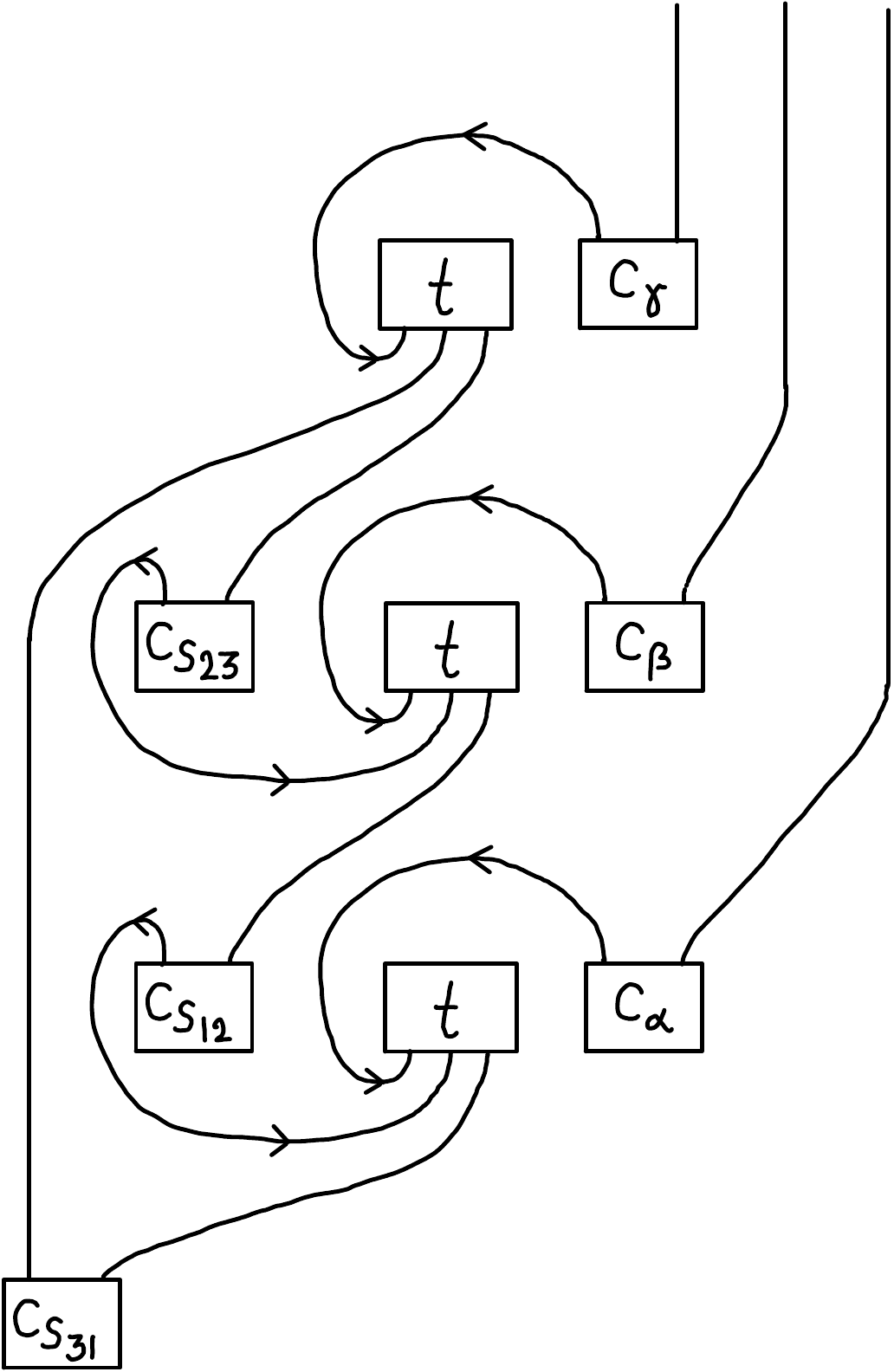}}
\quad
=
\quad
\raisebox{-0.5\height}{\includegraphics[scale=0.4]{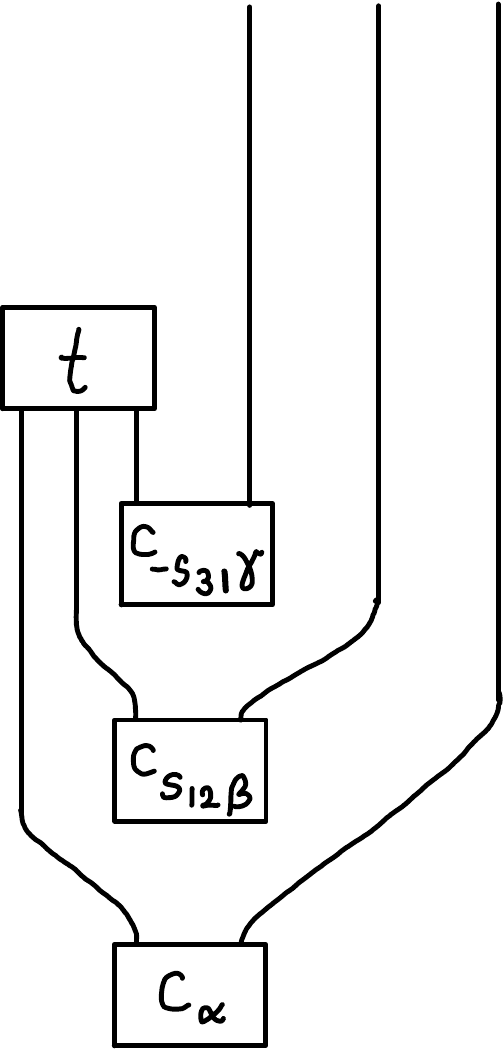}}
\end{equation}
\end{enumerate}
Using these conditions, we can now state:

\begin{proposition}\label{prop:Qspin-indep-choices}
Suppose $Q$ and $A,t,c_{\pm 1}$ satisfy conditions (1)--(5) above.
Let $\Sigma$ be a spin surface and $\Sigma^d_1$, $\Sigma^d_2$ surfaces with defects obtained by steps (a)--(c) above. Then $Q(\Sigma^d_1)=Q(\Sigma^d_2)$.
\end{proposition}

The proof is analogous to that of \cite[Prop.\,4.1]{Novak:2014oca} 
and we omit it here.

\subsection{Algebraic description}\label{sec:alg-desc}

\begin{figure}[tb]
\begin{center}
\begin{tabular}{cccc}
$\mathrm{ev}_U \,=~ \raisebox{-0.5\height}{\includegraphics[scale=0.4]{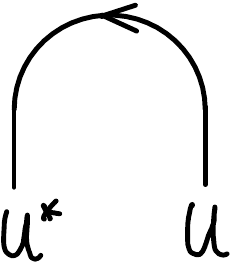}}$
&
$\widetilde{\mathrm{ev}}_U \,=~ \raisebox{-0.5\height}{\includegraphics[scale=0.4]{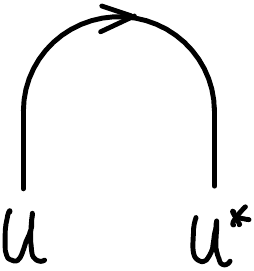}}$
&
$\mathrm{coev}_U \,=~ \raisebox{-0.5\height}{\includegraphics[scale=0.4]{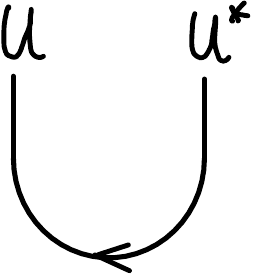}}$
&
$\widetilde{\mathrm{coev}}_U \,=~ \raisebox{-0.5\height}{\includegraphics[scale=0.4]{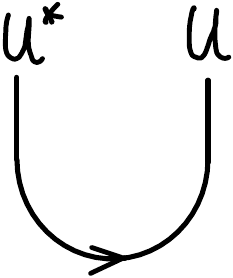}}$
\\[3em]
$\mu \,= \raisebox{-0.5\height}{\includegraphics[scale=0.4]{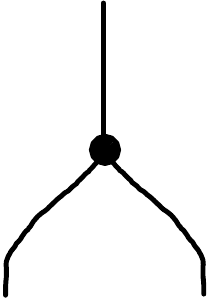}}$
&
$\eta \,=~ \raisebox{-0.5\height}{\includegraphics[scale=0.4]{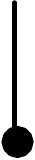}}$
&
$\Delta \,= \raisebox{-0.5\height}{\includegraphics[scale=0.4]{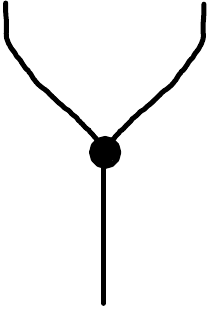}}$
&
$\varepsilon \,=~ \raisebox{-0.5\height}{\includegraphics[scale=0.4]{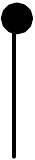}}$
\\[3em]
$b \,= \raisebox{-0.5\height}{\includegraphics[scale=0.4]{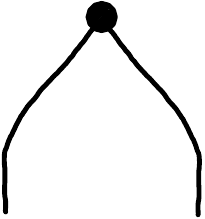}}$
&
$c_{-1} \,= \raisebox{-0.5\height}{\includegraphics[scale=0.4]{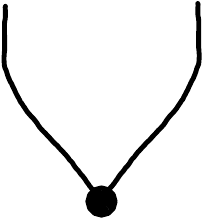}}$
&
$N \,=~ \raisebox{-0.5\height}{\includegraphics[scale=0.4]{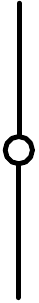}}$
\end{tabular}
\end{center}
\caption{Graphical notation for duality maps and shorthand graphical notation for frequently used morphisms.}
\label{fig:graph-short}
\end{figure}

A {\em Frobenius algebra} in a monoidal category is an object $A$ together with a unital algebra and counital coalgebra structure, such that the coproduct is a bimodule map. We denote the structure morphisms as
\begin{equation}
	\mu : A \otimes A \to A 
~,\quad
	\eta : \mathbf{1} \to A
~,\quad
	\Delta : A \to A \otimes A
~,\quad
	\varepsilon : A \to \mathbf{1} \ .
\end{equation}
A Frobenius algebra is called {\em $\Delta$-separable} if $\mu \circ \Delta = \id_A$. In a pivotal monoidal category, the Nakayama automorphism $N$ of a Frobenius algebra $A$ takes the form
\begin{equation}\label{eq:Nakayama-def}
	N ~=~~ \raisebox{-0.5\height}{\includegraphics[scale=0.4]{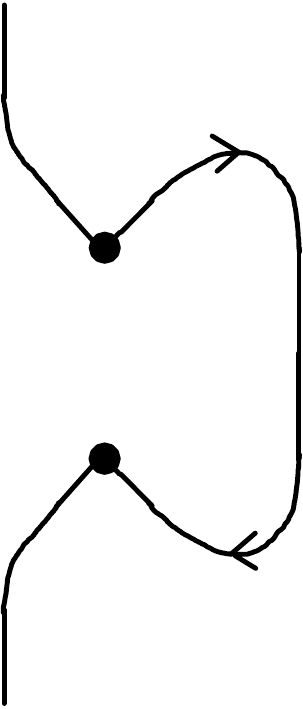}} \quad ,
\end{equation}
where in the above string diagram we use the graphical shorthands for morphisms as collected in Figure \ref{fig:graph-short}. One checks that $N$ is an algebra and coalgebra automorphism.

\medskip

After these preliminaries, we formulate the data $t,c_{\pm1}$ subject to conditions (1)--(5) above in terms of a Frobenius algebra with extra properties. This is done by the same procedure as \cite[Sect.\,4.3]{Novak:2014oca}
and we state the following results without proof.

Define the product $\mu : A \otimes A \to A$ as
\begin{equation}
\mu 
~=~ \raisebox{-0.5\height}{\includegraphics[scale=0.4]{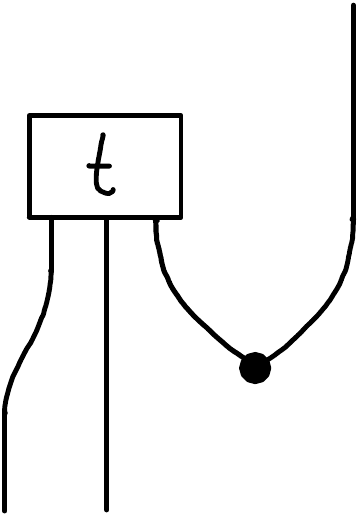}}
\quad .
\end{equation}
As in \cite[Sect.\,4.3]{Novak:2014oca} 
we make the
\begin{quote}
{\bf Assumptions:} 1) The coparing $c_{-1}$ is non-degenerate. 2) The product $\mu$ has a unit $\eta : \mathbf{1} \to A$.
\end{quote}
We denote the pairing dual to the copairing $c_{-1}$ by $b : A\otimes A \to \mathbf{1}$. Using the graphical shorthands in Figure \ref{fig:graph-short}, we set
\begin{equation}\label{eq:coproduct-counit}
    \Delta 
~=~ \raisebox{-0.5\height}{\includegraphics[scale=0.4]{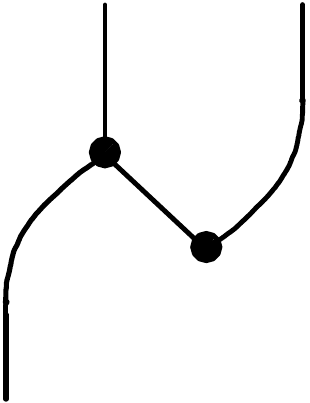}}
~=~ \raisebox{-0.5\height}{\includegraphics[scale=0.4]{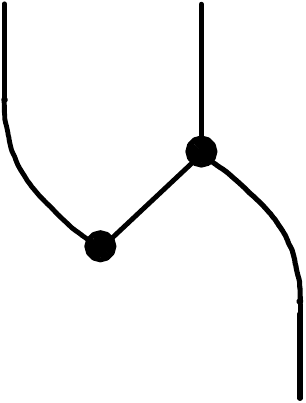}}
\quad,\qquad
    \varepsilon  
~=~ \raisebox{-0.5\height}{\includegraphics[scale=0.4]{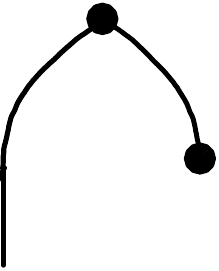}}
~=~ \raisebox{-0.5\height}{\includegraphics[scale=0.4]{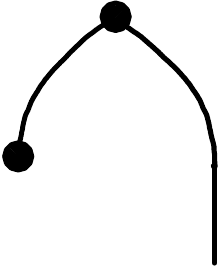}}
 \quad .
\end{equation}
That the two distinct expressions for $\Delta$ and $\varepsilon$ are indeed equal is straightforward to check, see e.g.\ \cite[Sect.\,4.3]{Novak:2014oca}. We have the following theorem, the proof of which is easily adapted from \cite[Sect.\,4.3]{Novak:2014oca}.

\begin{theorem}    \label{thm:moves_vs_alg}
\begin{enumerate}
\item
Let $t,c_{\pm 1}$ fulfil relations (1)--(5) in Section \ref{sec:spin-surf-to-def-surf}. Under assumpions 1,2 and defining $\mu,\eta,\Delta,\varepsilon$ as above, $A$ is a $\Delta$-separable Forbenius algebra satisfying $N \circ N = \id_A$.
\item
Let conversely  $A$ be a $\Delta$-separable Forbenius algebra satisfying $N \circ N = \id_A$. Then 
\begin{equation*} \label{eq:tc-via-A}
 t =\varepsilon \circ \mu  \circ (\mu \otimes \id_A) ~, \quad
 c_{-1} = \Delta \circ \eta ~, \quad
 c_{1} = (N \otimes \id) \circ \Delta \circ \eta \ .
\end{equation*} 
satisfy relations (1)--(5) in Section \ref{sec:spin-surf-to-def-surf}, as well as assumpions 1,2.
\end{enumerate}
\end{theorem}

\subsection{State spaces}\label{sec:state-spaces}

Consider a surface with defects $C$ which is diffeomorphic to a cylinder with one $A$-defect connecting the two parametrised boundaries. 
Let $\mathcal{H}_A$ be the state space assigned to the boundary component from which the $A$-defect points away. In this section we will define two subspaces $\mathcal{H}^{NS} \subset \mathcal{H}_A$ and $\mathcal{H}^R\subset \mathcal{H}_A$ and investigate their properties.

\medskip

To start with, we will need two more assumption on the behaviour of $Q$. The first is easy:
\begin{quote}
{\bf Assumption:} 3) The QFT $Q$ is non-degenerate.
\end{quote}

The second additional assumption is slightly more elaborate. Let $X,Y,Z \in \mathcal{D}$ and $f : Z \otimes X \to Y \otimes Z$ a morphism in $\mathcal{D}$, i.e.\ a topological junction field. 
Consider the two cylinders
\begin{equation}\label{eq:cylinder-junction-to-map}
C_{X,Y,Z}^f ~=~  \raisebox{-0.5\height}{\includegraphics[scale=0.3]{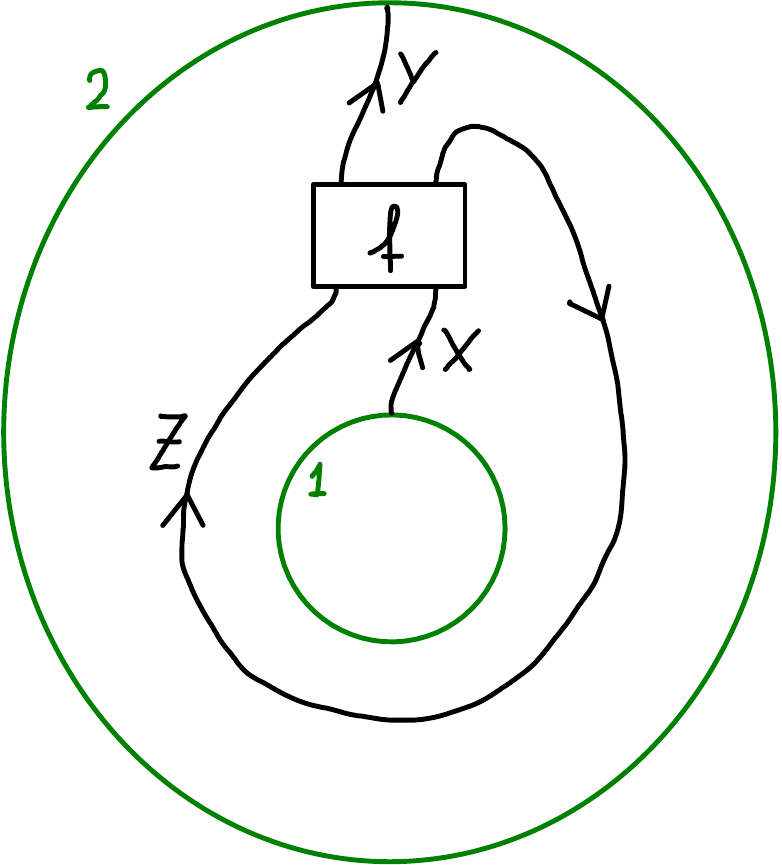}}
\qquad , \qquad
C_{Y} ~=~  \raisebox{-0.5\height}{\includegraphics[scale=0.3]{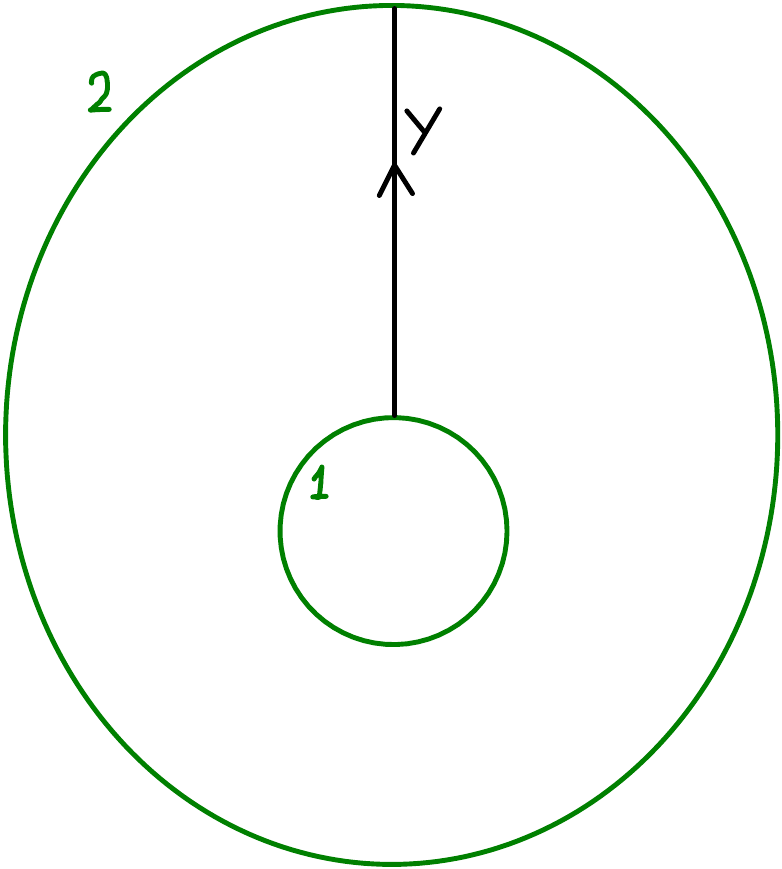}}  \ .
\end{equation}
Here we suppose that the underlying surface of both cylinders is the same, and that the defect graph (but not the labelling) agrees in a neighbourhood of both boundary components. Note that we can then write $C_Y = C_{Y,Y,\one}^{\id_Y}$. 

Let $\mathcal{K}_Y$ be the state space assigned to boundary component 2 of $C_{X,Y,Z}^f$ and $C_Y$, and let $\mathcal{H}_X$ (resp.\ $\mathcal{H}_Y$) be that assigned to boundary component 1 of  $C_{X,Y,Z}^f$ (resp.\ $C_Y$).
We would like to encode the difference between the value of $Q$ on the two cylinders
in a linear map $T[f] : \mathcal{H}_X \to \mathcal{H}_Y$. That is, we assume:
\begin{quote}
{\bf Assumption:} 4) In the situation \eqref{eq:cylinder-junction-to-map} there exists a linear map $T[f]: \mathcal{H}_X \to \mathcal{H}_Y$ such that for all $\phi \in \mathcal{H}_X$ and $\psi \in \mathcal{K}_Y$ we have
\begin{equation}\label{eq:T[f]-def}
	Q(C_{X,Y,Z}^f)(\phi \otimes \psi) = Q(C_Y)(T[f](\phi) \otimes \psi) \ .
\end{equation}
\end{quote}
The linear map $T[f]$ is unique by non-degeneracy in Assumption 3. The symbol ``$T$'' is chosen as $T[f]$ behaves like a partial trace:

\begin{lemma}\label{lem:T[]-properties}
Let $U,V,W,Z,Z' \in \mathcal{D}$. Then
\begin{enumerate}
\item for all $f : Z' \otimes U \to V \otimes Z$ and $h : Z \to Z'$ we have 
$$T\big[\,(\id_V \otimes h) \circ f\,\big] ~=~ T\big[\,f \circ (h \otimes \id_U)\,\big] \ ,$$
\item for all 
$f : Z \otimes U \to V \otimes Z$ and $g : Z' \times V \to W \otimes Z'$ we have
$$T[g] \circ T[f] ~=~ T\big[\,(g \otimes \id_Z) \circ (id_{Z'} \otimes f)\,\big] \ .$$
\end{enumerate}
\end{lemma}

\begin{proof} For part 1 consider the equalities
\begin{align}
& Q(C_V) \circ \big( T\big[\,(\id_V \otimes h) \circ f\,\big] \otimes \id \big)
~\overset{(1)}{=}~
Q\Bigg(~\raisebox{-0.5\height}{\includegraphics[scale=0.3]{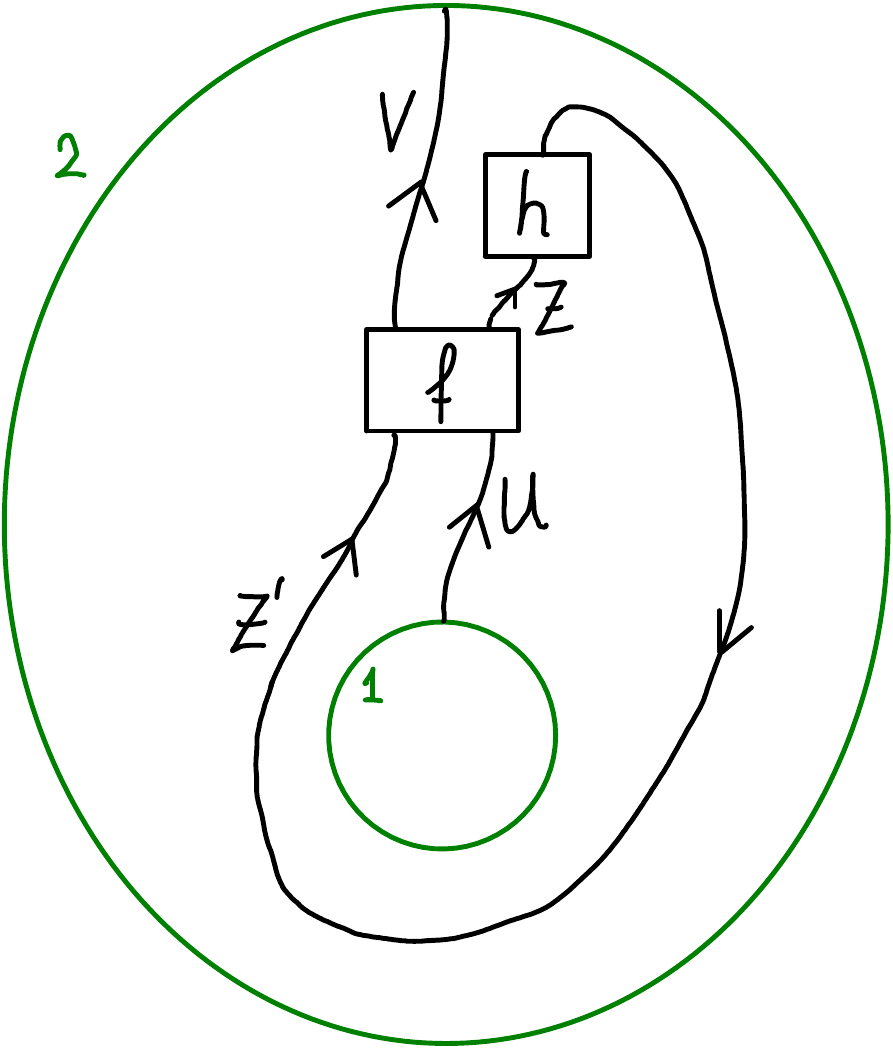}}~\Bigg)
\\[-2em] \nonumber
&\overset{(2)}{=}~
Q\Bigg(~\raisebox{-0.5\height}{\includegraphics[scale=0.3]{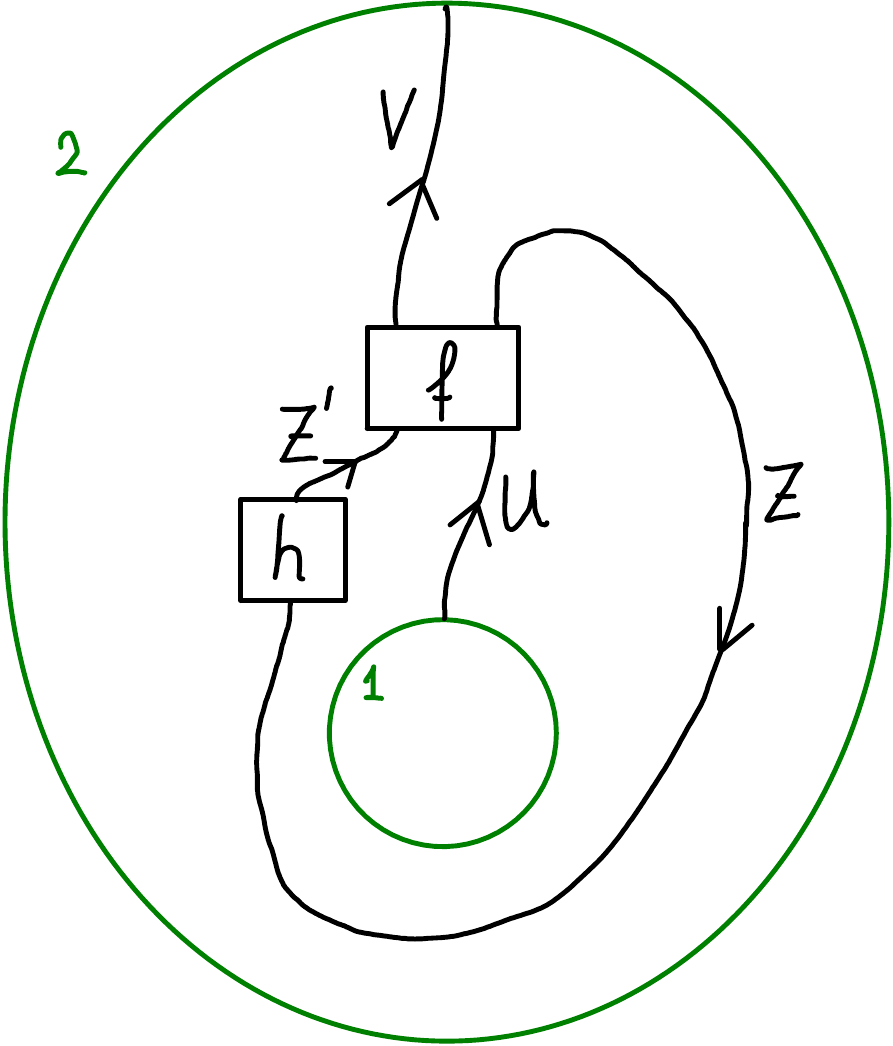}}~\Bigg)
~\overset{(3)}{=}~
Q(C_V) \circ \big( T\big[\,f \circ (h \otimes \id_U)\,\big] \otimes \id \big) \ .
\end{align}
Step 1 is the definition of $T$, step 2 uses invariance of topological defects and junctions under isotopies, and step 3 is again the definition of $T$. Non-degeneracy of $Q(C_V)$ (Assumption 3) implies the identity in part 1.
Part 2 works along the same lines:
\begin{align}
&Q(C_V) \circ \big( (T[g] \circ T[f]) \otimes \id \big)
~=~
Q\Bigg(~\raisebox{-0.5\height}{\includegraphics[scale=0.3]{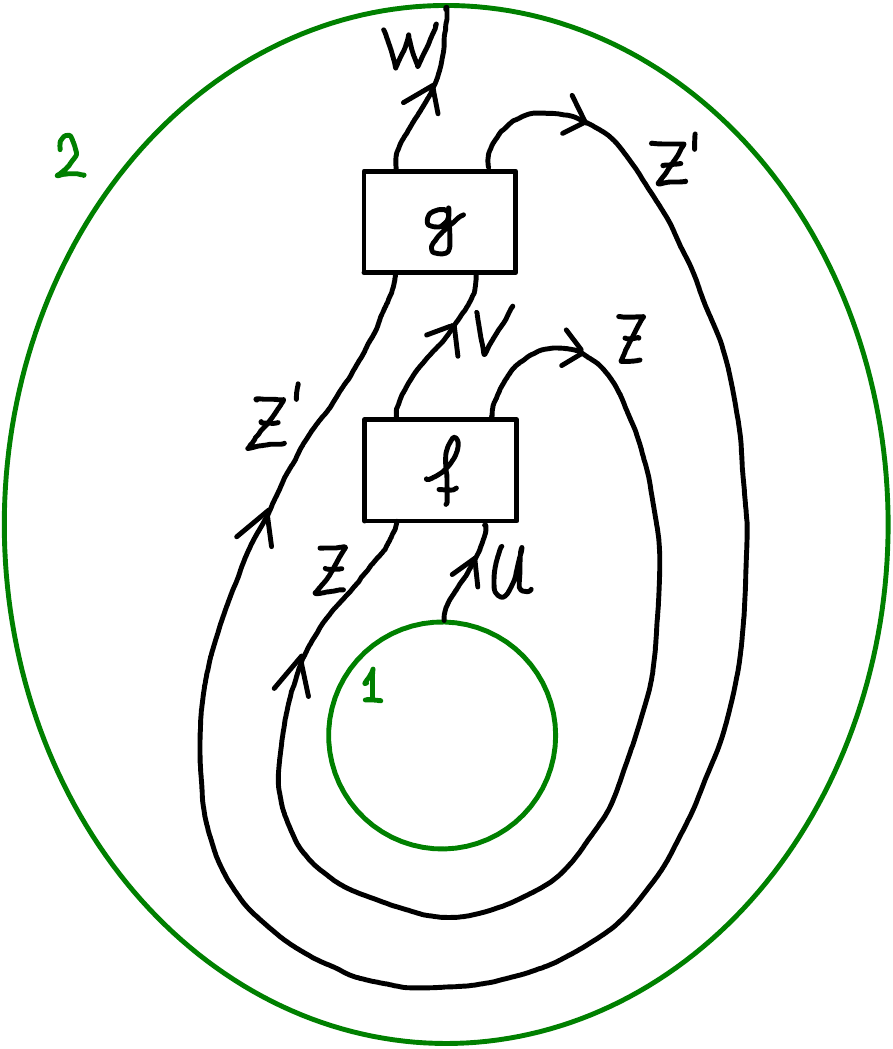}}~\Bigg)
\\[.5em] \nonumber
&
=
Q(C_V) \circ \big( T\big[\,(g \otimes \id_Z) \circ (id_{Z'} \otimes f)\,\big] \otimes \id \big) \ .
\end{align}
\end{proof}

In the following computation we fix an underlying cylinder $C$ as in \eqref{eq:cylinder-junction-to-map} and just vary the defect network and junction fields. In all cases a defect of type $A$ will emanate from boundary 1 (the ``inner'' boundary in the picture), and we denote by $\mathcal{H}_A$ the corresponding state space.
Using Assumption 4, we define the following endomorphisms of $\mathcal{H}_A$:
\begin{equation}\label{eq:PNS/R-def-via-T}
	N_\mathcal{H_A} := T[N]
	~~,\quad
	P^{NS} := T[\Delta \circ \mu \circ (N \otimes \id_A)]
	~~,\quad
	P^{R} := T[\Delta \circ \mu] \ .
\end{equation} 
Here, in the first definition $N$ is understood as a morphism $\one \otimes A \to A \otimes \one$. In terms of cylinders $C$ with defect networks, in the case of $P^{NS}$ for example, we have
\begin{equation}
	Q\Bigg(~\raisebox{-0.5\height}{\includegraphics[scale=0.3]{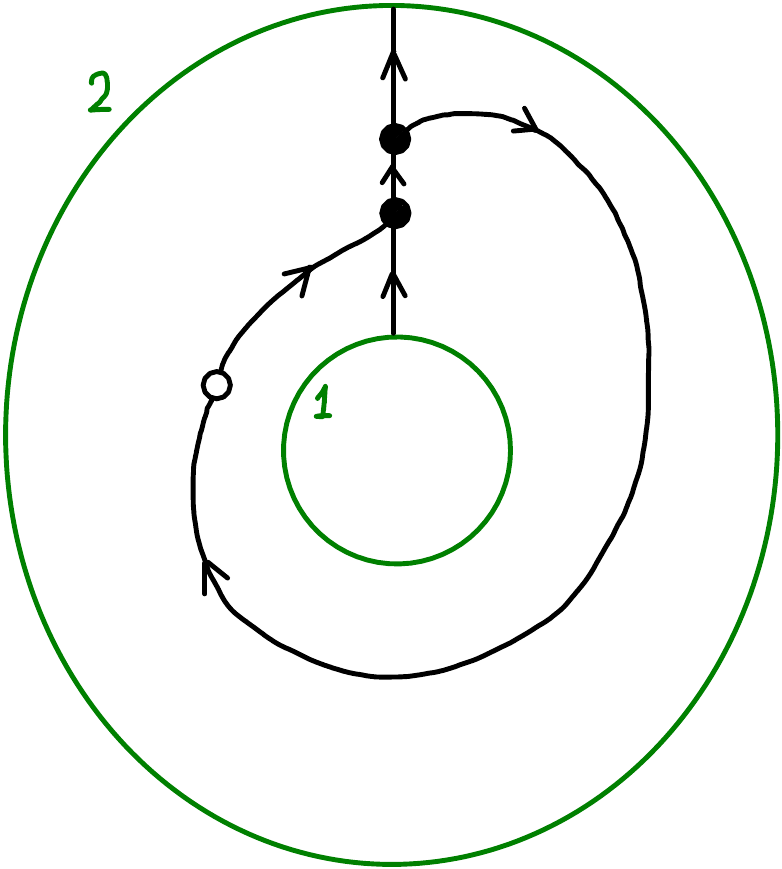}}~\Bigg) ~=~ Q(C_A) \circ (P^{NS} \otimes \id)
	~~.
\end{equation}

\begin{lemma}\label{lem:N+P-properties}
\begin{enumerate}
\item $N_\mathcal{H}$ is an involution.
\item $P^{NS}$ and $P^R$ are projectors.
\item For $\delta \in \{NS,R\}$, $N_\mathcal{H_A} \circ P^\delta = P^\delta \circ N_\mathcal{H_A}$.
\end{enumerate}
\end{lemma}

\begin{proof} Part 1 is immediate from Lemma \ref{lem:T[]-properties}\,(2). For part 2 compute
\begin{equation}\label{eq:PNS-idem-calc}
P^{NS} \circ P^{NS} 
~\overset{(1)}{=}~
T\!\left[~\raisebox{-0.5\height}{\includegraphics[scale=0.4]{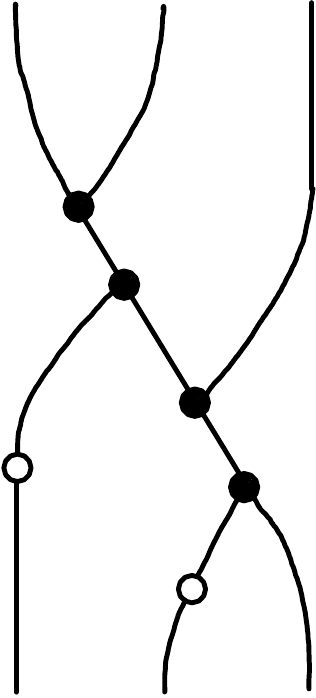}}~\right]
~\overset{(2)}{=}~
T\!\left[~\raisebox{-0.5\height}{\includegraphics[scale=0.4]{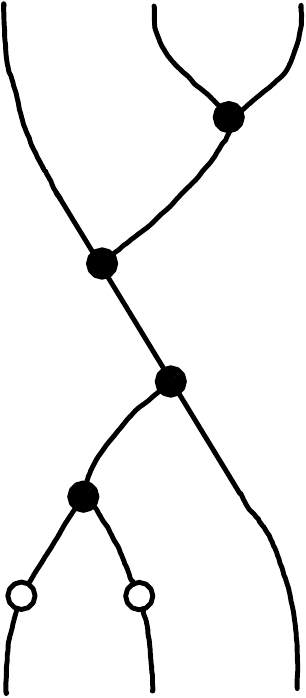}}~\right]
~\overset{(3)}{=}~
T\!\left[~\raisebox{-0.5\height}{\includegraphics[scale=0.4]{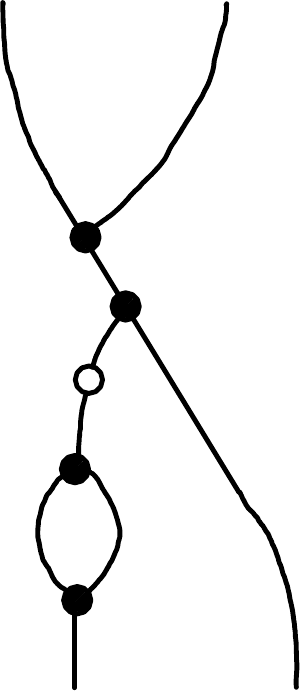}}~\right]
~\overset{(4)}{=}~
P^{NS} \ .
\end{equation}
Step 1 is Lemma \ref{lem:T[]-properties}\,(2), step 2 is associativity, coassociativity and the Frobenius property of $A$, step 3 is Lemma
\ref{lem:T[]-properties}\,(1) applied to the coproduct, and in step 4 $\Delta$-separability of $A$ is used and the definition of $P^{NS}$ is substituted.
The calculation for $P^R$ is similar. 
Part 3 for $\delta=NS$ follows from
\begin{align}
N_\mathcal{H_A} \circ P^{NS}
&\overset{(1)}{=}
T\big[\,(N \otimes \id)\circ  \Delta \circ \mu \circ (N \otimes \id)\, \big]
\overset{(2)}{=}
T\big[\,(N \otimes N) \circ \Delta \circ \mu\, \big]
\\ \nonumber
&\overset{(3)}{=}
T\big[\, \Delta \circ \mu \circ (N \otimes N)\, \big]
\overset{(4)}{=}
 P^{NS}  \circ N_\mathcal{H_A}
\end{align}
Steps 1 and 2 are Lemma \ref{lem:T[]-properties}, parts (2) and (1), respectively. Step 3 follows as $N$ is an automorphism of the Frobenius algebra $A$, and step 4 is as step 1. The argument for $P^R$ is the same.
\end{proof}

We can now define the quantum field theory on spin surfaces $Q_\mathrm{spin}$ in terms of $Q$. Let $\Sigma$ be a spin surface and $\Sigma^d$ a corresponding defect surface obtained by steps (a)--(c) in Section \ref{sec:spin-surf-to-def-surf}. We set
\begin{equation}\label{eq:Qspin-def}
	Q_\mathrm{spin}(\Sigma) := Q(\Sigma^d) \ .
\end{equation}
By  Theorem \ref{thm:moves_vs_alg}\,(2) and Proposition \ref{prop:Qspin-indep-choices}, $Q_\mathrm{spin}(\Sigma)$ is independent of the choices made in obtaining $\Sigma^d$.

\begin{lemma}\label{lem:Qspin-insert-P}
Let $\Sigma$ be a spin surface and let $\delta_i \in \{NS,R\}$ be the type of the $i$'th boundary component of $\Sigma$. Then, for $\psi_i \in \mathcal{H}_i$, $i=1,\dots,B$, we have
\begin{equation}
Q_\mathrm{spin}(\Sigma)(\psi_1 \otimes \cdots \otimes \psi_B)
=
Q_\mathrm{spin}(\Sigma)(P^{\delta_1}(\psi_1) \otimes \cdots \otimes P^{\delta_B}(\psi_B))
\end{equation}
\end{lemma}

The proof uses invariance of $Q_\mathrm{spin}$ under the choice of triangulation and is analogous to that of 
\cite[Prop.\,4.14]{Novak:2014oca}; 
we omit the details. 

The above lemma shows that $Q_\mathrm{spin}$ will in general be degenerate, but that it can easily be made non-degenerate by restricting to the image of $P^{NS/R}$. In more detail, one proceeds as follows. Define
\begin{equation}
\mathcal{H}^\delta := \mathrm{im}(P^\delta) \subset \mathcal{H}
\quad \text{where} \quad \delta \in \{NS,R\} \ .
\end{equation}

\begin{lemma}\label{lem:Qspin-nondeg}
Let $\delta \in \{NS,R\}$ and 
let $\Lambda^\delta$ be a spin cylinder with boundary components of type $\delta$. Then $Q_\mathrm{spin}(\Lambda^\delta) : \mathcal{H}^\delta_1 \otimes \mathcal{H}^\delta_2 \to \Cb$ is non-degenerate.
\end{lemma}

\begin{proof}
Let $C$ be a surface with defects such that $\underline C = \underline \Lambda^\delta$. On $C$ place the following defect network:
\begin{equation}
  C = \raisebox{-0.5\height}{\includegraphics[scale=0.3]{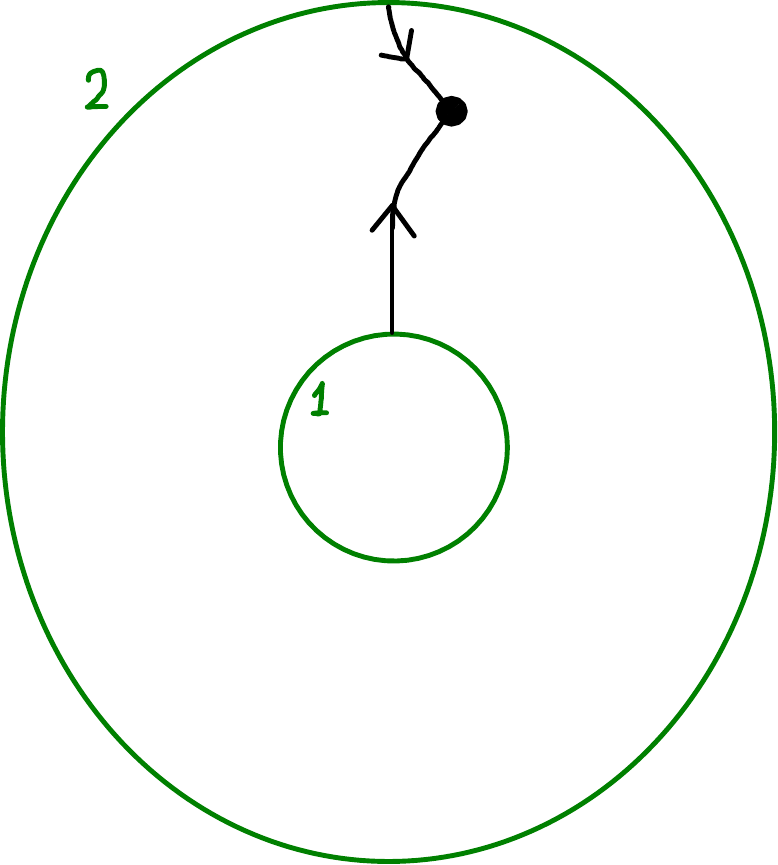}}
\end{equation}
By Assumption 4 we can write $Q(C)(\phi \otimes \psi) = Q(C_{(A,-)})(T[f]\phi \otimes \psi)$ for an appropriate $f : \one \otimes A \to A \otimes \one$. Using $(b \otimes \id_A) \circ (\id_A \otimes c) = \id_A = (\id_A \otimes c) \circ (b \otimes \id_A)$ together with Lemma \ref{lem:T[]-properties}\,(2), one checks that there is a $g : \one \otimes A \to A \otimes \one$ such that $T[f] \circ T[g] = \id_{\mathcal{H}_A} = T[g] \circ T[f]$. Since $Q(C_{(A,-)})$ is a non-degenerate pairing (Assumption 3) and $T[f]$ is invertible, also $(\phi,\psi) := Q(C)(\phi \otimes \psi)$ is non-degenerate.

Now proceed via steps (a)--(c) in Section \ref{sec:spin-surf-to-def-surf} to construct a cylinder with defect network encoding the spin structure of $\Lambda^\delta$. An analogous calculation has been carried out in \cite[Sect.\,4.5]{Novak:2014oca}. 
One finds that the resulting defect network can be simplified as shown in the first equality below -- we omit the details:
\begin{equation}
Q_\mathrm{spin}(\Lambda^\delta)(\phi \otimes \psi) = Q\Bigg(~\raisebox{-0.5\height}{\includegraphics[scale=0.3]{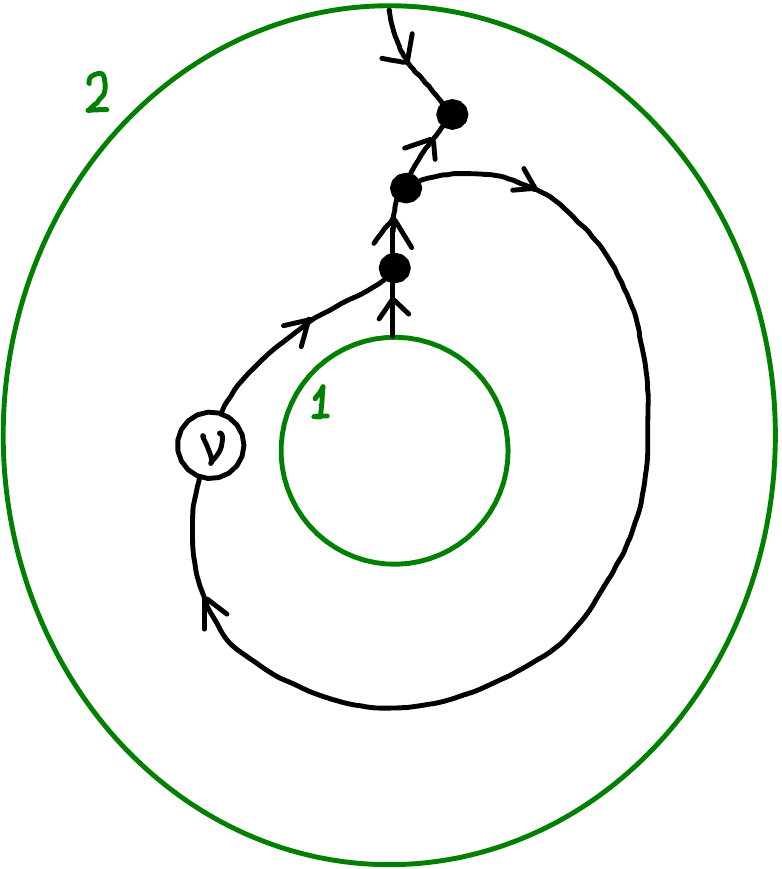}}~\Bigg)(\phi \otimes \psi) = (P^\delta \phi,\psi) \ ,
\end{equation}
where $\nu = -1$ for $\delta = NS$ and $\nu=1$ for $\delta = R$.
A short calculation shows that $(P^\delta \phi,\psi) = (\phi,P^\delta\psi)$ for all $\phi,\psi$, and so the pairing $(-,-)$ remains non-degenerate when restricted to $\mathcal{H}^\delta$. But for $\phi,\psi \in \mathcal{H}^\delta$ one has $Q_\mathrm{spin}(\Lambda^\delta)(\phi \otimes \psi) = (\phi,\psi)$, which is therefore non-degenerate.
\end{proof}

Lemmas \ref{lem:Qspin-insert-P} and \ref{lem:Qspin-nondeg} show that, apart from the glueing property {\bf Q4} which was not discussed,\footnote{
{\bf Q4} for $Q_\mathrm{spin}$ would follow from {\bf Q4} for Q and the behaviour of the combinatorial model for the spin structure under glueing as described in \cite[Sect.\,3.5]{Novak:2014oca}
see also \cite[Sect.\,4.4]{Novak:2014oca}, 
where the glueing procedure is described for the spin state sum TFT.
}, 
\begin{equation}
	Q_\mathrm{spin}(\Sigma) : \mathcal{H}_1^{\delta_1} \otimes \cdots \otimes \mathcal{H}_B^{\delta_B} \to \Cb 
\end{equation}
is a non-degenerate QFT and that one does not loose any information by restricting the arguments from $\mathcal{H}_i$ to $\mathcal{H}^{NS/R}_i$. 

In the context of 2d\,TFTs, state spaces defined in terms of projectors similar to $P^{NS/R}$ from \eqref{eq:PNS/R-def-via-T} appear in \cite{Brunner:2013ota,barrett2013spin,Novak:2014oca}.

\medskip

By Lemma \ref{lem:N+P-properties}\,(3), the involution $N_\mathcal{H}$ maps $\mathcal{H}^{NS}$ to itself, and dito for $\mathcal{H}^R$. We denote these restrictions of $N_\mathcal{H}$ by $N|_{NS}$ and $N|_R$, respectively. 
We stress that we have not required the Nakayama automorphism of $A$ to be parity involution. Accordingly, also the involutions $N|_{NS}$ and $N|_R$ do not have to be the parity involution. For theories with fermions, the parity involution is usually denoted by $(-1)^F$, so that in other words, in the formalism developed here it is not imposed that $N|_{NS}$ and $N|_R$ act as $(-1)^F$.
(Of course, the free fermion example discussed in Sections \ref{sec:1ff-ex} and \ref{sec:so(n)} will have this property.)

\medskip

Finally, we consider the action of Dehn twists on the state spaces $\mathcal{H}^{NS/R}$. 
We will need
\begin{quote}
{\bf Assumption:} 5) 
Let $\underline d$ be a sequence of defect conditions and let $\mathcal{H}_{\underline d}$ be the state space assigned to a boundary component with rotationally symmetric parametrisation (i.e.\ the metric on the parametrising annulus in $\mathbb{C}$ is invariant under rotations).
Then there is an action $J$ of infinitesimal rotations on $\mathcal{H}_{\underline d}$ such that
\begin{equation} \label{eq:ass-rotation-action}
Q\Bigg(~\raisebox{-0.5\height}{\includegraphics[scale=0.3]{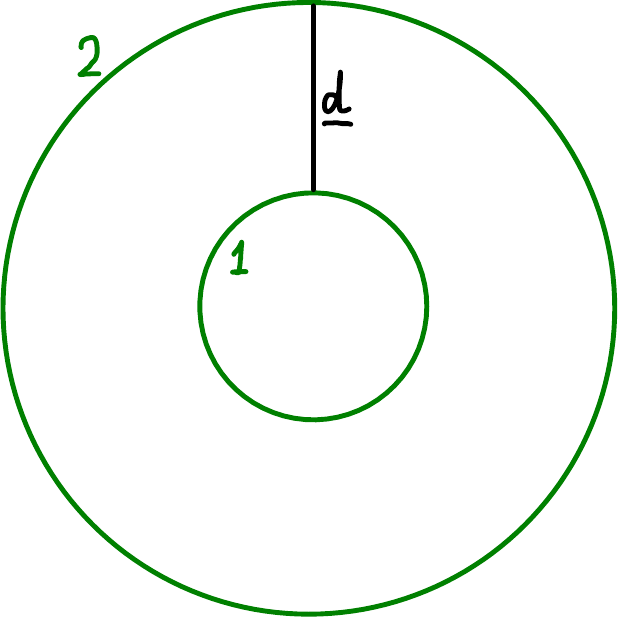}}~\Bigg)((e^{2 \pi i J} \psi) \otimes \xi)
=
Q\Bigg(~\raisebox{-0.5\height}{\includegraphics[scale=0.3]{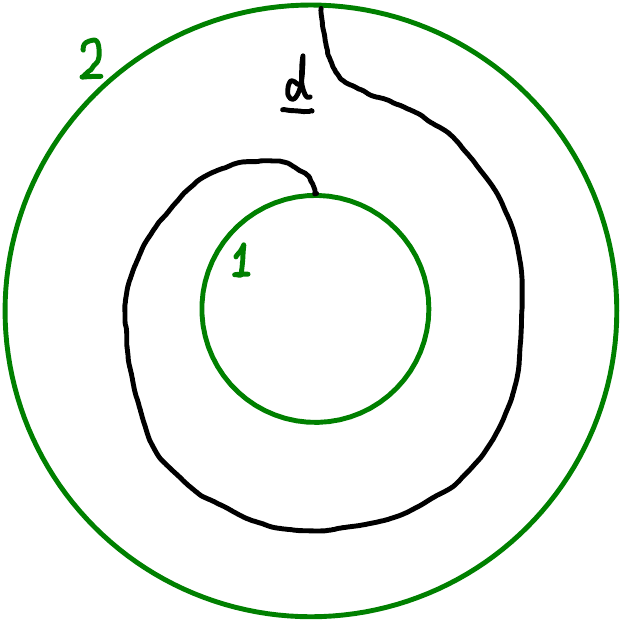}}~\Bigg)(\psi \otimes \xi) \ .
\end{equation}
\end{quote}
In particular, if $\underline d=()$, i.e.\ if $\underline d$ is the trivial defect, then $e^{2 \pi i J}$ is the identity map. 
Furthermore, $e^{2 \pi i J}=\id$ on the subspace of topological junctions -- this relates to pivotality of $\mathcal{D}$.
In terms of the assignment $T$ from Assumption 4, \eqref{eq:ass-rotation-action} can be written as
\begin{equation}
	T[ \id_{\underline d \otimes \underline d} ] = e^{2 \pi i J} \ .
\end{equation}

\begin{lemma}\label{lem:dehn-twist-vs-N}
On $\mathcal{H}^{NS}$ we have $e^{2 \pi i J} = N|_{NS}$. On $\mathcal{H}^{R}$ we have $e^{2 \pi i J} = \id$.
\end{lemma}

\begin{proof}
Geometrically, this is just the statement that a full rotation amounts to sheet exchange in the boundary parametrisation for an $NS$-type boundary and to the identity for an $R$-type boundary. In terms of the combinatorial model, this is described as follows. 

Consider the marked triangulation with edge signs of a cylinder as in \cite[Sect.\,4.5]{Novak:2014oca}, 
\begin{equation}\label{eq:dehntwist-aux1}
\raisebox{-0.5\height}{\includegraphics[scale=0.3]{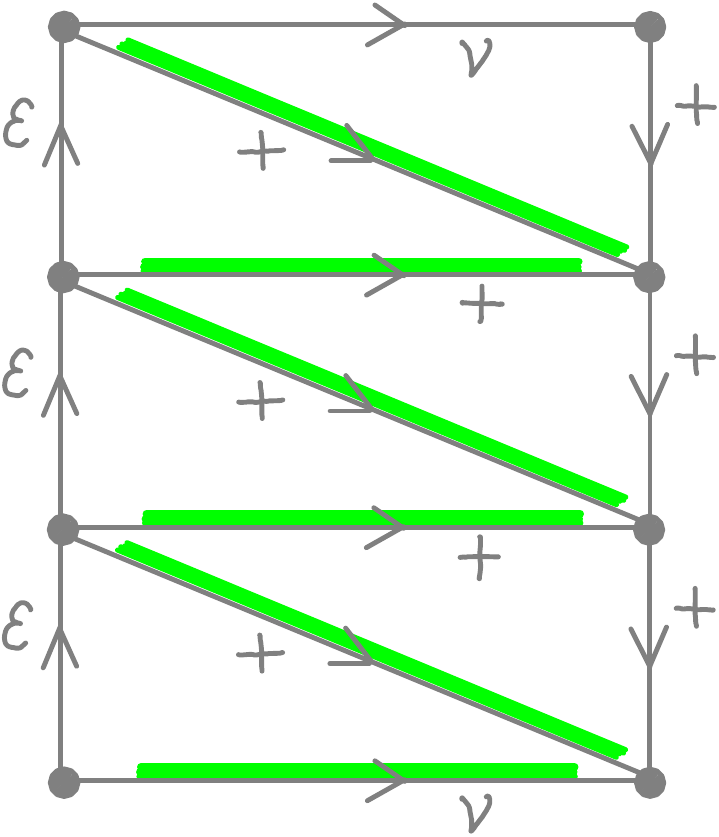}}
\end{equation}
Here, the top and bottom edges are identified, and $\nu=-1$ if both boundaries are of $NS$-type and $\nu=1$ if both boundaries are of $R$-type. According to Assumption 5, acting with $e^{2 \pi i J}$ on the state space amounts to deforming (the defect network corresponding to) the triangulation above to (the defect network corresponding to) the triangulation below:
\begin{equation}\label{eq:dehntwist-aux2}
\raisebox{-0.5\height}{\includegraphics[scale=0.3]{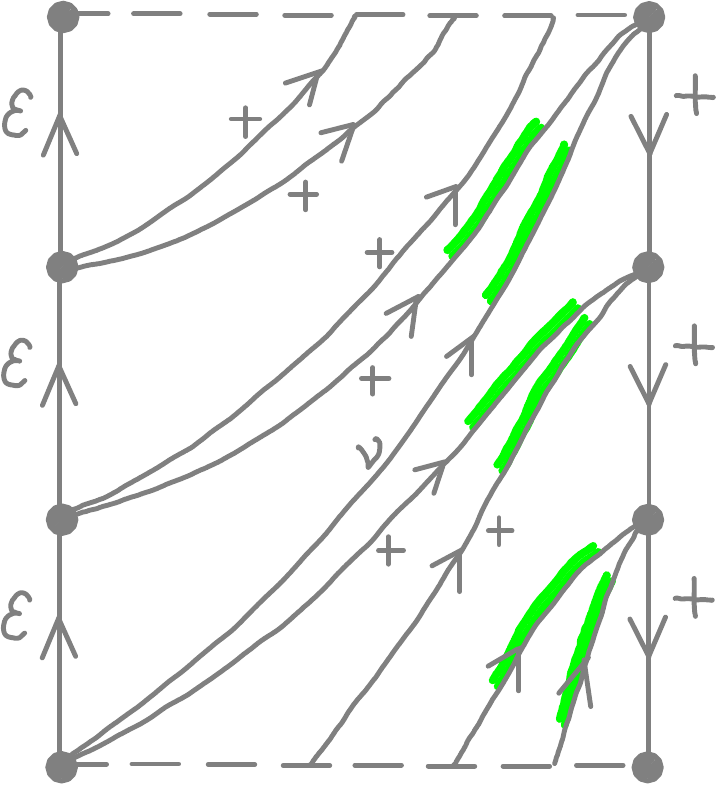}}
\end{equation}
The triangulation can be brought back to the form \eqref{eq:dehntwist-aux1} via 
the moves (1)--(5) listed in Section \ref{sec:spin-surf-to-def-surf}. One way to do this explicitly is to iterate the following transformation three times:
\begin{equation}
\raisebox{-0.5\height}{\includegraphics[scale=0.3]{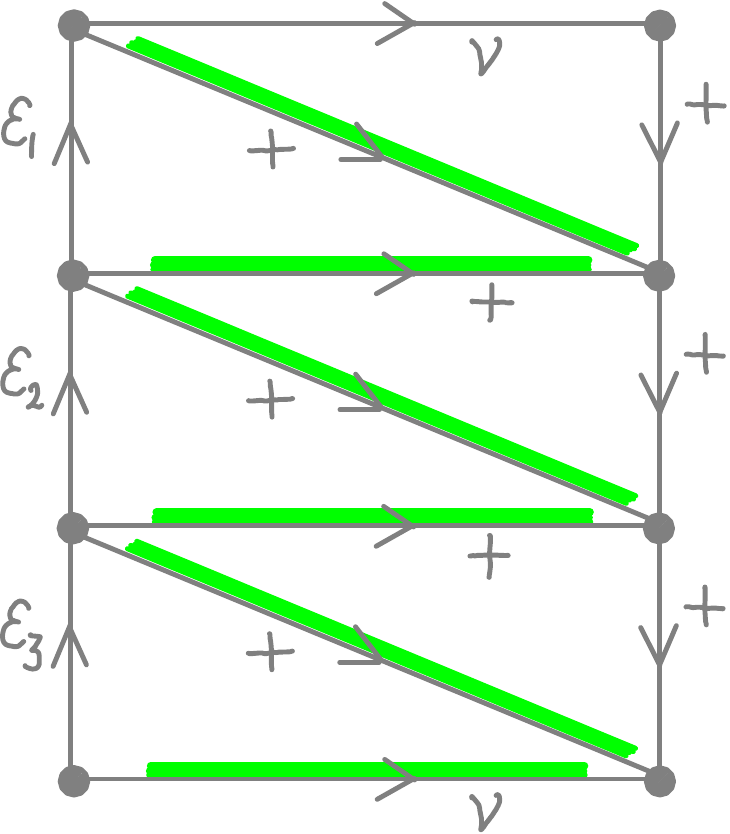}}
\overset{\text{$\frac13$ shift}}
\longmapsto
\raisebox{-0.5\height}{\includegraphics[scale=0.3]{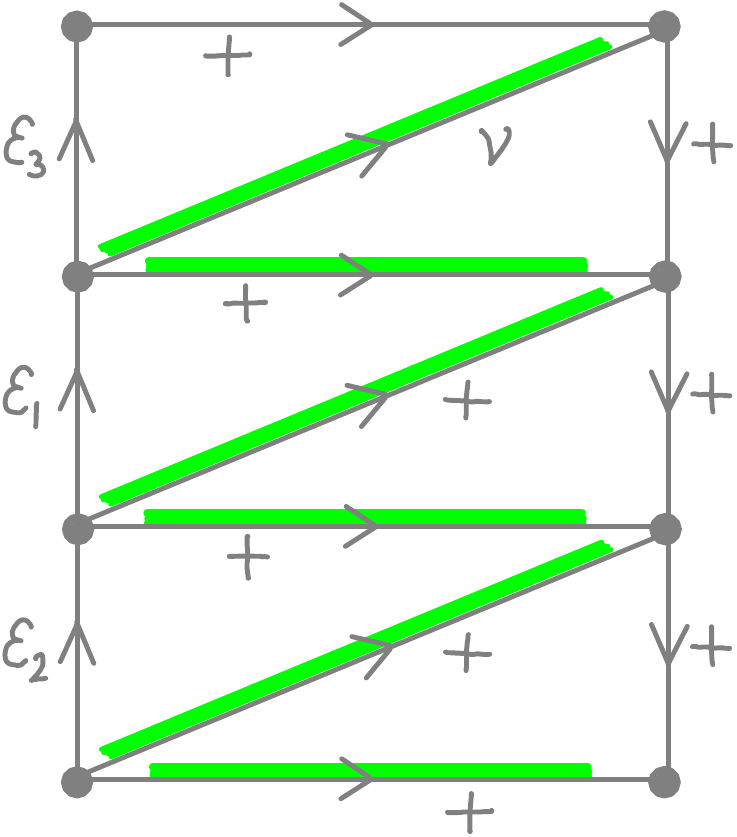}}
\overset{(*)}
\cong
\raisebox{-0.5\height}{\includegraphics[scale=0.3]{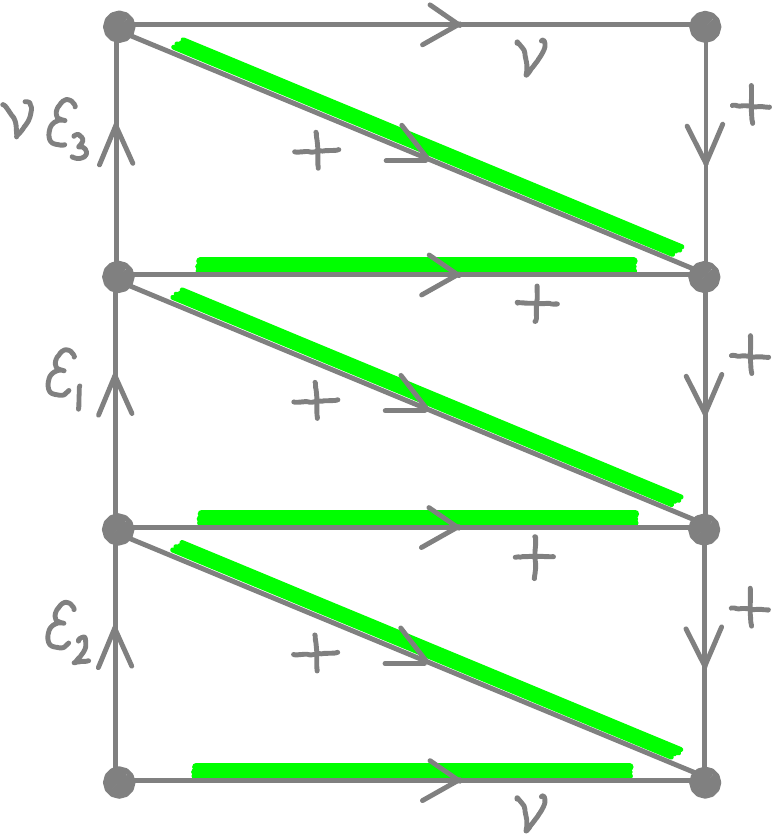}}
\ .
\end{equation}
By ``$\frac13$ shift'' we mean rotating the boundary circle by $2 \pi /3$, which deforms the triangulation as shown. In (*) the variant of the 2-2 Pachner move from 
\cite[Prop.\,3.19]{Novak:2014oca}
shown below is used three times, and move (2) ``leave exchange'' is applied to one of the triangles:
\begin{equation}
\raisebox{-0.5\height}{\includegraphics[scale=0.3]{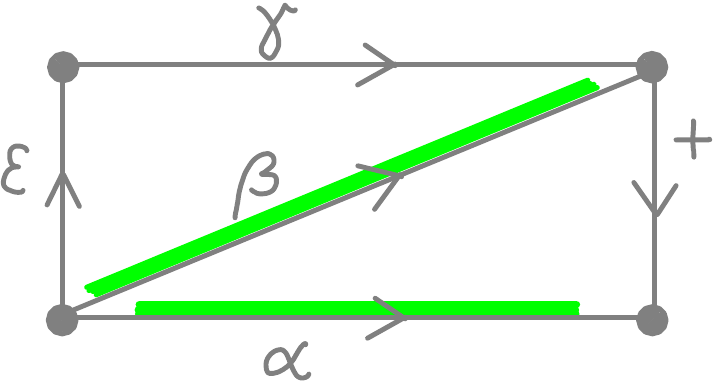}}
\quad \longleftrightarrow \quad
\raisebox{-0.5\height}{\includegraphics[scale=0.3]{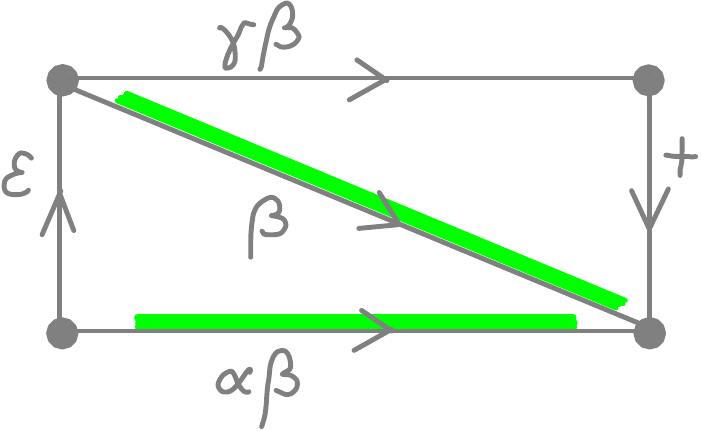}}
\ .
\end{equation}
The overall effect of iterating the above step three times is to replace $\varepsilon$ by $\nu \varepsilon$. This proves the statement of the lemma.
\end{proof}

\section{Application to two-dimensional rational conformal field theory}\label{sec:RCFT}

In this section we apply the general considerations above to two-dimensional rational conformal field theories (CFTs). In these theories one has good control over the monoidal category of topological defects which preserve the rational symmetry. This allows one to make the constructions of Section \ref{sec:amp-spin-def} explicit.

\subsection{Topological defects in rational conformal field theory}

Fix a chiral symmetry algebra, that is, a vertex operator algebra $\mathcal{V}$ (which is not $\mathbb{Z}_2$-graded, i.e.\ not a vertex operator super algebra). We require $\mathcal{V}$ to be rational in the sense that its representation category $\Rep\mathcal{V}$ is modular \cite{Moore:1988qv,Huang2005}. This means that $\Rep\mathcal{V}$ is a $\mathbb{C}$-linear finitely semi-simple abelian ribbon category with simple tensor unit and with a non-degeneracy condition on the braiding, see e.g.\ \cite{baki} for details.

We will say that a CFT {\em has symmetry} $\mathcal{V} \otimes_{\mathbb{C}} \bar{\mathcal{V}}$ if a copy of $\mathcal{V}$ is contained in the holomorphic fields and in the anti-holomorphic fields of the CFT. The overline in the notation $\mathcal{V} \otimes_{\mathbb{C}} \bar{\mathcal{V}}$ signifies that one copy of $\mathcal{V}$ is embedded in the anti-holomorphic fields.

One approach to study the properties of rational CFT is via the relation to three-dimensional topological field theories \cite{Felder:1999mq,tft1}. 
In this approach, CFTs with symmetry $\mathcal{V} \otimes_{\mathbb{C}} \bar{\mathcal{V}}$ which contain the vacuum representation with multiplicity one are classified by Morita classes of simple symmetric $\Delta$-separable Frobenius algebras in $\Rep\mathcal{V}$, see \cite{Fjelstad:2006aw} and \cite[Thm.\,1.1]{Kong:2007yv}, \cite[Thm.\,3.22]{Kong:2008ci}. 

The three-dimensional TFT approach to CFT with defects has been developed in detail in \cite{frohlich2009defect,fjelstad2012rcft}. Let $F$ be an algebra as above.
Topological defects of the CFT described by $F$ which are transparent not only to the stress tensor, but to all fields in $\mathcal{V} \otimes_{\mathbb{C}} \bar{\mathcal{V}}$, have an easy representation theoretic interpretation: the defect category $\mathcal{D}$ for such defects is the monoidal category of $F$-$F$-bimodules in $\Rep\mathcal{V}$ together with bimodule intertwiners. 

\medskip

The situation simplifies significantly if one chooses $F=\mathbf{1}:= \mathcal{V}$, the monoidal unit in $\Rep\mathcal{V}$. This is known as the {\em Cardy case}. Then simply $\mathcal{D} =\Rep\mathcal{V}$, so that in this special case $\mathcal{D}$ is equipped with a braiding and a twist, i.e.\ it is a ribbon category. Let us describe the state spaces in this simplified setting (see \cite{frohlich2009defect} for the general case). Denote by $\{S_i \,|\, i\,{\in}\,\mathcal{I}\}$ a choice of representatives of the isomorphism classes of simple objects in $\Rep\mathcal{V}$.  For $X \in \mathcal{D}$ we have
\begin{equation}\label{eq:state-space-ungraded-Cardy}
	\mathcal{H}_{X} \equiv \mathcal{H}_{(X,+)} = \bigoplus_{i,j \in \mathcal{I}} 
	\mathcal{D}(S_i \otimes S_j,X) \otimes_{\mathbb{C}} S_i \otimes_\mathbb{C} \bar S_j \ .
\end{equation}
Let us explain this equation in more detail. The left hand side, as always, stands for the state space of a boundary circle with a defect line labelled $X$ which is oriented away from the boundary circle. On the right hand side, $S_i \otimes_\mathbb{C} \bar S_j$ is an irreducible representation of the holomorphic and anti-holomorphic symmetry algebra $\mathcal{V} \otimes_{\mathbb{C}} \bar{\mathcal{V}}$. 
Finally, $\mathcal{D}(S_i \otimes S_j,X)$ is the $\Cb$-vector space of intertwiners from $S_i \otimes S_j$ to $X$ in $\mathcal{D}=\Rep\mathcal{V}$; the ``$\otimes$'' in $S_i \otimes S_j$ stands from the fusion tensor product of $\mathcal{V}$-modules.

As a special case, consider the defect labelled by the monoidal unit $\mathbf{1}$. This is the trivial defect, and so $\mathcal{H}_\one$ is the space of bulk fields. Indeed, inserting $X=\one$ in the above formula produces $\bigoplus_{i \in \mathcal{I}} S_i \otimes_{\mathbb{C}} \bar S_{i^*}$, where $i^* \in \mathcal{I}$ is the unique index such that $S_{i^*}$ is isomorphic to the dual $(S_i)^*$. This is the space of bulk fields in the Cardy case.

We will write $Q_{\mathcal{V}}$ for the Cardy case conformal field theory for the symmetry $\mathcal{V} \otimes_{\mathbb{C}} \bar{\mathcal{V}}$, and $\mathcal{D}_{\mathcal{V}} = \Rep\mathcal{V}$ for its symmetry-preserving topological defects. The QFT $Q_{\mathcal{V}}$ satisfies {\bf Q1}--{\bf Q3}, and {\bf Q4} for surfaces of genus 0 and 1,\footnote{
More precisely, these results should follow by extending the methods of \cite{Huang:2005gz,Kong2006c} to the present setting with defects by combining it with the TFT description in  \cite{frohlich2009defect}. However, the details have so far not been worked out.
} 
and for all surfaces if suitable factorisation and monodromy properties of the conformal blocks are assumed \cite{frohlich2009defect,fjelstad2012rcft}. $Q_{\mathcal{V}}$ also satisfies Assumptions 3 and 4.

\medskip

The assignment $T$ from Assumption 4 can easily be given explicitly. Recall the cylinders in \eqref{eq:cylinder-junction-to-map} and the definition of $T[f] : \mathcal{H}_X \to \mathcal{H}_Y$ in \eqref{eq:T[f]-def}.
Since the defect conditions in $\mathcal{D}_\mathcal{V}$ describe defects transparent to $\mathcal{V} \otimes_{\mathbb{C}} \bar{\mathcal{V}}$, $T[f]$ must be of the form
\begin{equation}\label{eq:defect-wrap-operator-ungraded}
	T[f] = \bigoplus_{i,j \in \mathcal{I}} 
	T[f]_{ij} \otimes_{\mathbb{C}} \id_{S_i \otimes_\mathbb{C} \bar S_j} \ ,
\end{equation}
where $T[f]_{ij}$ is a linear map from $\mathcal{D}_\mathcal{V}(S_i \otimes S_j,X)$ to $\mathcal{D}_\mathcal{V}(S_i \otimes S_j,Y)$. The map $T[f]_{ij}$ can be computed via the methods in \cite{frohlich2009defect} to be, for $h \in \mathcal{D}_\mathcal{V}(S_i \otimes S_j,X)$,
\begin{equation}\label{eq:ungraded-defect-operator-component}
	T[f]_{ij}(h) = \raisebox{-0.5\height}{\includegraphics[scale=0.4]{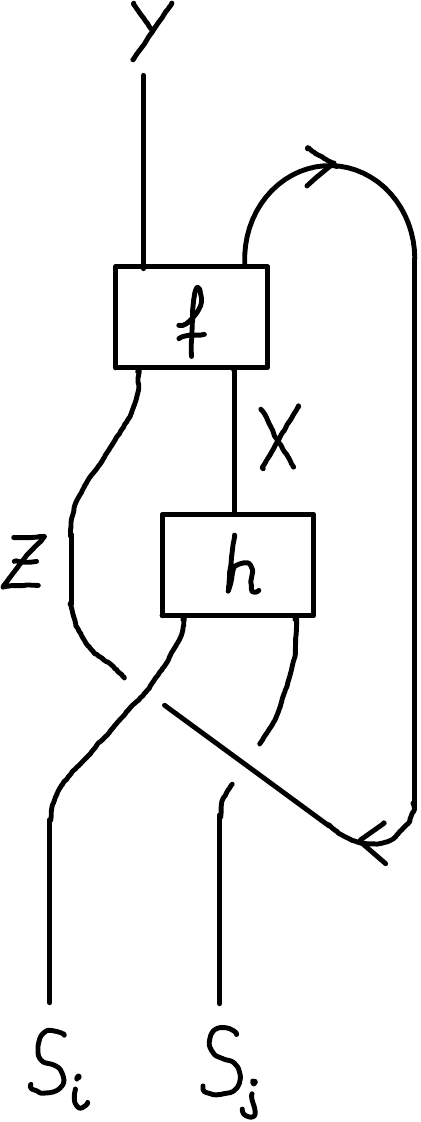}} \quad .
\end{equation}
The diagram describes the morphism $T[f]_{ij}(h) : S_i \otimes S_j \to Y$ via the standard graphical notation in ribbon categories, see e.g.\ \cite{baki} or see \cite[Sect.\,2.1]{tft1} for our precise conventions.

\subsection{Spin theories from defects in the graded case}\label{sec:spin-def-graded}

We will now describe conformal field theories on spin surfaces starting from the product defect category $\mathcal{D} \boxtimes \mathcal{D}^{SV}$. Here, $\mathcal{D}$ is the category of $F$-$F$-bimodules in $\Rep\mathcal{V}$, as described above. According to Section \ref{sec:alg-desc}, to describe a ($\mathbb{Z}_2$-graded) CFT on spin surfaces we need to identify a $\Delta$-separable Frobenius algebra $A$ in $\mathcal{D} \boxtimes \mathcal{D}^{SV}$ whose Nakayama automorphism is an involution.

By an argument similar to the ``orbifolding twice can be combined into a single orbifold'' statement in \cite[Sect.\,4.2]{carqueville2012orbifold}, one can check that without loss of generality one may choose $F = \one$, since the choice of $F$ can be absorbed into the choice of $A$. 
Thus we choose as our starting point the product of $Q_{\mathcal{V}}$, the Cardy case conformal field theory for the symmetry $\mathcal{V} \otimes_{\mathbb{C}} \bar{\mathcal{V}}$,
with the topological theory $Q^{SV}$. We write $\hat Q$ and $\hat{\mathcal{H}}$ for the amplitudes and state spaces of the product theory.
The defect category of the product theory is
\begin{equation}
\hat{\mathcal{D}} := \mathcal{D}_{\mathcal{V}} \boxtimes \mathcal{D}^{SV} = \Rep\mathcal{V} \boxtimes \mathbf{SVect}^{fd} \ .
\end{equation}
The category $\hat{\mathcal{D}}$ is a $\mathbb{C}$-linear additive pivotal monoidal category as we required in Section \ref{sec:def-surf+ampl}. Since we are in the Cardy case,  $\hat{\mathcal{D}}$ is in addition ribbon (i.e.\ braided with a twist). Our choice of ribbon structure on $\mathbf{SVect}^{fd}$ is such that $\theta_V = \id_V$ for all $V \in \mathbf{SVect}^{fd}$.

For the simple objects in $\mathbf{SVect}^{fd}$ we choose the representatives
\begin{equation}
	K_+ := \mathbb{C}^{1|0} \quad , \qquad
	K_- := \mathbb{C}^{0|1} \ .
\end{equation}
We will use the following notations for objects in $\hat{\mathcal{D}}$: An object which is just the product of two factors (rather than a direct sum of such) is written as $U \boxtimes V$, $U \in \Rep\mathcal{V}$, $V \in \mathbf{SVect}^{fd}$; for $V$ being $K_\pm$ we use the shorthands $U := U \boxtimes K_+$ and $\Pi U := U \boxtimes K_-$. We will also write $K_\pm := \one \boxtimes K_\pm$.
We choose $\{ S_i , \Pi S_i | i \in \mathcal{I} \}$ as representatives of the isomorphism classes of simple objects in $\hat{\mathcal{D}}$.

As before, $\Pi$ also denotes the parity flip on $\mathbf{SVect}$. For example, in the following lemma ``$\Pi(S_i \otimes S_j)$'' stands for the object $(S_i \otimes S_j) \boxtimes K_-$ of $\hat{\mathcal{D}}$, while $\Pi(S_i \otimes_\mathbb{C} \bar S_j)$ stands for the super vector space with $\{0\}$ as even component and the irreducible $\mathcal{V} \otimes_{\mathbb{C}} \bar{\mathcal{V}}$-module $S_i \otimes_\mathbb{C} \bar S_j$ as odd component.

\begin{lemma}
Given the state spaces \eqref{eq:state-space-ungraded-Cardy} of the Cardy case CFT, the state spaces of the product theory are, for $X \in \hat{\mathcal{D}}$,
\begin{align}
	\hat{\mathcal{H}}_X 
	\cong
	\bigoplus_{i,j \in \mathcal{I}} \Big(&
	\hat{\mathcal{D}}(S_i \otimes S_j,X) \otimes_{\mathbb{C}} S_i \otimes_\mathbb{C} \bar S_j 
\\[-1em]
	& \hspace{1em} \oplus ~
	\hat{\mathcal{D}}(\Pi(S_i \otimes S_j),X) \otimes_{\mathbb{C}} \Pi(S_i \otimes_\mathbb{C} \bar S_j) \Big) \ .
\nonumber
\end{align}
\end{lemma}

\begin{proof}
Any $X \in \hat{\mathcal{D}}$ is isomorphic to $U \oplus \Pi V$ for some $U,V \in \mathcal{D}_\mathcal{V}$. Then also 
$\hat{\mathcal{H}}_X \cong \hat{\mathcal{H}}_U \oplus \hat{\mathcal{H}}_{\Pi V}$.
 By definition of the product theory, $\hat{\mathcal{H}}_U = \mathcal{H}_U \otimes_{\mathbb{C}} \mathbb{C}^{1|0}$ and $\hat{\mathcal{H}}_{\Pi V} = \mathcal{H}_V \otimes_{\mathbb{C}} \mathbb{C}^{0|1}$. Since for $X =U \oplus \Pi V$ we have $\hat{\mathcal{D}}(S_i \otimes S_j,X) = \mathcal{D}_\mathcal{V}(S_i \otimes S_j,U)$ and 
$\hat{\mathcal{D}}(\Pi(S_i \otimes S_j),X) = \mathcal{D}_\mathcal{V}(S_i \otimes S_j,V)$ we obtain the statement of the lemma.
\end{proof}

The action \eqref{eq:defect-wrap-operator-ungraded} of defects on state spaces is transported to the graded case as follows.

\begin{lemma}\label{lem:Dhat-via-stringdiag}
Let $X,Y,Z \in \hat{\mathcal{D}}$ and $f : Z \otimes X \to Y \otimes Z$. Then $T[f] : \hat{\mathcal{H}}_X \to \hat{\mathcal{H}}_Y$ is given by
\begin{equation}\label{eq:defect-wrap-operator-graded}
	T[f] = \bigoplus_{i,j \in \mathcal{I}} \Big(
	T[f]^+_{ij} \!\otimes_{\mathbb{C}} \id_{S_i \otimes_\mathbb{C} \bar S_j} 
	~\oplus~
	T[f]^-_{ij} \!\otimes_{\mathbb{C}} \id_{\Pi(S_i \otimes_\mathbb{C} \bar S_j)} \Big) \ ,
\end{equation}
where $T[f]^\alpha_{ij} : \hat{\mathcal{D}}\big((S_i \otimes S_j) \boxtimes K_\alpha,X\big) \to  \hat{\mathcal{D}}\big((S_i \otimes S_j) \boxtimes K_\alpha,Y\big)$ is
\begin{equation}\label{eq:graded-defect-operator-component}
	T[f]^\alpha_{ij}(h) = \raisebox{-0.5\height}{\includegraphics[scale=0.4]{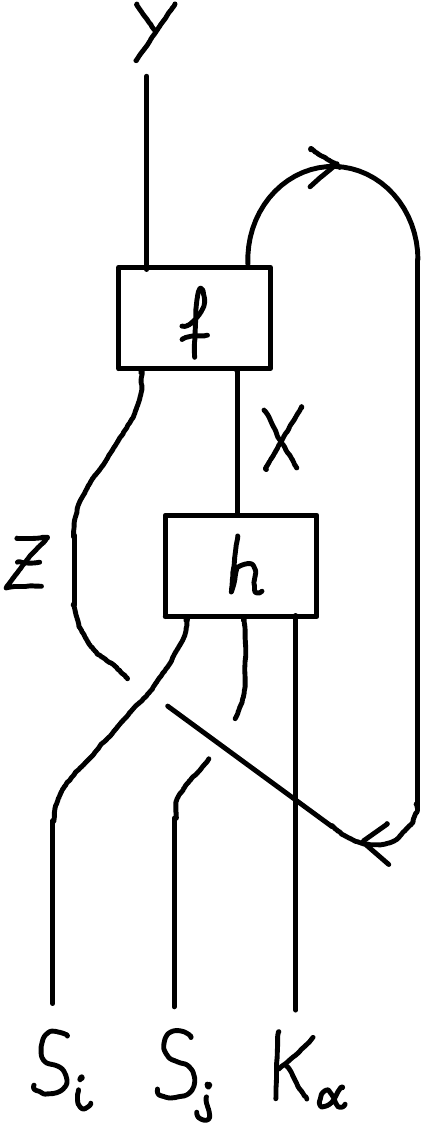}} \qquad ; ~~ \alpha = \pm1 \ .
\end{equation}
Since $K_\pm$ is transparent in $\hat{\mathcal{D}}$, we do not distinguish over- and under-crossings for strings labelled $K_\pm$.
\end{lemma}

\begin{proof}
It is sufficient to prove the formula for $T[f]$ in the case that all objects and morphisms factorise. That is, $X = X^{(1)} \boxtimes X^{(2)}$, with $X^{(1)} \in \Rep\mathcal{V}$ and $X^{(2)} \in \mathbf{SVect}^{fd}$, etc. By definition of the product theory in Section \ref{sec:QxSV}, in this case $\hat{\mathcal{H}}_X = \mathcal{H}_{X^{(1)}} \otimes_{\mathbb{C}} X^{(2)}$ and dito for $\hat{\mathcal{H}}_Y$. 
Let $C := C_{X,Y,Z}^f$ be the cylinder with defect network  in \eqref{eq:cylinder-junction-to-map} (with all labels from $\hat{\mathcal{D}}$), and let $C^{(1)}$ and $C^{(2)}$ be the corresponding cylinders with labels $X^{(1)}, f^{(1)},\dots$ and $X^{(2)}, f^{(2)},\dots$, respectively. 
Then, again by definition of the product theory, $\hat Q(C) = Q_\mathcal{V}(C^{(1)}) \otimes_{\mathbb{C}} Q^{SV}(C^{(2)})$. The induced map $T[f]$ is then equally a tensor product, $T[f] = T[f^{(1)}]  \otimes_{\mathbb{C}} T[f^{(2)}]$, where $T[f^{(1)}]$ is as in \eqref{eq:defect-wrap-operator-ungraded} (with all labels of the form $(-)^{(1)}$). 
To make contact with \eqref{eq:defect-wrap-operator-graded}, we need to rewrite $T[f^{(2)}]$ as an even linear map from $\bigoplus_{\alpha \in \{\pm\}} \mathrm{Hom}_{\mathbf{SVect}}(K_\alpha,X^{(2)}) \otimes_\mathbb{C} K_\alpha$ to $\bigoplus_{\alpha \in \{\pm\}} \mathrm{Hom}_{\mathbf{SVect}}(K_\alpha,Y^{(2)}) \otimes_\mathbb{C} K_\alpha$. That is, we set \begin{equation}
T[f^{(2)}] = \bigoplus_{\alpha \in \{\pm\}} T[f^{(2)}]_\alpha \otimes_\mathbb{C} \id_{K_\alpha} \ ,
\end{equation}
where, for $u \in \mathrm{Hom}_{\mathbf{SVect}}(K_\alpha,X^{(2)})$,
\begin{equation}
	T[f^{(2)}]_\alpha(u) ~=~ \raisebox{-0.5\height}{\includegraphics[scale=0.4]{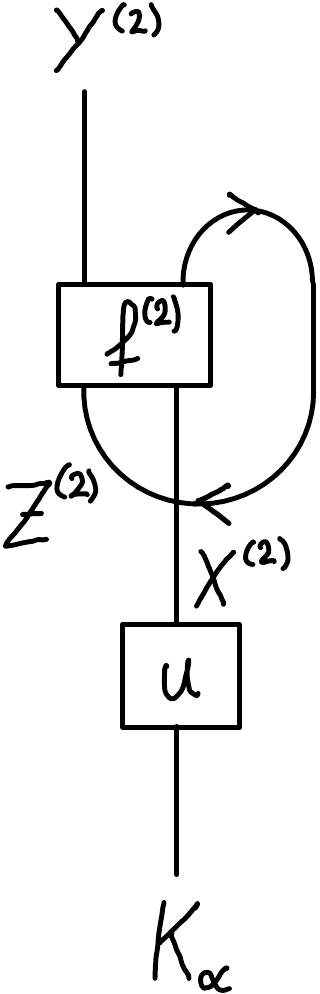}} \ .
\end{equation}
Altogether, for $h \in \hat{\mathcal{D}}\big((S_i \otimes S_j) \boxtimes K_\alpha,X\big)$ of the form $h = h^{(1)} \boxtimes h^{(2)}$,
\begin{align}
	T[f](h^{(1)} \boxtimes h^{(2)}) 
	&= \bigoplus_{i,j\in\mathcal{I}, \alpha \in \{\pm\}}
	T[f^{(1)}]_{ij}(h^{(1)}) \boxtimes T[f^{(2)}]_{\alpha}(h^{(2)})
	 \otimes_\mathbb{C} 
	 \id_{S_i \otimes_\mathbb{C} \overline S_j \otimes_\mathbb{C} K_\alpha} 
	\\ \nonumber 
	&= \bigoplus_{i,j\in\mathcal{I}, \alpha \in \{\pm\}}
	T[f]^\alpha_{ij}(h^{(1)} \boxtimes h^{(2)})
	 \otimes_\mathbb{C} 
	 \id_{S_i \otimes_\mathbb{C} \overline S_j \otimes_\mathbb{C} K_\alpha} \ .
\end{align}
\end{proof}

After these preliminaries, we turn to the analysis of state spaces in the spin theory. Let in the following $A$ be a $\Delta$-separable Frobenius algebra in $\hat{\mathcal{D}}$ whose Nakayama automorphism $N$ is an involution. 
For $\alpha,\nu \in \{ \pm 1\}$ and $i,j \in \mathcal{I}$, define the linear endomorphism
$Q_{\nu}^{i,j,\alpha}$ of $\hat{\mathcal{D}}(S_i \otimes S_j \otimes K_\alpha,A)$ as
\begin{equation}\label{eq:Q-string-diag}
	Q_{\nu}^{i,j,\alpha} = \raisebox{-0.5\height}{\includegraphics[scale=0.4]{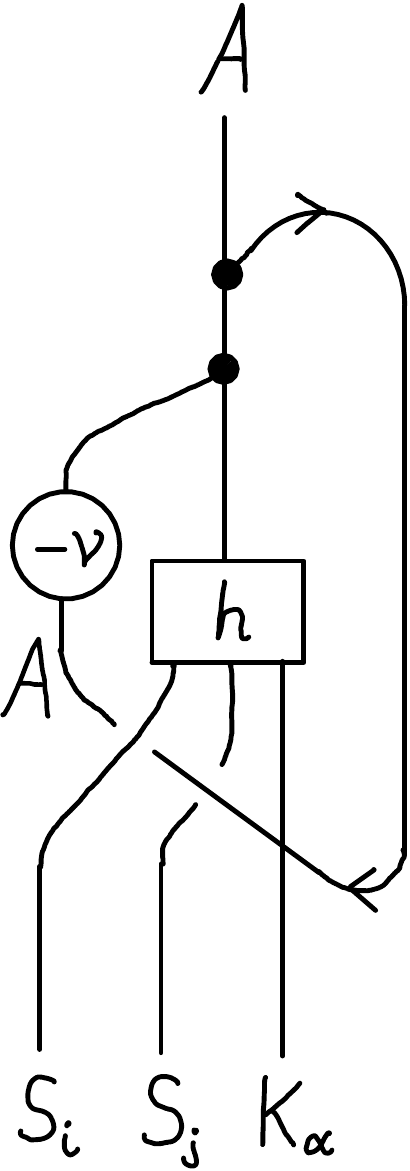}} \quad .
\end{equation}
One checks in a calculation analogous to \eqref{eq:PNS-idem-calc} that $Q_{\nu}^{i,j,\alpha}$ is a projector,
$\big(Q_{\nu}^{i,j,\alpha}\big)^2 = Q_{\nu}^{i,j,\alpha}$.
Applying Lemma \ref{lem:Dhat-via-stringdiag} to the projectors $P^{NS/R}$ on $\hat{\mathcal{H}}_A$ defined in Section \ref{sec:state-spaces} shows that
\begin{equation}\label{eq:PNSR-CFT-graded}
	P^{NS/R} = \bigoplus_{i,j \in \mathcal{I}} \Big(
	Q_{\nu}^{i,j,+} \!\otimes_{\mathbb{C}} \id_{S_i \otimes_\mathbb{C} \bar S_j} 
	~\oplus~
	Q_{\nu}^{i,j,-} \!\otimes_{\mathbb{C}} \id_{\Pi(S_i \otimes_\mathbb{C} \bar S_j)} \Big) \ ,
\end{equation}
where $\nu=+1$ for $P^{NS}$ and $\nu=-1$ for $P^R$. 
Recall that the state spaces in the $NS$- and $R$-sector were defined as $\mathcal{H}^{NS/R} = \mathrm{im}(P^{NS/R}) \subset \hat{\mathcal{H}}_A$. The decomposition of $\mathcal{H}^{NS/R}$ as $\mathbb{Z}_2$-graded $\mathcal{V} \otimes_{\mathbb{C}} \bar{\mathcal{V}}$-representations is therefore described by the images of the projectors $Q_{\nu}^{i,j,\alpha}$. The next lemma will give an alternative description of these images. But first we need some notation.

Given two $A$-$A$-bimodules $M,N$ in $\hat{\mathcal{D}}$, denote by $\mathrm{Hom}_{AA}(M,N)$ the subspace of $\hat{\mathcal{D}}(M,N)$ that intertwines the left and right $A$-action.

Denote by $\lambda^M : A \otimes M \to M$ and $\rho^M : M \otimes A \to M$ the left and right $A$-action on $M$. Given an object $X \in \hat{\mathcal{D}}$, write $M \otimes^+ X$ and $M \otimes^- X$ for the $A$-$A$-bimodules with left/right action
\begin{align}
\lambda^{M \otimes^+ X} &= \lambda^M \otimes \id_X ~~,
&
\rho^{M \otimes^+ X} &= (\rho^M \otimes \id_X) \circ (\id_M \otimes c_{A,X}^{-1}) \ ,
\\
\lambda^{M \otimes^- X} &= \lambda^M \otimes \id_X~~,
&
\rho^{M \otimes^- X} &= (\rho^M \otimes \id_X) \circ (\id_M \otimes c_{X,A}) \ .
\nonumber
\end{align}
Since $K_\alpha$ is transparent, we have $M \otimes^+ K_\alpha = M \otimes^- K_\alpha$, and we will just write $M \otimes K_\alpha$.

By ${}_NA$ we mean the $A$-$A$-bimodule whose left action is twisted by $N$. Let $N_+=\id_A$ and $N_- = N$.
We will be interested in the $A$-$A$-bimodule ${}_{N_\nu}A \otimes^+ S_i \otimes^- S_j \otimes K_\alpha$. Diagramatically, the left and right action are
\begin{equation}
	\raisebox{-0.5\height}{\includegraphics[scale=0.4]{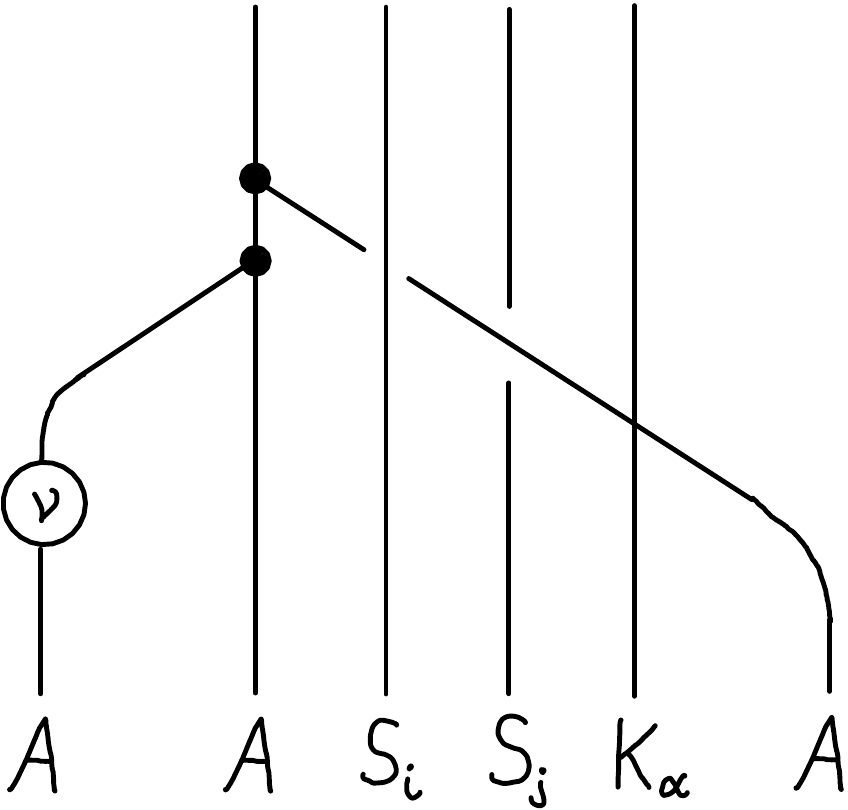}}
\end{equation}

\begin{lemma}\label{lem:imQ}
The map $\varphi : \mathrm{Hom}_{AA}({}_{N_{\nu}}A \otimes^+ S_i \otimes^- S_j \otimes K_\alpha,A) \to \hat{\mathcal{D}}(S_i \otimes S_j \otimes K_\alpha,A)$, $\varphi(h) = h \circ (\eta \otimes \id_{S_i \otimes S_j \otimes K_\alpha})$ is injective and has image $\mathrm{im}(Q_{\nu}^{i,j,\alpha})$.
\end{lemma}

\begin{proof}
That $\varphi(h) \in \mathrm{im}(Q_{\nu}^{i,j,\alpha})$ for all $h$ can be seen as follows:
\begin{align}\label{eq:imQ-calc-aux1}
Q_{\nu}^{i,j,\alpha}(\varphi(h)) ~~
&\overset{(1)}= ~~
 \raisebox{-0.5\height}{\includegraphics[scale=0.4]{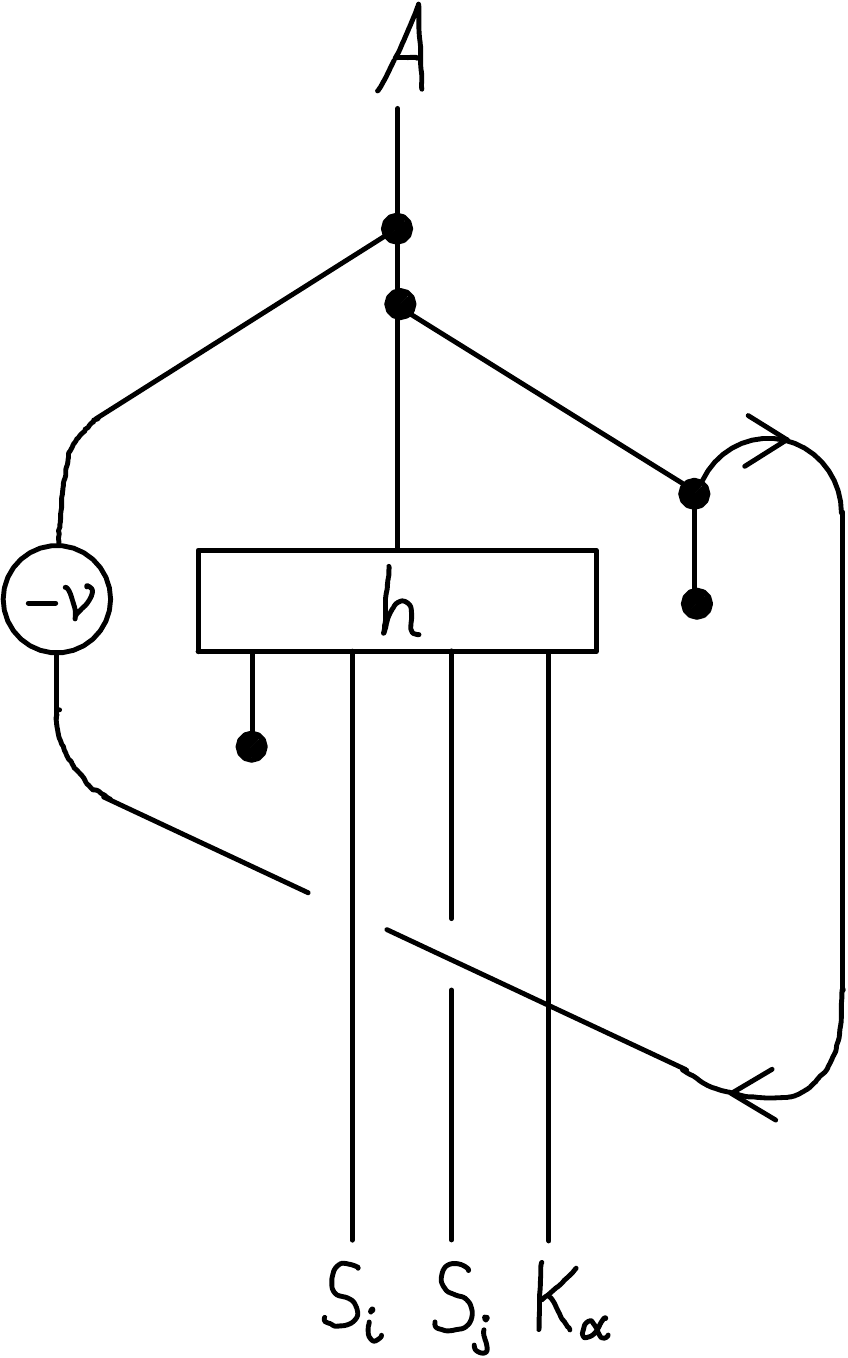}}
\quad \overset{(2)}=~~ 
 \raisebox{-0.5\height}{\includegraphics[scale=0.4]{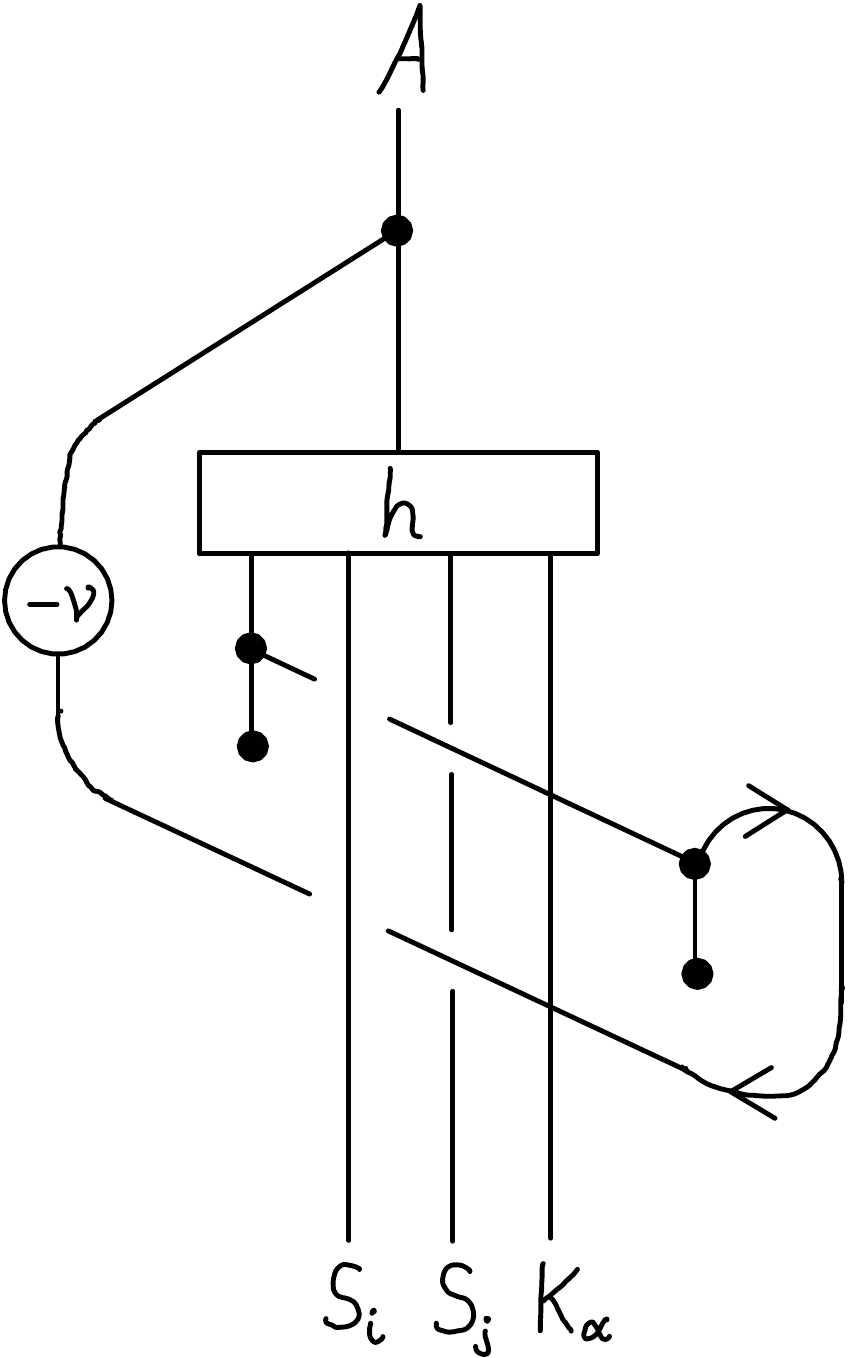}}
 \\[1em]
&\overset{(3)}= \quad
 \raisebox{-0.5\height}{\includegraphics[scale=0.4]{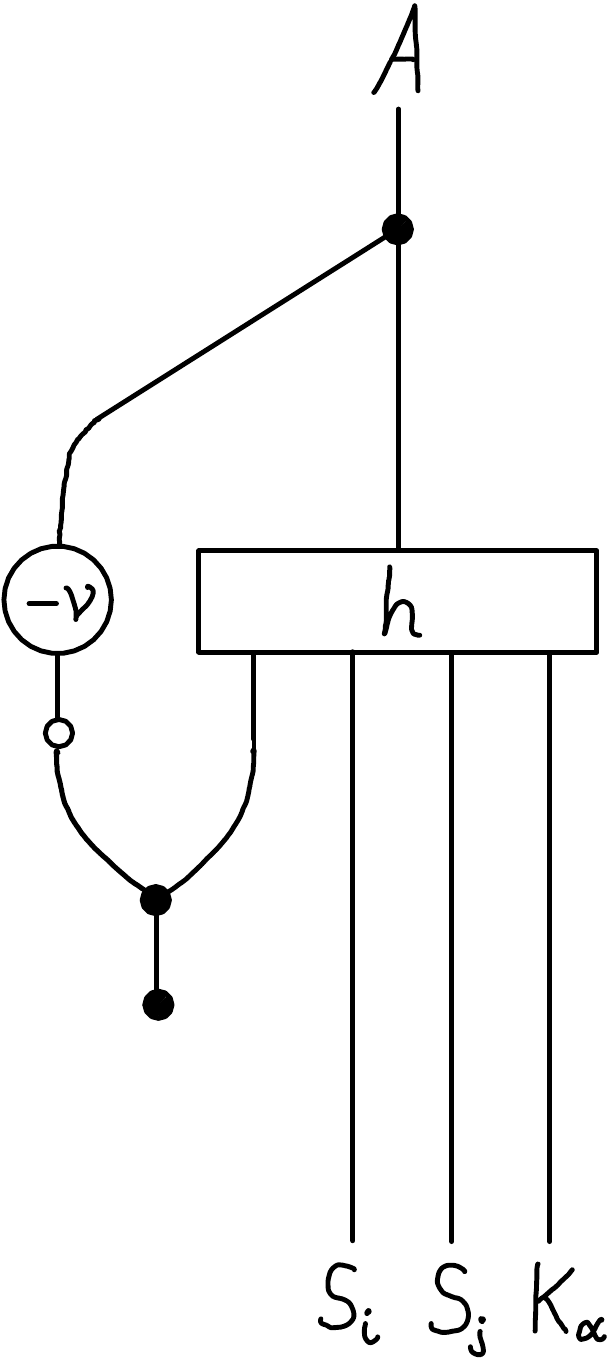}}
\quad\overset{(4)}=\quad 
 \raisebox{-0.5\height}{\includegraphics[scale=0.4]{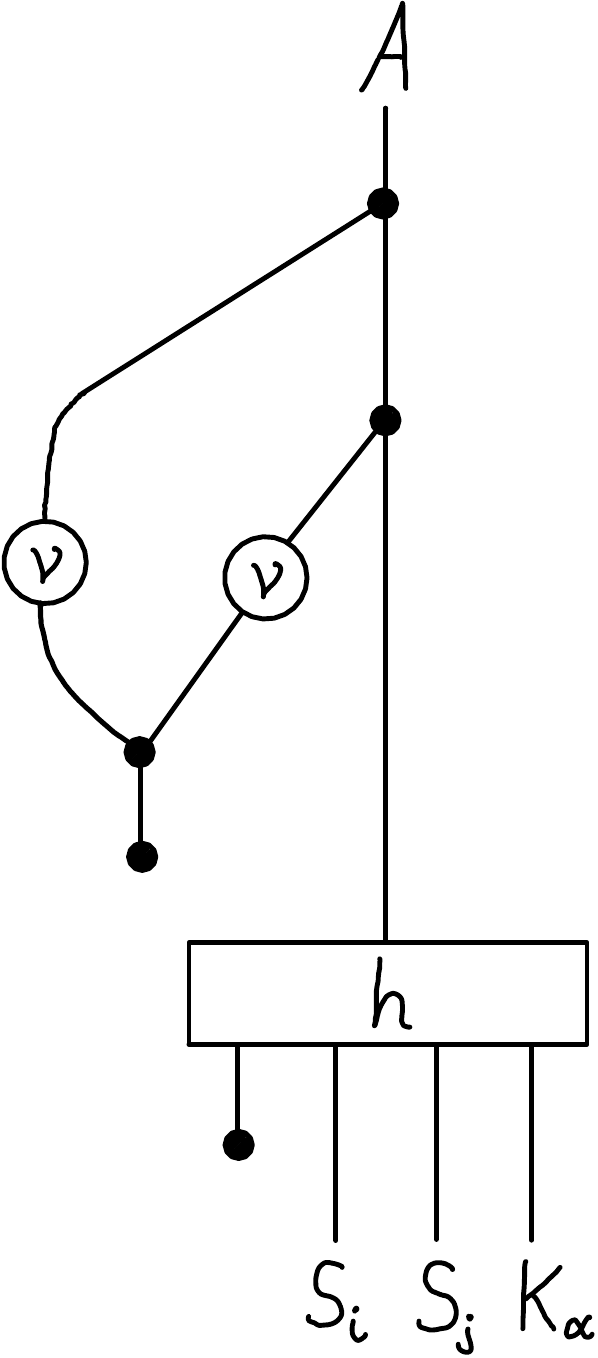}}
\overset{(5)}=
~~\varphi(h) \ .
\nonumber
\end{align}
In step 1 the definition \eqref{eq:Q-string-diag} of $Q_{\nu}^{i,j,\alpha}$ as well as that of $\varphi(h)$ are inserted, and the coproduct is traded for a product and a copairing $\Delta \circ \varepsilon$ via the Frobenius property. Step 2 is compatibility of $f$ with the right action. In step 3 first the unit property is used to remove one of the units. Then one sees that the left-most expression in \eqref{eq:inv-cond-1} (with $c_{-1}$) appears. By Theorem \ref{thm:moves_vs_alg}, this can be replaced by $c_1$. Step 4 uses that $h$ intertwines the left action on $A$ and on ${}_{N_{\nu}}A$ (and that $N_{\nu}^{\,2}=\id_A$). Step 5 follows from $\Delta$-separability of $A$.

Next, consider the map 
$\psi :\mathrm{im}(Q_{\nu}^{i,j,\alpha}) \to \mathrm{Hom}_{AA}({}_{N_{\nu}}A \otimes^+ S_i \otimes^- S_j \otimes K_\alpha,A)$ given by $\psi(f) = \mu \circ (N_{\nu} \otimes f)$. 
To check that $\psi(f)$ is indeed an $A$-$A$-bimodule morphism, one uses $f = Q_{\nu}^{i,j,\alpha}(f)$ to move the right-$A$-action to the left:
\begin{align}
\mu \circ (f \otimes \id_A)
&=
\mu \circ (Q_{\nu}^{i,j,\alpha}(f) \otimes \id_A)
=
\dots
\\
\nonumber
&= 
\mu \circ (N_{\nu} \otimes f) 
\circ 
(c_{S_i,A} \otimes \id_{S_j} \otimes \id_{K_\alpha})
\\
\nonumber
& \qquad
\circ 
(\id_{S_i} \otimes c_{A,S_j}^{-1}  \otimes \id_{K_\alpha})
\circ 
(\id_{S_i} \otimes \id_{S_j} \otimes c_{K_\alpha,A})
\end{align}
The intermediate steps abbreviated by ``\dots'' can for example be carried out in string diagram notation as we did in \eqref{eq:imQ-calc-aux1}, we omit the details.

It is immediate that $\psi$ satisfies $\varphi(\psi(f))=f$ for all $f \in \mathrm{im}(Q_{\nu}^{i,j,\alpha})$, as well as $\psi(\varphi(h))=h$ for all $h \in \mathrm{Hom}_{AA}({}_{N_{\nu}}A \otimes^+ S_i \otimes^- S_j \otimes K_\alpha,A)$.
\end{proof}

The explicit form \eqref{eq:PNSR-CFT-graded} of the projectors $P^{NS/R}$ and Lemma \ref{lem:imQ} provide the proof of the following theorem, which summarises the discussion of the state spaces in the spin theory.

\begin{theorem}\label{thm:HNSR-via-HomAA}
Let $A \in \hat{\mathcal{D}}$ be a $\Delta$-separable Frobenius algebra whose Nakayama automorphism is an involution. Let $\mathcal{H}^{NS}$ and $\mathcal{H}^{R}$ be the Neuveu-Schwarz and Ramond state spaces of the spin theory defined by $A$ as in Section \ref{sec:amp-spin-def}. The decomposition of  $\mathcal{H}^{NS/R}$ as $\mathbb{Z}_2$-graded $\mathcal{V} \otimes_{\mathbb{C}} \bar{\mathcal{V}}$-representations is given by
\begin{align}
\mathcal{H}^{NS}
&= \bigoplus_{i,j \in \mathcal{I}} \Big(
	\mathrm{Hom}_{AA}(A \otimes^+ S_i \otimes^- S_j , A) \otimes_\mathbb{C} 
		S_i \otimes_\mathbb{C} \bar S_j \\[-1em]
	&
	\hspace*{4em}
	~\oplus~
	\mathrm{Hom}_{AA}(A \otimes^+ S_i \otimes^- S_j \otimes K_- , A) \otimes_\mathbb{C} 
		\Pi(S_i \otimes_\mathbb{C} \bar S_j) 
\Big) \ ,
 \nonumber \\
\mathcal{H}^{R}
&= \bigoplus_{i,j \in \mathcal{I}} \Big(
	\mathrm{Hom}_{AA}({}_NA \otimes^+ S_i \otimes^- S_j , A) \otimes_\mathbb{C} 
		S_i \otimes_\mathbb{C} \bar S_j  \nonumber \\[-1em]
	&
	\hspace*{4em}
	~\oplus~
	\mathrm{Hom}_{AA}({}_NA \otimes^+ S_i \otimes^- S_j \otimes K_- , A) \otimes_\mathbb{C} 
		\Pi(S_i \otimes_\mathbb{C} \bar S_j) 
\Big) \ .
\nonumber
\end{align}
\end{theorem}

Variants of the isomorphism in Lemma \ref{lem:imQ} and the description of the state spaces in Theorem \ref{thm:HNSR-via-HomAA} also appear in \cite[Sect.\,5.4]{tft1} (for symmetric Frobenius algebras, so without distinguishing $NS$- and $R$-state spaces) and in \cite[Sect.\,3.2]{Brunner:2013ota} (when $Q_\mathcal{V}$ is also topological, so that $S_i=S_j=\one$).

\subsection{Example: Ising model}\label{sec:1ff-ex}

Here we will start from the two-dimensional critical Ising model -- the Virasoro minimal model conformal field theory of central charge $\tfrac12$ -- and will show how in the graded setup one can recover the theory of a free fermion.

\medskip

Let $\mathcal{V}$ be the simple Virasoro vertex operator algebra at $c=\frac12$ and let $\mathbf{Is} = \Rep\mathcal{V}$ be its representation category. Then $\mathbf{Is}$ has three simple objects, which we denote as follows (here, $h$ is the conformal weight of the ground state)
\begin{equation}
	\one \quad (h=0)
	~~ , \qquad 
	\epsilon \quad (h=\tfrac12)
	~~ , \qquad 
	\sigma \quad (h=\tfrac1{16})
	\ .
\end{equation}
The non-trivial fusion rules are $\epsilon \otimes \sigma \cong \sigma \otimes \epsilon \cong \sigma$ and $\sigma \otimes \sigma \cong \one \oplus \epsilon$.
We hope that there will be no confusion between the simple object $\epsilon$ and the counit $\varepsilon$ of the algebra $A$ we will define shortly.

A full description of the ribbon structure of $\mathbf{Is}$ can be found in many places, for example in \cite[Sect.\,4.2]{tft2}. Here we just need to know the following properties. Firstly, the braiding, twist and quantum-dimension of $\epsilon$ satisfy
\begin{equation}\label{eq:some-Is-data}
	c_{\epsilon,\epsilon} = - \id_{\epsilon \otimes \epsilon}
	\quad , \qquad
	\theta_\epsilon =  - \id_\epsilon
	\quad , \qquad
	\dim(\epsilon) = 1 \ .
\end{equation}
Secondly, we can reconnect $\epsilon$-lines as follows
\begin{equation}
\raisebox{-0.5\height}{\includegraphics[scale=0.4]{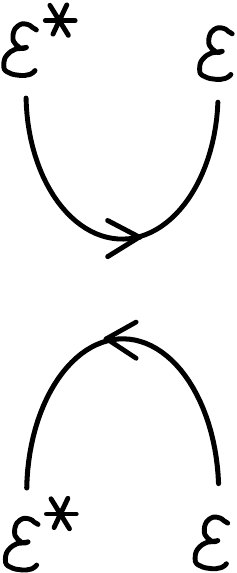}}
~=~
\raisebox{-0.5\height}{\includegraphics[scale=0.4]{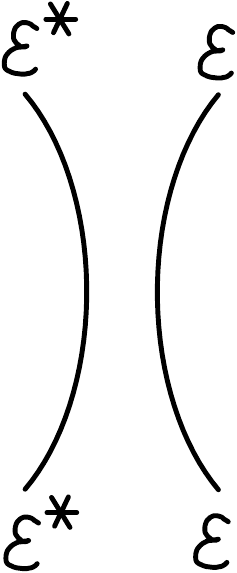}}
\quad .
\end{equation}
Finally, if $S$ is one of the simple objects, $S \in \{ \one,\epsilon,\sigma \}$, then
\begin{equation}\label{eq:eps-through-S}
	c_{\epsilon,S} = q_S \cdot c_{S,\epsilon}^{-1} \ ,
\end{equation}
where $q_{\one} = q_\epsilon = 1$ and $q_\sigma  = -1$ (the quantity $q_S$ is called the {\em monodromy charge} of $S$ with respect to $\epsilon$).

We could now look for a suitable algebra in $\mathbf{Is}$, for example $\one \oplus \epsilon$. But the unique-up-to-isomorphism $\Delta$-separable Frobenius algebra structure on  $\one \oplus \epsilon$ is symmetric, i.e.\ $N=\id$. The theory on spin surfaces constructed from a symmetric $A$ will be insensitive to the spin structure, since  $c_{+1} = c_{-1}$ in the symmetric case and so the theory does not depend on the edge signs. 

This outcome is maybe not too surprising since a holomorphic free fermion forms a vertex operator super algebra, whose even component is $\one$ and whose odd component is $\epsilon$. Hence we will now pass to the graded case.

\medskip

Define
\begin{equation}
\widehat{\mathbf{Is}} := \mathbf{Is} \boxtimes \mathbf{SVect}^{fd} \ .
\end{equation}
and in $\widehat{\mathbf{Is}}$ consider the object $A := \one \oplus \Pi \epsilon$. 
For simplicity we replace $\widehat{\mathbf{Is}}$ by an equivalent strict category.
We recall our convention for the ribbon structure on $\mathbf{SVect}^{fd}$:
\begin{align}\label{eq:some-SVect-data}
	&c_{K_-,K_-} = - \id_{K_- \otimes K_-}
~~,\quad
	\theta_{K_-} = \id_{K_-}
~~,\quad
	\dim(K_-) = -1 \ ,
\\
	&
\raisebox{-0.5\height}{\includegraphics[scale=0.4]{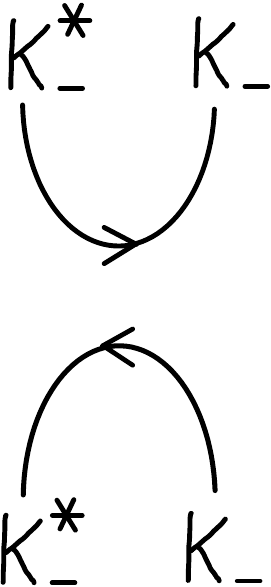}}
~=~-
\raisebox{-0.5\height}{\includegraphics[scale=0.4]{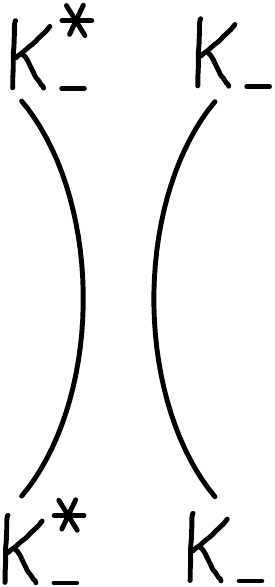}}
\nonumber
\end{align}
To give the algebra and coalgebra structure, pick a nonzero element
\begin{equation}
\lambda^{(\Pi\epsilon,\Pi\epsilon),\one} \in \widehat{\mathbf{Is}}(\Pi\epsilon \otimes \Pi\epsilon,\one) \ .
\end{equation}
Define furthermore
\begin{equation}
\lambda^{(\one,\one),\one} = \id_{\one}
~~,\quad
\lambda^{(\Pi\epsilon,\one),\Pi\epsilon} = \id_{\Pi\epsilon}
~~,\quad
\lambda^{(\one,\Pi\epsilon),\Pi\epsilon} = \id_{\Pi\epsilon}
\ ,
\end{equation}
and let us agree that all $\lambda^{(a,b)c}$ for which the corresponding morphism space in $\widehat{\mathbf{Is}}$ is zero-dimensional are set to zero. Then the multiplication on $A$ is
\begin{equation}
	\mu := \sum_{a,b,c \in \{\one,\Pi\epsilon\}}
	\lambda^{(a,b)c} \ .
\end{equation}
One verifies that this is an associative product (for example by using the explicit fusing matrices given  in \cite[Sect.\,4.2]{tft2}). The unit $\eta$ of $\mu$ is the embedding $e_{\one} : \one \to A$ of the monoidal unit. With \eqref{eq:some-Is-data} and \eqref{eq:some-SVect-data} it is straightforward to check that $A$ is commutative:
\begin{equation}
	\mu \circ c_{A,A} = \mu \ .
\end{equation}
To define the comultiplication, choose basis morphisms $\lambda^{c(a,b)} : c \to a \otimes b$ dual to the $\lambda^{(a,b)c}$ in the sense that $\lambda^{(a,b)c} \circ \lambda^{c(a,b)} = \id_c$ whenever $\lambda^{(a,b)c} \neq 0$ (otherwise set $\lambda^{c(a,b)}=0$). Then we define the comultiplication on $A$ to be
\begin{equation}
	\Delta := \tfrac12 \hspace{-1em} \sum_{a,b,c \in \{\one,\Pi\epsilon\}}
	\hspace{-1em}  \lambda^{c(a,b)} 
\end{equation}
with counit $\varepsilon = 2 p_{\one} : A \to \one$, and where $p_{\one}$ is the projection on the summand $\one$ of $A$. A short calculation shows that $A$ is $\Delta$-separable (this is the reason for the factors of 2).

The Nakayama automorphism $N$ of $A$ can be computed as follows:
\begin{equation}
	N 
~\overset{\text{deform \eqref{eq:Nakayama-def}}}= \quad
\raisebox{-0.5\height}{\includegraphics[scale=0.4]{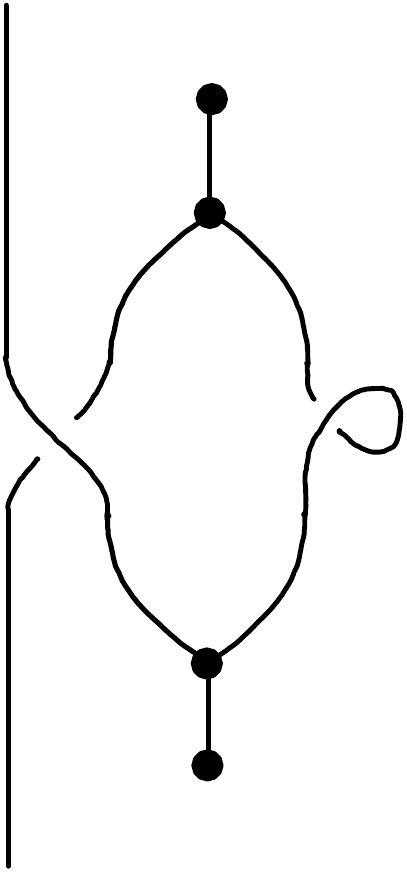}}
\quad\overset{(*)}= \quad
\raisebox{-0.5\height}{\includegraphics[scale=0.4]{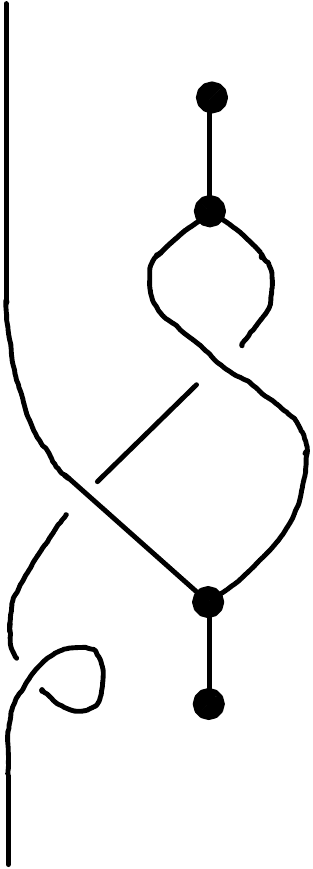}}
\quad\overset{\text{Frob.}}=~
\theta_A =\id_{\one} - \id_{\Pi\epsilon} \ ,
\end{equation}
where ($*$) is commutativity of $A$ and naturality of the twist $\theta_A$. In the last step we used  \eqref{eq:some-Is-data} and \eqref{eq:some-SVect-data} to conclude $\theta_A = \id_{\one} - \id_{\Pi\epsilon}$. In particular, $N^2=id_A$. Thus, altogether,
\begin{quote}
	$A$ is a $\Delta$-separable Frobenius algebra whose Nakayama automorphism is an involution.
\end{quote}
In fact, $A$ has another property which we will use below, namely 
\begin{equation}\label{eq:N*id=0}
	\mu \circ (N \otimes \id_A) \circ \Delta = 0 \ ,
\end{equation}
which one again checks by direct calculation. This property was investigated for 2d TFTs in \cite[Sect.\,4.9]{Novak:2014oca} 
and was found to enforce the admissibility conditions \eqref{eq:admiss-inner} and \eqref{eq:admiss-bnd} for edge signs by setting amplitudes with non-admissible configurations to zero. 

Next we compute the $NS$- and $R$-state spaces. To do so we need one more ingredient, which is that for $S \in \{\one,\epsilon,\sigma\}$ we have
\begin{equation}\label{eq:S-through-A-omegaS}
	c_{A,S} = c_{S,A}^{-1} \circ (\omega_S \otimes \id_S) \ ,
\end{equation}
where $\omega_{\one} = \id_A = \omega_\epsilon$ and $\omega_\sigma = N$. This follows directly from \eqref{eq:eps-through-S}. Using this, we compute the action of $Q_\nu^{S,S',\alpha}$ on $\widehat{\mathbf{Is}}(S \otimes S' \otimes K_\alpha , A)$ to be:
\begin{equation}
	Q_\nu^{S,S',\alpha}(h)
	~\overset{\text{\eqref{eq:S-through-A-omegaS}}}=\quad
\raisebox{-0.5\height}{\includegraphics[scale=0.4]{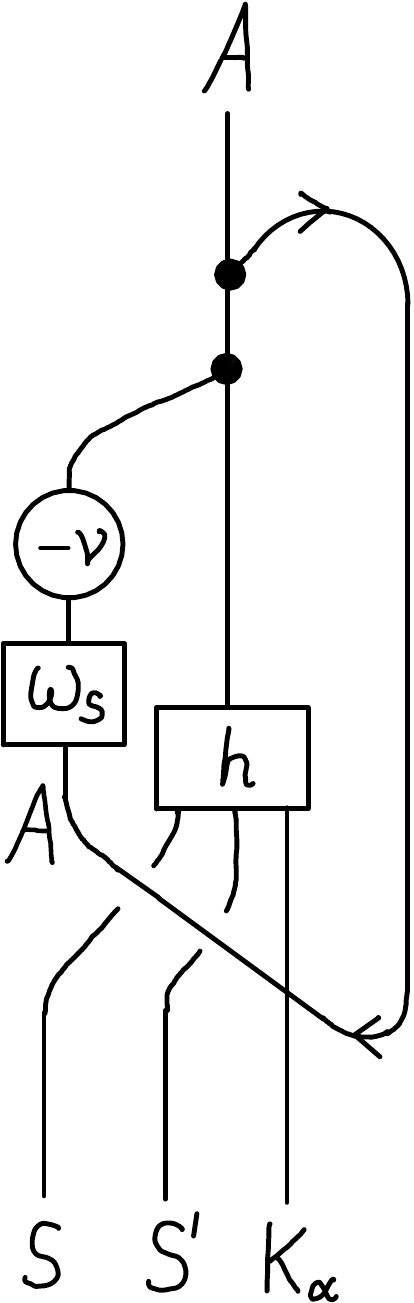}}
	\quad\overset{\text{$A$ comm.}}=\quad
\raisebox{-0.5\height}{\includegraphics[scale=0.4]{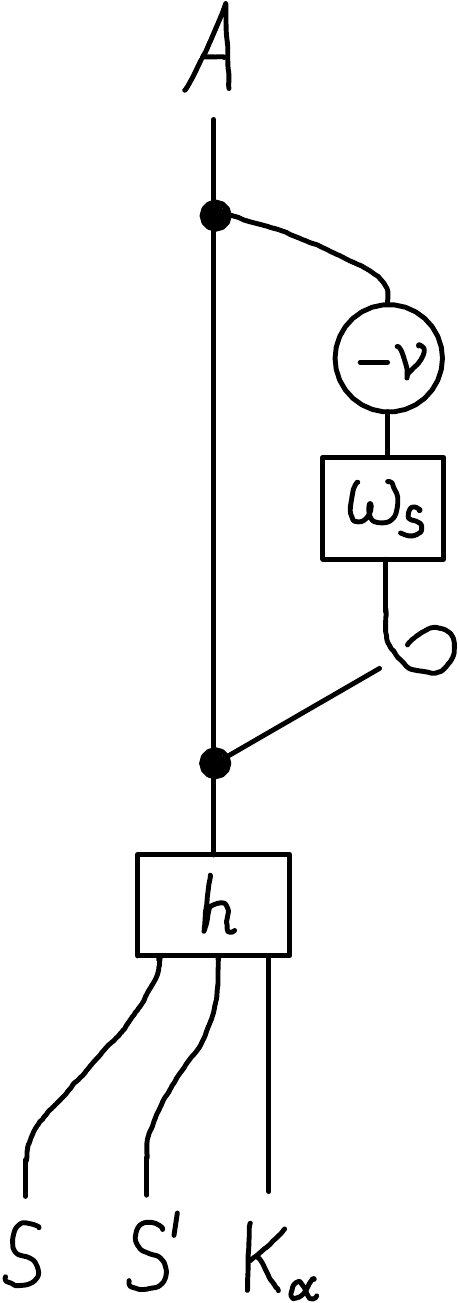}}
	\quad\overset{(*)}=~ \delta_{N_\nu,\omega_S} \cdot h \ ,
\end{equation}
where (*) uses $\theta_A = N$ and the fact that by $\Delta$-separability and by \eqref{eq:N*id=0}, $\mu \circ (\id_A \otimes (N_\nu \circ \omega_S)) \circ \Delta$ is either $\id_A$ or $0$, depending on whether $N_\nu \circ \omega_S$ equals $\id_A$ or $N$.
This result for the image of $Q$ can be plugged into \eqref{eq:PNSR-CFT-graded} to read off the state spaces.

We also observe that because of the fusion rules of $\mathbf{Is}$ and for parity reasons, $\widehat{\mathbf{Is}}(S \otimes S' \otimes K_\alpha , A)$ is at most one-dimensional. Non-zero elements of this morphism space have image in $\one \subset A$ if $\alpha=+$ and in $\Pi\epsilon \subset A$ if $\alpha=-$. 
It follows that $N|_{NS/R}$ acts as the parity involution, that is, it acts as $\id$ on $S \otimes_{\mathbb{C}} S'$ and as $-\id$ on $\Pi(S \otimes_{\mathbb{C}} S')$. 

Let us collect the results of the above discussion in two tables, whose entries give the multiplicity super vector space of the corresponding representation $S \otimes_{\mathbb{C}} S'$ in $\hat{\mathcal{H}}^{NS/R}$:
\begin{equation}\label{eq:Ising-NSR-spaces}
\hat{\mathcal{H}}^{NS} 
~:~
\renewcommand{\arraystretch}{1.3}
\begin{tabular}{c|ccc}
& $\one$ & $\epsilon$ & $\sigma$ \\
\hline
$\one$ & $\mathbb{C}^{1|0}$ & $\mathbb{C}^{0|1}$ & $0$ \\
$\epsilon$ & $\mathbb{C}^{0|1}$ & $\mathbb{C}^{1|0}$ & $0$ \\
$\sigma$ & $0$ & $0$ & $0$
\end{tabular}
\quad , \qquad
\hat{\mathcal{H}}^{R}
~:~
\renewcommand{\arraystretch}{1.3}
\begin{tabular}{c|ccc}
& $\one$ & $\epsilon$ & $\sigma$ \\
\hline
$\one$ & $0$ & $0$ & $0$ \\
$\epsilon$ & $0$ & $0$ & $0$ \\
$\sigma$ & $0$ & $0$ & $\mathbb{C}^{1|1}$
\end{tabular}
\qquad ,
\end{equation}
and $N|_{NS/R}$ acts by parity involution on the multiplicity space.

\medskip

This is the expected result for free fermions: the $NS$-sector consists of a holomorphic and an anti-holomorphic free fermion with $\mathbb{Z}_2$-grading given by fermion number. The $R$-sector contains two twist fields of weight $(\frac1{16},\frac1{16})$ due to the two-dimensional representation of the fermion zero mode algebra.
The statement of Lemma \ref{lem:dehn-twist-vs-N} in the present case is as follows. In CFT we have $J = L_0 - \bar L_0$, and so $\exp(2 \pi i (L_0 - \bar L_0))$ acts as $N|_{NS}$ on $\mathcal{H}^{NS}$ and as the identity on $\mathcal{H}^R$. This can of course also be read off directly from the conformal weights in above table.

\subsection{Example: affine so(n) at level 1}\label{sec:so(n)}

The WZW model with symmetry $\widehat{so}(n)_1$ ($n \ge 3$) has a description in terms of $n$ free fermions. This example hence extends the previous one and we will be brief. Let $\mathbf{S}_n$ be the modular tensor category of integrable highest weight representations of  $\widehat{so}(n)_1$. 
The simple objects of $\mathbf{S}_n$ together with the conformal weight of their ground states are:
\begin{center}
\parbox{.4\textwidth}{\begin{center}
$n$ even\\[.6em]
\renewcommand{\arraystretch}{1.5}
\begin{tabular}{l|cccc}
symbol 		& $1$  &  $v$ & $s$ & $c$ \\
\hline
weight $h$	& $0$ &  $\frac12$ & $\frac{n}{16}$ & $\frac{n}{16}$
\end{tabular}
\end{center}
}
\parbox{.4\textwidth}{\begin{center}
$n$ odd\\[.6em]
\renewcommand{\arraystretch}{1.5}
\begin{tabular}{l|cccc}
symbol 		& $1$  &  $v$ & $\sigma$  \\
\hline
weight $h$	& $0$ &  $\frac12$ & $\frac{n}{16}$ 
\end{tabular}
\end{center}
}
\end{center}
Here, $1$ is the vacuum representation, and $v$, $s$, $c$ stands for ``vector'', ``spinor'', and ``conjugate spinor''. The fusion rules are (we omit the ``$\otimes$'' for brevity)
\begin{center}
\renewcommand{\arraystretch}{1.5}
\begin{tabular}{l|l}
$n$ odd
&
$v  v \cong 1
~,~~ 
v  \sigma \cong \sigma 
~,~~ 
\sigma  \sigma \cong 1 \oplus v$
\\
$n \equiv 0 \mod 4$
&
$vv \cong ss \cong cc \cong 1
~,~~ 
s  c \cong v$
\\
$n \equiv 2 \mod 4$
&
$vv \cong sc \cong 1
~,~~ 
ss \cong cc \cong v$
\end{tabular}
\end{center}
The monodromy charges are given by $c_{v,S} \, c_{S,v} = q_S \, \id_{S \otimes V}$, where $q_S = \theta_{v\otimes S} / (\theta_{v}\theta_S)$. Explicitly,
\begin{center}
\renewcommand{\arraystretch}{1.5}
\begin{tabular}{l|ccccc}
$S$ & 
$1$ &
$v$ &
$s$ &
$c$ &
$\sigma$
\\
\hline
$q_S$ &
$1$ &
$1$ &
$-1$ &
$-1$ &
$-1$
\end{tabular}
\quad .
\end{center}
As in the Ising case, we consider $\hat{\mathbf{S}}_n := \mathbf{S}_n \boxtimes \mathbf{SVect}^{fd}$ and therein the algebra $A = 1 \oplus \Pi v$. Again, the $NS$-sector is build from representations with $q_S=1$ and the $R$-sector from representations with $q_S=-1$.
For the $NS$-state space one finds:
\begin{align}
&\qquad n \text{ even}
&&\qquad
n \text{ odd} \nonumber\\
\hat{\mathcal{H}}^{NS} 
~:~ \qquad &
\renewcommand{\arraystretch}{1.3}
\begin{tabular}{c|cccc}
& $\one$ & $v$  & $s$  & $c$  \\
\hline
$\one$ & $\mathbb{C}^{1|0}$ & $\mathbb{C}^{0|1}$ & $0$ & $0$ \\
$v$ & $\mathbb{C}^{0|1}$ & $\mathbb{C}^{1|0}$  & $0$ & $0$ \\
$s$  & $0$ & $0$ & $0$ & $0$ \\
$c$  & $0$ & $0$ & $0$ & $0$ 
\end{tabular}
&&
\renewcommand{\arraystretch}{1.3}
\begin{tabular}{c|ccc}
& $\one$ & $v$  & $\sigma$ \\
\hline
$\one$ & $\mathbb{C}^{1|0}$ & $\mathbb{C}^{0|1}$ & $0$ \\
$v$ & $\mathbb{C}^{0|1}$ & $\mathbb{C}^{1|0}$  & $0$ \\
$s$  & $0$ & $0$ & $0$
\end{tabular}
\quad . 
\end{align}
For the $R$-state space one has to distinguish three cases:
\begin{align}
& n \equiv 0 \!\!\mod 4 
&&
n \equiv 2 \!\!\mod 4 \nonumber\\
\hat{\mathcal{H}}^{R}
~:~ \qquad &
\renewcommand{\arraystretch}{1.3}
\begin{tabular}{c|cccc}
& $\one$ & $v$  & $s$  & $c$  \\
\hline
$\one$  & $0$ & $0$ & $0$ & $0$  \\
$v$  & $0$ & $0$ & $0$ & $0$ \\
$s$ & $0$ & $0$ & $\mathbb{C}^{1|0}$ & $\mathbb{C}^{0|1}$ \\
$c$ & $0$ & $0$ & $\mathbb{C}^{0|1}$ & $\mathbb{C}^{1|0}$ 
\end{tabular}
&&
\renewcommand{\arraystretch}{1.3}
\begin{tabular}{c|cccc}
& $\one$ & $v$  & $s$  & $c$  \\
\hline
$\one$  & $0$ & $0$ & $0$ & $0$  \\
$v$  & $0$ & $0$ & $0$ & $0$ \\
$s$ & $0$ & $0$ & $\mathbb{C}^{0|1}$ & $\mathbb{C}^{1|0}$ \\
$c$ & $0$ & $0$ & $\mathbb{C}^{1|0}$ & $\mathbb{C}^{0|1}$ 
\end{tabular}
\\[1em]
\nonumber
&\qquad n \text{ odd}\\
\nonumber
& \begin{tabular}{c|ccc}
& $\one$ & $v$  & $\sigma$ \\
\hline
$\one$ & $0$ & $0$ & $0$ \\
$v$ & $0$ & $0$  & $0$ \\
$s$  & $0$ & $0$ & $\mathbb{C}^{1|1}$
\end{tabular}
\quad .
\end{align}


\newcommand\arxiv[2]      {\href{http://arXiv.org/abs/#1}{#2}}
\newcommand\doi[2]        {\href{http://dx.doi.org/#1}{#2}}
\newcommand\httpurl[2]    {\href{http://#1}{#2}}

\providecommand{\href}[2]{#2}\begingroup\raggedright

\end{document}